\documentclass[english,12pt, draftcls, onecolumn]{IEEEtran}
\usepackage[T1]{fontenc}
\usepackage[latin9]{inputenc}
\usepackage{color}
\usepackage{float}
\usepackage{amsmath}
\usepackage{amsthm}
\usepackage{amssymb}
\usepackage{stackrel}
\usepackage{graphicx}
\usepackage{setspace}
\usepackage{nomencl}

\providecommand{\printnomenclature}{\printglossary}
\providecommand{\makenomenclature}{\makeglossary}
\makenomenclature

\makeatletter

\providecommand{\tabularnewline}{\\}
\floatstyle{ruled}
\newfloat{algorithm}{tbp}{loa}
\providecommand{\algorithmname}{Algorithm}
\floatname{algorithm}{\protect\algorithmname}

\numberwithin{equation}{section}
\numberwithin{figure}{section}
\theoremstyle{plain}
\newtheorem{thm}{\protect\theoremname}
\theoremstyle{definition}
\newtheorem{defn}[thm]{\protect\definitionname}
\theoremstyle{definition}
\newtheorem{problem}[thm]{\protect\problemname}
\theoremstyle{remark}
\newtheorem{rem}[thm]{\protect\remarkname}
\theoremstyle{plain}
\newtheorem{lem}[thm]{\protect\lemmaname}
\theoremstyle{definition}
\newtheorem{example}[thm]{\protect\examplename}
\theoremstyle{plain}
\newtheorem{prop}[thm]{\protect\propositionname}

\usepackage{caption}
\usepackage{graphics}
\usepackage{epstopdf}

\newtheorem{assumption}[thm]{Assumption}
\newtheorem{note}[thm]{Note}

\usepackage{cite}

\newcommand{\blfootnote}[1]{%
  \begingroup
  \renewcommand\thefootnote{}\footnote{#1}%
  \addtocounter{footnote}{-1}%
  \endgroup
}

\makeatother

\usepackage{babel}

\providecommand{\definitionname}{Definition}
\providecommand{\examplename}{Example}
\providecommand{\lemmaname}{Lemma}
\providecommand{\problemname}{Problem}
\providecommand{\propositionname}{Proposition}
\providecommand{\remarkname}{Remark}
\providecommand{\theoremname}{Theorem}

\begin{document}
\title{A time-optimal feedback control for a particular case of the game
of two cars }
\author{Aditya Chaudhari and Debraj Chakraborty}
\maketitle
\begin{abstract}
\noindent In this paper, a computationally efficient time-optimal
feedback solution to the game of two cars, for the case where the
pursuer is faster and more agile than the evader, is presented. The
concept of continuous subsets of the reachable set is introduced to
characterize the time-optimal pursuit-evasion game under feedback
strategies. Using these subsets it is shown that, if initially the
pursuer is distant enough from the evader, then the feedback saddle
point strategies for both the pursuer and the evader are coincident
with one of the common tangents from the minimum radius turning circles
of the pursuer to the minimum radius turning circles of the evader.
Using geometry, four feasible tangents are identified and the feedback
$\min-\max$ strategy for the pursuer and the $\max-\min$ strategy
for the evader are derived by solving a $2\times2$ matrix game at
each instant. Insignificant computational effort is involved in evaluating
the pursuer and evader inputs using the proposed feedback control
law and hence it is suitable for real-time implementation. 
\end{abstract}

\begin{IEEEkeywords}
Dubins vehicle, game of two cars, time-optimal feedback policy, reachable
set
\end{IEEEkeywords}

\blfootnote{The authors are with the Department of Electrical Engineering,
Indian Institute of Technology Bombay, India. \tt{Email: {adityac,dc}@ee.iitb.ac.in}

} \vspace{-0.5cm}

\begin{spacing}{0.2}
\noindent \settowidth{\nomlabelwidth}{$T_{c}({\bf \gamma_{p}}({\bf x}),{\bf \gamma_{e}}({\bf x}))$}
\printnomenclature{}\nomenclature[01]{${\bf u_{p}}:={\bf \gamma_{p}}({\bf x})$}{Feedback input policy of pursuer}\nomenclature[02]{${\bf u_{e}}:={\bf \gamma_{e}}({\bf x})$}{Feedback input policy of evader}\nomenclature[03]{$T_{c}({\bf \gamma_{p}}({\bf x}),{\bf \gamma_{e}}({\bf x}))$}{Time to capture under inputs ${\bf \gamma_{p}}({\bf x})$, ${\bf \gamma_{e}}({\bf x})$}\nomenclature[04]{$T^{*}$}{Min-max time to capture}\nomenclature[05]{$ C_{p}(\bar{t})$}{Clockwise pursuer circle}\nomenclature[06]{$ A_{p}(\bar{t})$}{Anti-clockwise pursuer circle}\nomenclature[07]{$R_{e}({\bf e_{0}},\bar{t})$}{Reachable set of the evader}\nomenclature[07.1]{$R_{p}({\bf p_{0}},\bar{t})$}{Reachable set of the pursuer}\nomenclature[08]{$\partial R_{p}({\bf p_{0}},\bar{t})$}{External boundary of pursuer's reachable set}\nomenclature[09]{$\partial R_{e}({\bf e_{0}},\bar{t})$}{External boundary of evader's reachable set}\nomenclature[10]{$R_{p}^{l}({\bf p_{0}},\bar{t})$}{Left reachable set of the pursuer}\nomenclature[101]{$R_{p}^{l_t}({\bf p_{0}},\bar{t})$}{Truncated left reachable set of the pursuer}\nomenclature[11]{$R_{p}^{r}({\bf p_{0}},\bar{t})$}{Right reachable set of the pursuer}\nomenclature[111]{$R_{p}^{r_t}({\bf p_{0}},\bar{t})$}{Truncated right reachable set of the pursuer}\nomenclature[12]{$\partial R_{p}^{r}({\bf p_{0}},\bar{t})$}{Boundary of pursuer's right reachable set}\nomenclature[13]{$\partial R_{p}^{l}({\bf p_{0}},\bar{t})$}{Boundary of pursuer's left reachable set}\nomenclature[16]{$R_{e}^{-}({\bf e_{0}},R_{p}^{s},\bar{t})$}{Safe region of evader\\}\nomenclature[17]{$R_{p}^{c}({\bf p_{0}},\bar{t})$}{Continuous subset of pursuer's reachable set\\\\}\nomenclature[18]{$ R_{p}^{CR}({\bf p_{0}},\bar{t}) $}{Central reachable set of the pursuer\\\\}\nomenclature[19]{$ R_{e}^{c-}({\bf e_{0}},R_{p}^{c},\bar{t}) $}{Continuous safe region  of a CSE $R_{e}^{c}({\bf e_{0}},\bar{t})$  w.r.t. a CSP $R_{p}^{c}({\bf p_{0}},\bar{t})$\\\\}\nomenclature[20]{$\mathcal{L}_{e}({\bf e_{0}},T)$}{LCSE\\}\nomenclature[21]{$\mathcal{C}_{p}({\bf p_{0}},T)$}{CCSP\\}\nomenclature[22]{$ B_{p}({\bf p_{0}},{\bf e_{0}},\bar{t}) $}{Blocking set of the pursuer\\\\}\nomenclature[23]{$ \mathcal{A}_{p}({\bf p_{0}},T_{a})$}{Active set of the pursuer\\\\} 
\end{spacing}

\section{Introduction}

Pursuit evasion is a differential game between two antagonistic agents
called the pursuer and the evader. The pursuer aims to capture the
evader in minimum time whereas evader aims to avoid capture for as
long as possible. The problem of time optimal pursuit evasion for
two Dubins vehicles or the game of two cars was first introduced by
Issacs in \cite{rufus}. A Dubins vehicle consists of a point moving
in a plane with a given maximum forward velocity and a minimum turning
radius. In \cite{rufus} the pursuer is considered superior to the
evader and the regions of capture are characterized by backward integration
of Issacs equation for the game in reduced state-space. This results
in the so called retrogressive path equations. The optimal control
law synthesis involves solving these non-linear algebraic equations
numerically, thereby necessitating significant computation if such
laws are to be implemented as instantaneous feedback. In this paper,
we propose an alternative and novel geometric technique for efficiently
computing the time optimal feedback control for this game.

Using the Issacs equation, capture regions have been characterized
for different variations of the game of two cars. The problem has
been studied in detail in \cite{Merz1972} for the case where the
evader and the pursuer have equal speed and the pursuer is more agile
than the evader. Also, various regions from which capture is possible
are characterized. An asymmetrical version of the game of two cars
is discussed in \cite{Exar-2014}. In \cite{Bera2017}, the retrogressive
path equations are derived for all possible cases i.e. the pursuer
being superior than the evader, the pursuer having angular velocity
greater than the evader, and the pursuer having linear velocity greater
than the evader. The homicidal chauffeur problem is another variation
of the game of two cars in which the evader can turn instantaneously
\cite{rufus}. Variations of the homicidal chauffeur problem have
been studied in \cite{Exarchos2015}, and the regions of capture have
been characterized using the Issacs equation. In \cite{ruiz2016},
the retrogressive path equations have been used to characterize switching
surfaces in terms of state variables, for a variation of the homicidal
chauffeur game described in \cite{ruiz2013}. Such a representation
makes it possible to implement the optimal control laws as feedback.\textcolor{blue}{{}
}However, to the best of our knowledge, no such time-optimal computationally
efficient \textit{feedback} law exists for the game of two cars. 

Another approach taken to solve pursuit-eavsion games is that of reachable
sets. At a given time, the reachable set consists of points which
can be reached by an agent using admissible inputs. We use the concept
of reachable set extensively in this paper. Reachable sets have been
used for analyzing differential games since \cite{cockayne1967pursuit,mizukami1977geometrical}.
The reachable sets of the Dubins vehicle are characterized in \cite{cockayne1975plane,bui1994accessibility}.
In order to characterize the reachable sets analytically, it is necessary
to find the time-optimal trajectories of the agent. The optimal paths
for a single Dubins vehicle have been studied extensively since \cite{dubins1957curves}.
Optimal control theory is used in conjunction with geometric techniques
in \cite{sussmann1991shortest,boissonnat1994shortest,soueres1996shortest}
to characterize the curves followed by the Dubins vehicle to reach
from a given initial configuration to a final configuration in minimum
time. Time optimal feedback laws have been derived for path tracking
by Dubins vehicle in \cite{soueres2001}. Recently reachable sets
have been used to derive feedback strategies under varying flow fields
for various types of pursuers and evaders \cite{sun2015pursuit,tsiotras2017reachabilitysets}.
For a specific type of agent it was shown that the containment of
evaders reachable set in the reachable set of the pursuer characterizes
capture. However, we show that such a characterization does not hold
when the agents are Dubins vehicles. Instead we provide a novel characterization
in terms of continuous subsets of reachable set, which we introduce
in this paper.

The feedback solution to the game of two cars, studied in the paper,
is necessary for the applications which involves capture of an uncertain
evader by a pursuer modeled as Dubins vehicles. Examples include tail
chase \cite{merz77}, and aerial dog fighting \cite{rufus}, anti-aircraft
missiles \cite{missile_dubins_1993}, etc. In each of these applications,
the feedback algorithm for the pursuer or the evader needs to be implemented
in real time. Such implementations on real systems by solving non-linear
retrogressive path equations is often computationally infeasible \cite{merz77}.
On the other hand the method presented here is easily implementable
with minimal real time computation for each of these applications.
Synthesis of feedback laws for pursuit evasion of multiple Dubins
vehicle has been attempted in \cite{tomlin2005}. However, the feedback
laws are not time-optimal. Thus, the theory developed in this paper,
has possible applications for a wide range of time-optimal multi-player
games \cite{tomlin2005} involving Dubins vehicle.

In this paper we consider the problem of the game of two cars when
the pursuer is superior than the evader. We do not impose any restriction
on the final orientation of the pursuer. In this case, capture by
the pursuer is always guaranteed for all possible configurations of
pursuer and evader \cite{cockayne1967pursuit}. Our aim is to derive
a feedback law which involves only evaluation of and comparison with
closed form algebraic expressions and can be computed in real time
along the trajectories, effectively providing an implementable feedback
solution. We derive such an law by first characterizing the nature
of saddle point trajectories. This characterization is done by analyzing
the relation between feedback $\min-\max$ strategies and reachable
sets. We establish a necessary and sufficient condition for saddle
point capture in terms of some special subsets of the reachable set,
which we call continuous subsets of the reachable set. Using these
continuous subsets we characterize the time and point of capture under
feedback strategies. From this characterization we conclude that,
if the distance between the pursuer and the evader is large compared
to their turning radius, then the saddle point strategies consist
of a circle followed by straight line. Further, using Pontryagin's
minimum principle, we show that the trajectories are common tangents
to the minimum turning radius circles of the pursuer and the evader.
Since both the vehicles are restricted to have a minimum turning radius,
the pursuer and the evader each will have one clockwise minimum radius
turning circle and one anti-clockwise minimum radius turning circle
at each time instant. This gives us sixteen common tangents between
the pursuer and evader circle pairs. Using geometrical arguments,
we are able to reduce the number of feasible tangents to four, one
each for every pair of circles between the evader and the pursuer.
A $2\times2$ matrix game is formulated and the min-max solution of
the matrix game gives the strategy for the pursuer while the max-min
solution gives the strategy for the evader. The matrix game is solved
at each instant of time to obtain the feedback strategies for the
differential game. In summary, our contributions are as follows:
\begin{enumerate}
\item We introduce the concept of continuous subsets of reachable sets in
order to completely characterize capture under feedback trajectories. 
\item If the distance between pursuer and evader is greater than a certain
distance, then we show that the saddle point pursuit-evasion trajectories
are coincident with a common tangent from minimum radius turning circles
of pursuer to minimum radius turning circles of the evader. 
\item Using these novel results, we design a computationally efficient feedback
law which can be implemented in real time.
\end{enumerate}
The paper is structured as follows. The problem statement and preliminaries
are described in Section \ref{sec:problem_formulation} and Section
\ref{sec:Preliminaries} respectively. A couple of counter examples
that show that only containment by reachable sets does not characterize
capture for the game of two cars, is presented in Section \ref{sec:analysis_reachability}.
The novel continuous subsets which characterize capture under feedback
strategies, are introduced in Section \ref{sec:continuous_subsets}.
The main theorems and results have been presented in Section \ref{sec:Main-Results}.
The subsequent sections (Section \ref{sec:continuity_proofs}-\ref{sec:impt_theorem_3})
contain proofs of results presented in Section \ref{sec:Main-Results}. 

\section{\label{sec:problem_formulation}Problem Formulation }

Consider a pursuer $P$ and an evader $E$ following the equations:

\begin{eqnarray}
 & \dot{x}_{i}(t)=v_{i}(t)\cos\theta_{i}(t);\label{eq:pursuer_evader_dyn}\\
 & \dot{y}_{i}(t)=v_{i}(t)\sin\theta_{i}(t);\\
 & \dot{\theta}_{i}(t)=v_{i}(t)w_{i}(t)
\end{eqnarray}
where $i\in\{p,e\}$. The subscript $p$ corresponds to the pursuer
while $e$ corresponds to the evader. The pursuer (evader) can control
its velocity $v_{i}(t)$ in direction $\theta_{i}(t)$ and the angular
velocity $w_{i}(t)$. We denote by $\mathcal{C}(\mathbb{R^{+}},\mathbb{R}^{n})$
the set of continuous functions from positive real line $\mathbb{R}^{+}$
to $\mathbb{R}^{n}$. Let ${\bf z_{i}}(t)=[x_{i}(t)\;y_{i}(t)]^{\top}\in\mathbb{R}^{2},\mathbf{z_{i}}\in\mathcal{C}(\mathbb{R^{+}},\mathbb{R}^{2})$
denote the position of the pursuer (evader) in the $x-y$ plane at
time $t$. Also, let $\theta_{i}(t)$ be the orientation of the pursuer
(evader) in the $x-y$ plane, measured in anti-clockwise direction
with respect to the $x-\text{axis}$ at time $t$. The complete state
vector of the pursuer at time $t$ is given by ${\bf p}(t)=[x_{p}(t)\,y_{p}(t)\,\theta_{p}(t)]\in\mathbb{R}^{3},\mathbf{p}\in\mathcal{C}(\mathbb{R^{+}},\mathbb{R}^{3})$
while that of the evader is given by ${\bf e}(t)=[x_{e}(t)\,y_{e}(t)\,\theta_{e}(t)]^{\top}\in\mathbb{R}^{3},\mathbf{e}\in\mathcal{C}(\mathbb{R^{+}},\mathbb{R}^{3})$.
We also denote the projection of the pursuer's (evader's) trajectory
to the $x-y$ plane corresponding to the trajectory ${\bf p}$ (${\bf e}$)
by ${\bf z}_{i}={\bf p}_{|\mathbb{R}^{2}}\ ({\bf e}_{|\mathbb{R}^{2}})$.
Let the initial state of the pursuer be denoted by ${\bf p}(0)={\bf p_{0}}=[x_{p_{0}}\,y_{p_{0}}\,\theta_{p_{0}}]^{\top}$
and that of the evader by ${\bf e}(0)={\bf e_{0}}=[x_{e_{0}}\,y_{e_{0}}\,\theta_{e_{0}}]^{\top}$.
Also, let $d_{pe}(t)=||{\bf z_{p}}(t)-{\bf z_{e}}(t)||_{2}$ be the
distance between pursuer and evader at time instant $t$ and $d_{pe}(0):=d_{pe}^{0}$.
We denote the input of the pursuer (evader) at time $t$ by ${\bf u_{i}}(t)=[v_{i}(t)\;w_{i}(t)]^{\top}\in\mathbb{R}^{2}$
where $i\in\{p,e\}$. Also, $v_{i}(t)\in V_{i}$ where $V_{i}=\{v_{i}(t)\in\mathbb{R}:0\leq v_{i}(t)\leq v_{i_{m}}\}$
and $w_{i}(t)\in W_{i}$ where $W_{i}=\{w_{i}(t)\in\mathbb{R}:|w_{i}(t)|\leq w_{i_{m}}\}$
for $i\in\{p,e\}$. These restrictions limit the maximum forward velocity
with which the pursuer and evader can move and also limits the rate
at which the vehicles can change direction. We define the set of feasible
inputs for the pursuer and the evader as ${\bf U_{i}}:=\{{\bf u}_{i}(t):v_{i}(t)\in V_{i}\,\text{and}\,w_{i}(t)\in W_{i}\}$
where $i\in\{p,e\}$. If the input of the pursuer (evader) ${\bf u}_{i}(t)\in{\bf U_{i}}$
for all $t$ then we write ${\bf u}_{i}\in\mathcal{U}_{i}$. If the
pursuer (evader) sets $w_{i}(t)=w_{i_{m}}$, then it moves along a
circle of radius $r_{i}=1/w_{i_{m}}$ in anti-clockwise direction
while if it applies an input of $w_{i}(t)=-w_{i_{m}}$ then it moves
in clockwise direction along the circle of radius $r_{i}$. 

In order to guarantee capture of evader by the pursuer, we impose
the following restriction on the evader's input. 

\begin{assumption}\label{ass:sup_purs}Maximum velocities of the
pursuer and the evader satisfy $v_{p_{m}}>v_{e_{m}}$, while the maximum
turning rates are such that $w_{p_{m}}>w_{e_{m}}$. \end{assumption}

In this paper, the pursuer's objective is to intercept the evader
in minimum possible time, while that of the evader is to avoid interception
by the pursuer for as long as possible. The complete state of the
game is described by ${\bf x}(t)=[{\bf p}(t)^{\top}\,{\bf e}(t)^{\top}]^{\top}\in\mathbb{R}^{6},\mathbf{x}\in\mathcal{C}(\mathbb{R}^{+},\mathbb{R}^{6})$.
For capture we require that only the $x$ and $y$ coordinates of
both the pursuer and evader must match. We do not impose the restriction
that the final orientation of the pursuer and the evader must be the
same. Hence, the condition at the time of capture $T_{c}$ is
\begin{equation}
\psi({\bf x}(T_{c})):=\left[x_{p}(t)-x_{e}(t),\ \ \ y_{p}(t)-y_{e}(t)\right]^{\top}|_{t=T_{c}}=0\label{eq:end_point_constatint}
\end{equation}
Thus the time of capture $T_{c}$ at which the game terminates i.e.
the cost function in the game of two cars, is defined by $T_{c}=\inf\{t\in\mathbb{R}^{+}:\psi({\bf x}(t))=0\}$.
The pursuer tries to minimize $T_{c}$ while the evader tries to maximize
it using feedback strategies ${\bf u_{p}}:={\bf \gamma_{p}}({\bf x})\in{\bf \mathcal{U}}_{p}$
and ${\bf u_{e}}:={\bf \gamma_{e}}({\bf x})\in{\bf \mathcal{U}}_{e}$.
The time to capture is a function of feedback strategy pair $\left({\bf \gamma_{p}}({\bf x}),{\bf \gamma_{e}}({\bf x})\right)$
and we denote it by $T_{c}({\bf \gamma_{p}}({\bf x}),{\bf \gamma_{e}}({\bf x}))$.
The pursuer must guard against the worst-case strategies of the evader.
Hence, the minimum time capture problem for the pursuer is a $\min-\max$
time-optimal problem. So the pursuer's aim is to find a feedback control
strategy ${\bf u_{p}}={\bf \gamma_{p}^{*}}({\bf x})$ which solves
${\bf \gamma_{p}^{*}}({\bf x})=\underset{{\bf \gamma_{p}}}{\text{argmin}}\,\underset{{\bf \gamma_{e}}}{\text{max}}\,T_{c}({\bf \gamma_{p}}({\bf x}),{\bf \gamma_{e}}({\bf x}))$.
Similarly, the evader must guard against every possible strategy of
the pursuer. Thus, the maximum time evasion problem for the evader
is to find the $\max-\min$ strategy of the evader. Thus, the evader
aims to find a feedback control strategy ${\bf u_{e}}(t)={\bf \gamma_{e}^{*}}({\bf x})$
which solves ${\bf \gamma_{e}^{*}}({\bf x})=\underset{{\bf \gamma_{e}}}{\text{argmax}}\,\underset{{\bf \gamma_{p}}}{\text{min}}\,T_{c}({\bf \gamma_{p}}({\bf x}),{\bf \gamma_{e}}({\bf x}))$.
Since the objectives of the pursuer and the evader are conflicting
the optimality is characterized in terms of saddle-point strategies
\cite{rufus,bacsar1998dynamic,bryson1975applied}. 
\begin{defn}
\label{def:saddle_point}A \textit{feedback strategy pair} $({\bf \gamma_{p}^{*}},{\bf \gamma_{e}^{*}})$
is a saddle-point equilibrium if
\begin{equation}
T_{c}({\bf \gamma_{p}^{*}},{\bf \gamma_{e}})\leq T_{c}({\bf \gamma_{p}^{*}},{\bf \gamma_{e}^{*}})\leq T_{c}({\bf \gamma_{p}},{\bf \gamma_{e}^{*}})\label{eq:saddle_point}
\end{equation}
$\forall\ \gamma_{e}({\bf x})\in{\bf U_{e}},\gamma_{p}({\bf x})\in{\bf U_{p}}$
and the value of the game, if it exists, is $T^{*}:=T_{c}({\bf \gamma_{p}^{*}},{\bf \gamma_{e}^{*}})$.
\end{defn}
In the game of two cars considered above, under Assumption \ref{ass:sup_purs},
the following result holds.
\begin{thm}
\label{thm:guaranteedcapture}\cite{cockayne1967pursuit} If Assumption
\ref{ass:sup_purs} holds, there exists a pursuer input ${\bf u_{p}}:={\bf \gamma}_{p}({\bf x})$
such that for all ${\bf u_{e}}:={\bf \gamma}_{e}({\bf x})$, capture
is guaranteed i.e. $\psi({\bf x}(T_{c}))=0$ for some $T_{c}({\bf \gamma_{p}}({\bf x}),{\bf \gamma_{e}}({\bf x}))<\infty$.
\end{thm}
Since the capture is guaranteed and the Hamiltonian (considered later
in Section \ref{eq:hamiltonian}) is separable in pursuer's input
and evader's input, the existence of saddle point strategies $\left({\bf \gamma_{p}^{*}},{\bf \gamma_{e}^{*}}\right)$
follows \cite{bryson1975applied}. 

\vspace*{0.25cm}

\noindent \textbf{Main problem statement}

\noindent In this paper, we aim to develop time-optimal feedback control
strategies which can be implemented in real time often on a relatively
less powerful onboard computer, and hence requiring tractable computation.
This is possible if we are able to determine the pursuer/evader input
based on a fixed and small number of algebraic evaluations. This problem
is stated below:
\begin{problem}
\label{prob:implementable_feedback_law}Given a time instant $t$
and pursuer and evader states, ${\bf x}(t)=[{\bf p}(t)^{\top}\,{\bf e}(t)^{\top}]^{\top}$,
find an feedback control algorithm $[{\bf u_{p}}(t),{\bf u_{e}}(t)]=F({\bf x}(t))$
such that $F({\bf x}(t))$ involves fixed and small number of algebraic
evaluations and comparisons.
\end{problem}
Note that pursuer and evader independently use this algorithm to evaluate
their respective optimal policies. The solution to Problem \ref{prob:implementable_feedback_law}
is given in Section \ref{sec:Main-Results}. 

\vspace*{0.25cm}

\noindent \textbf{Subsidiary Problem Statement}

\noindent In order to solve Problem \ref{prob:implementable_feedback_law},
we need to characterize the geometry of the optimal capture/evasion
trajectories. We accomplish this through the solution of the following
subsidiary problem, which aims to describe the optimal trajectories
in terms of some novel subsets of the pursuer's/evader's reachable
sets. Recall the definition of reachable sets:
\begin{defn}
\cite{cockayne1975plane} The reachable set of pursuer (evader), denoted
by $R_{p}({\bf p_{0}},\bar{t})\subset\mathbb{R}^{2}$ ($R_{e}({\bf e_{0}},\bar{t})\subset\mathbb{R}^{2}$),
at time $\bar{t}$ from initial state ${\bf p}(0)={\bf p_{0}}$ (${\bf e}(0)={\bf e_{0}}$),
is the set of all points that can be reached in time $t\leq\bar{t}$
by applying inputs ${\bf u_{p}}\in\mathcal{U}_{p}$ (${\bf u_{e}}\in\mathcal{U}_{e}$)
. 
\begin{eqnarray*}
R_{p}({\bf p_{0}},\bar{t}) & = & \{{\bf z}\in\mathbb{R}^{2}:\exists\,\,{\bf u_{p}}\in\mathcal{U}_{p}\text{ and corresponding}\\
 &  & \ \text{trajectory }{\bf p}\in\mathcal{C}(\mathbb{R}^{+},\mathbb{R}^{3})\text{ s.t. }\\
 &  & \text{\ensuremath{{\bf p}_{|\mathbb{R}^{2}}}(t)={\bf z} for some}\ t\leq\bar{t}\ \text{ and }{\bf p}(0)={\bf p_{0}}\}
\end{eqnarray*}
\end{defn}
The following subsidiary problem is solved in three parts in Theorem
\ref{thm:capture_strategy_pursuer_evader}, Theorem \ref{thm:cs-type}
and Theorem \ref{thm:one_valid_tangent} in Section \ref{sec:Main-Results}.
\begin{problem}
\label{prob:reachable_characterization}Characterize the feedback
time optimal pursuit evasion trajectories ${\bf u_{p}}={\bf \gamma}_{p}({\bf x})$,
${\bf u_{e}}={\bf \gamma}_{e}({\bf x})$ and the capture point ${\bf z}\in\mathbb{R}^{2}$
in terms of subsets of pursuer's and evader's reachable set $R_{p}({\bf p_{0}},T)$
and $R_{e}({\bf e_{0}},T)$ at capture time $T$. 
\end{problem}

\section{\label{sec:Preliminaries}Preliminaries}

This section contains some preliminary definitions and notations.

\subsection{\label{subsec:Minimum-radius-turning}Minimum radius turning circles
and common tangents}

First we define the minimum radius turning circles of the pursuer
and evader and also the common tangents which will be used to derive
feedback laws.
\begin{defn}
\textit{Clockwise pursuer-circle} (\textit{evader-circle}), $C_{i}(\bar{t})$,
at time $\bar{t}$ is the clockwise circle of radius $r_{i}=\frac{1}{w_{i_{m}}}$
and center $(x_{c}(\bar{t})=x(\bar{t})+\sin(\theta(\bar{t}))/w_{i_{m}},y_{c}(\bar{t})=y(\bar{t})-\cos(\theta(\bar{t}))/w_{i_{m}})$
that the pursuer (evader) follows when $w_{i}(t)=-w_{i_{m}}$ $\forall\ t\geq\bar{t}$,
where $i\in\{p,e\}$.
\end{defn}
\begin{defn}
\textit{Anti-clockwise pursuer-circle} (\textit{evader-circle}), $A_{i}(\bar{t})$,
at time $\bar{t}$ is the anti-clockwise circle of radius $r_{i}=\frac{1}{w_{i_{m}}}$
and center $(x_{c}(\bar{t})=x(\bar{t})-\sin(\theta(\bar{t}))/w_{i_{m}}$,
$y_{c}(\bar{t})=y(\bar{t})+\cos(\theta(\bar{t}))/w_{i_{m}})$ that
the pursuer (evader) follows when $w_{i}(t)=+w_{i_{m}}$ $\forall\ t\geq\bar{t}$,
where $i\in\{p,e\}$. 
\end{defn}
\begin{rem}
These circles for a particular position of the pursuer (evader) are
shown in Figure \ref{fig:ApCe} and will be called the pursuer-circles
(evader-circles). These circles are also called the minimum radius
turning circles of Dubins vehicle. Note that, the expression of $\dot{\theta}$
in (\ref{eq:pursuer_evader_dyn}) has the velocity term multiplied.
This results in turning radius being independent of $v_{i}(t)$ and
depends only on $w_{i}(t)$. Thus for $w_{i}(t)=w$ for $t\in[t_{1},t_{2}]$
the vehicle will move along an arc of a circle of radius $r_{i}=1/w$
with velocity $v_{i}(t)$ for the duration $[t_{1},t_{2}]$. If the
input $w_{i}(t)=0$ and $v_{i}(t)\in(0,v_{i_{m}}]$ for $i\in\{p,e\}$
in the time interval $[t_{1},t_{2}]$ then the pursuer (evader) moves
in a straight line during this time interval. 

\begin{figure}
\begin{centering}
\includegraphics[scale=0.9]{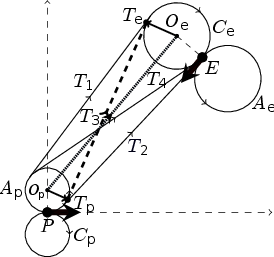}
\par\end{centering}
\caption{\label{fig:ApCe} Common Tangents $\{A_{p}(t),C_{e}(t)\}$}

\end{figure}
\end{rem}
In order to design feedback laws we make use of the common tangents
from the pursuer-circles to the evader-circles. Note that the pursuer-circles
(evader-circles) at time $t$ depend only on the position and orientation
of the pursuer (evader) at time instant $t$. Let $OC(t):=\{A_{p}(t),C_{p}(t),A_{e}(t),C_{e}(t)\}$
denote the set of pursuer-circles and evader-circles at time $t$.
A $PE$-pair is a pair of circles with one circle belonging to the
pursuer-circles and the other circle belonging to the evader-circles.
Thus, in total we have four $PE$-pairs. The set of $PE$-pairs at
time $t$ is denoted by $PE(t)$.
\begin{eqnarray*}
PE(t) & := & \{\{C_{p}(t),C_{e}(t)\},\,\{C_{p}(t),A_{e}(t)\},\\
 &  & \,\{A_{p}(t),C_{e}(t)\},\,\{A_{p}(t),A_{e}(t)\}\}
\end{eqnarray*}
Between the two circles belonging to a $PE$-pair, whose centers are
at a minimum distance of $r_{p}+r_{e}$ away from each other, there
will be four tangents. These tangents have been shown in between the
$PE$-pair $\{A_{p}(t),C_{e}(t)\}$ in Figure \ref{fig:ApCe}. We
assign direction to the tangents from the pursuer to the evader and
we call them directed common tangents. 
\begin{defn}
\label{def:Valid_tangent}Valid common tangent for a $PE-\text{pair}$
is a directed common tangent whose orientation matches with the direction
of both the pursuer circle and the evader circle in the $PE-\text{pair}$.
\end{defn}
In Figure \ref{fig:ApCe} only the tangent $T_{3}$ (shown by dashed
line) is a valid tangent for pair $\{A_{p}(t),C_{e}(t)\}$.
\begin{defn}
If a pursuer's (evader's) trajectory up to some time $t_{f}>0$ is
such that it traverses one of the pursuer-circles (evader-circles)
in time interval $[0,t']$ with $t'\leq\min(t_{f},2\pi r_{p}/v_{p_{m}})$,
and then traverses one of the tangents to that pursuer-circle (evader-circle)
in time interval $[t',t_{f}]$, then such a \textit{trajectory is
of the type $CS$} (circle and straight line) up to time $t_{f}$.
\end{defn}
\begin{rem}
The condition $t'\leq2\pi r_{p}/v_{p_{m}}$ ensures that no part along
the circumference of the circle is traversed more than once. If a
trajectory of type $CS$ follows anticlockwise (left) circle and after
that follows a straight line path then we say the trajectory belongs
to the type $LS$. Similarly, if it follows clockwise (right) circle
and after that follows a straight line path then we say the trajectory
is of the type $RS$. 
\end{rem}
\begin{defn}
If a pursuer's (evader's) trajectory up to time $t_{f}$ is such that
it traverses one of the pursuer-circles (evader-circles) say $A_{p}(0)$
($A_{e}(0)$) in time interval $[0,t']$ with $t'\leq\text{min}(t_{f},2\pi r_{p}/v_{p_{m}})$,
and then traverses the other pursuer-circle (evader-circle) i.e. $C_{p}(0)$
($C_{e}(0)$) in the time interval $[t',t_{f}]$, then such a \textit{trajectory
is of the type $CC$ }(circle-circle).
\end{defn}

\subsection{\label{subsec:Reachable-Sets}Reachable sets of Dubins vehicle}

Next, we describe the reachable sets of Dubins vehicle. Reachable
set is used to characterize the solution for the game of two cars.
The reachable set for the Dubins vehicle has been studied in \cite{cockayne1975plane,bui1994accessibility}.
It is known that, the points inside the pursuer (evader) circles can
be reached in minimum time by $CC$ (circle-circle) types of curves.
The points external to the pursuer (evader) circles can be reached
in minimum time by $CS$ type of curves. The external boundary of
$R_{p}({\bf p_{0}},\bar{t})$ ($R_{e}({\bf p_{0}},\bar{t})$) is denoted
by $\partial R_{p}({\bf p_{0}},\bar{t})$ ($\partial R_{e}({\bf e_{0}},\bar{t})$).

It is known \cite{bui1994accessibility} that, if $\bar{t}\geq2\pi r_{p}/v_{p_{m}}$
($\bar{t}\geq2\pi r_{e}/v_{e_{m}}$) the points on $\partial R_{p}({\bf p_{0}},\bar{t})$
($\partial R_{e}({\bf e_{0}},\bar{t})$) at time $\bar{t}$ can be
reached only by the trajectories of the type $CS$. Thus $\partial R_{p}({\bf p_{0}},\bar{t})$
($\partial R_{e}({\bf e_{0}},\bar{t})$) at $\bar{t}\geq2\pi r_{p}/v_{p_{m}}$
($\bar{t}\geq2\pi r_{e}/v_{e_{m}}$) is comprised of two portions.
The first portion is characterized by trajectories which begin on
anti-clockwise circle and then follow a straight line. The second
portion is characterized by trajectories which begin on the clockwise
circle and then follow a straight line. 

Consider the pursuer with initial state vector ${\bf p}(0)=\left[x_{0}\ y_{0}\ \theta_{0}\right]^{\top}$.
The trajectories ${\bf p}\in\mathcal{C}(\mathbb{R}^{+},\mathbb{R}^{3})$
corresponding to the input 
\begin{eqnarray*}
w_{p}(t) & = & +w_{p_{m}}\quad t\in[0,t_{1}]\\
w_{p}(t) & = & 0\quad t\in(t_{1},\bar{t}]\\
v_{p}(t) & = & v_{p_{m}}\quad t\in[0,\bar{t}]
\end{eqnarray*}
for some $0\leq t_{1}\leq2\pi r_{p}/v_{p_{m}}$, will initially follow
the anti-clockwise circle and then travel on a tangent to the anti-clockwise
circle. The pursuer moves up to time $\bar{t}>2\pi r_{p}/v_{p_{m}}$
with speed $v_{p_{m}}$ throughout. Since $2\pi r_{p}/v_{p_{m}}$
is the time required to travel a complete circle, it will cover a
distance greater than circumference of the circle. (Note that the
switching time $t_{1}$ is less than $2\pi r_{p}/v_{p_{m}}$ so that
the no length of the circle is repeated). The trajectory is parameterized
by the switching time $t_{1}$ and can be obtained by integrating
(\ref{eq:pursuer_evader_dyn}) as
\noindent \begin{flushleft}
\begin{eqnarray}
x_{fl}(\bar{t}) & = & x_{0}+(\sin(\tilde{\theta})-\sin(\theta_{0}))/w_{p_{m}}+v_{p_{m}}\cos(\tilde{\theta})\tilde{t}\label{eq:left_reachable_set}\\
y_{fl}(\bar{t}) & = & y_{0}-(\cos(\tilde{\theta})-\cos(\theta_{0}))/w_{p_{m}}+v_{p_{m}}\sin(\tilde{\theta})\tilde{t}\nonumber 
\end{eqnarray}
where $\tilde{\theta}=\theta_{0}+v_{p_{m}}w_{p_{m}}t_{1}$ and $\tilde{t}=\bar{t}-t_{1}$.
Using these, the left reachable set of pursuer is defined as 
\begin{eqnarray*}
R_{p}^{l}({\bf p_{0}},\bar{t}) & = & \{{\bf z}=[x\ y]^{\top}\in\mathbb{R}^{2}|x=x_{fl}(t)\ \text{and}\ \\
 &  & \ y=y_{fl}(t)\ \forall\ t_{1}\leq\bar{t},t\leq\bar{t}\ \text{s.t.}t_{1}<t\}
\end{eqnarray*}
The left reachable set for the evader $R_{e}^{l}({\bf e_{0}},\bar{t})$
is defined analogously. The left reachable set is shown in Figure
\ref{fig:Left-Reachable-Set}. The boundary of left reachable set
of pursuer (evader) is denoted by $\partial R_{p}^{l}({\bf p_{0}},\bar{t})$
($\partial R_{e}^{l}({\bf e_{0}},\bar{t})$).
\par\end{flushleft}

For $\bar{t}\geq2\pi r_{p}/v_{p_{m}}$ the right reachable sets for
the pursuer $R_{p}^{r}({\bf p_{0}},\bar{t})$ and evader $R_{e}^{r}({\bf e_{0}},\bar{t})$
are defined similarly by the trajectories which first travel on the
clockwise circle and then on the tangent to the clockwise circle.
The right reachable set is shown in Figure \ref{fig:Right-Reachable-Set}.
The boundary of right reachable set of pursuer (evader) is denoted
by $\partial R_{p}^{r}({\bf p_{0}},\bar{t})$ ($\partial R_{e}^{r}({\bf e_{0}},\bar{t})$).

The left reachable set and right reachable set of the pursuer (evader)
are the subsets of the reachable set $R_{p}({\bf p_{0}},\bar{t})$
($R_{e}({\bf e_{0}},\bar{t})$). The union of $R_{p}^{l}({\bf p_{0}},\bar{t})$
and $R_{p}^{r}({\bf p_{0}},\bar{t})$ is shown in Figure \ref{fig:Complete-Reachable-Set}. 

The boundary of left reachable set is divided in two parts for $\bar{t}\geq2\pi r_{p}/v_{p_{m}}$
in next definition.
\begin{defn}
\label{def:internal_boundary}For $\bar{t}\geq2\pi r_{p}/v_{p_{m}}$,
the portion $ADE$ as shown in Figure \ref{fig:Left-Reachable-Set}
will be called the \textbf{\textit{left external boundary}} whereas
the portion $EFB$ will be called the \textbf{\textit{left internal
boundary}}. Similarly, for the right reachable set the portion $ADE$
as shown in Figure \ref{fig:Right-Reachable-Set} will be called the
\textbf{\textit{right external boundary}} whereas the portion $EFB$
will be called the \textbf{\textit{right internal boundary}} for $\bar{t}\geq2\pi r_{p}/v_{p_{m}}$. 
\end{defn}
\begin{rem}
The external boundary of the reachable set $\partial R_{p}({\bf p_{0}},\bar{t})$
($\partial R_{e}({\bf e_{0}},\bar{t})$) is the union of left external
boundary and the right external boundary for $\bar{t}\geq2\pi r_{p}/v_{p_{m}}$.
Note that the shape of the reachable sets is as shown in Figures \ref{fig:Left-Reachable-Set},
\ref{fig:Right-Reachable-Set}, \ref{fig:Complete-Reachable-Set}
only if $\bar{t}\geq2\pi r_{e}/v_{e_{m}}\geq2\pi r_{p}/v_{p_{m}}$.
\end{rem}
\begin{figure*}
\begin{minipage}[t]{0.3\textwidth}%
\begin{center}
\includegraphics[scale=0.33]{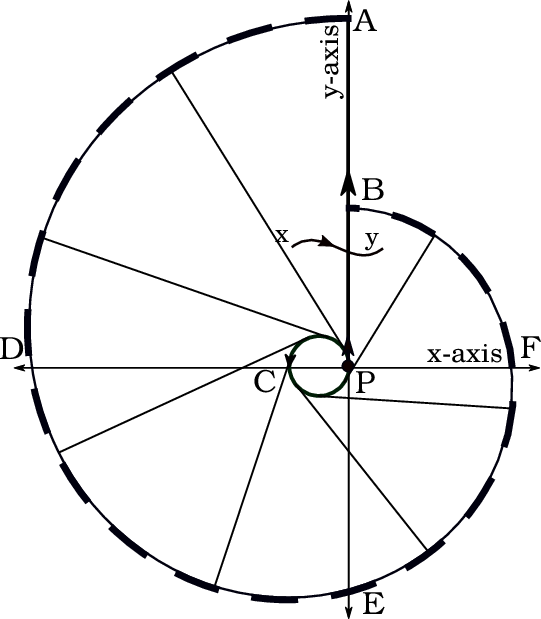}
\par\end{center}
\caption{\label{fig:Left-Reachable-Set}Left reachable set ($R_{p}^{l}({\bf p_{0}},\bar{t})$)}
\end{minipage}\hfill{}%
\begin{minipage}[t]{0.3\textwidth}%
\begin{center}
\includegraphics[scale=0.33]{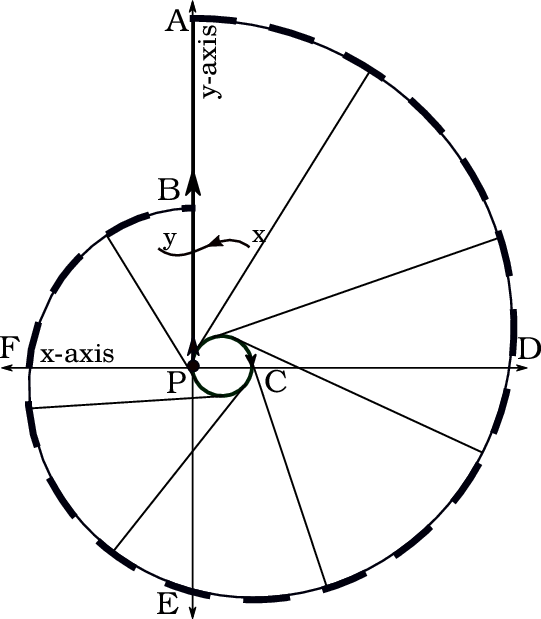}
\par\end{center}
\caption{\label{fig:Right-Reachable-Set}Right reachable set ($R_{p}^{r}({\bf p_{0}},\bar{t})$)}
\end{minipage}\hfill{}%
\begin{minipage}[t]{0.3\textwidth}%
\begin{center}
\includegraphics[scale=0.33]{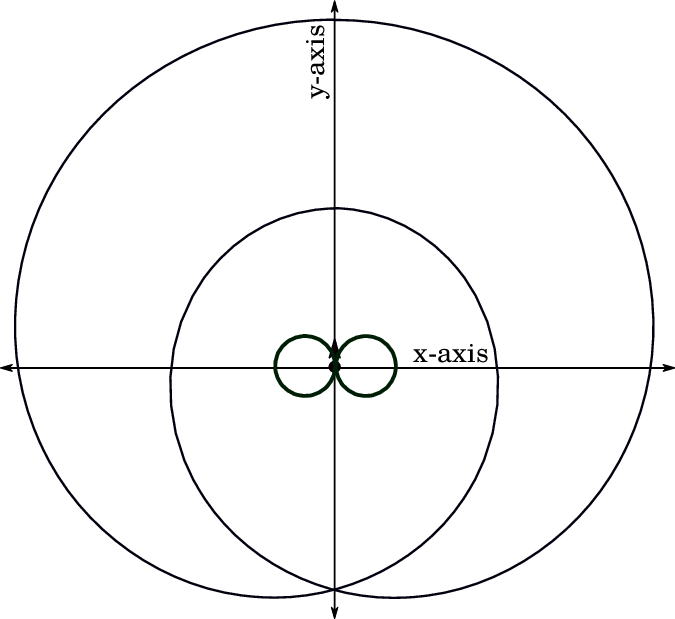}
\par\end{center}
\caption{\label{fig:Complete-Reachable-Set}$R_{p}^{r}({\bf p_{0}},\bar{t})\cup R_{p}^{l}({\bf p_{0}},\bar{t})$ }
\end{minipage}
\end{figure*}

\section{\label{sec:analysis_reachability}Counter examples using reachable
sets}

The reachable set characterizes the points which the Dubins vehicle
can reach in a given time. It would seem that the evader can always
escape capture if the evader's reachable set is not contained completely
inside the pursuer's reachable set. However, this is only possible
if there exists an evader trajectory which can enter the region not
contained in pursuer's reachable set, without passing through the
pursuer's reachable set. For if the trajectory passes through the
pursuer's reachable set it will be intercepted by some pursuer's trajectory.
This notion is formalized in the next definition by introducing the
safe region of the evader with respect to subsets of pursuer's reachable
set. 
\begin{defn}
The \textbf{\textit{safe region}} of the evader, at time $\bar{t}$,
with respect to a subset of pursuer's reachable set $R_{p}^{s}({\bf p_{0}},\bar{t})\subset R_{p}({\bf p_{0}},\bar{t})$
is defined as
\begin{eqnarray*}
R_{e}^{-}({\bf e_{0}},R_{p}^{s},\bar{t}) & = & \{{\bf z}\in\mathbb{R}^{2}:{\bf z}\in R_{e}({\bf e_{0}},\bar{t})\backslash R_{p}^{s}({\bf p_{0}},\bar{t})\ \\
 &  & \text{and }\exists\ \text{an evader}\text{ trajectory }\\
 &  & \text{with }{\bf e}_{|\mathbb{R}^{2}}(t_{1})={\bf z}\text{ for some }t_{1}\leq\bar{t},\\
 &  & \text{and }{\bf e}{}_{|\mathbb{R}^{2}}(t)\notin R_{p}^{s}({\bf p_{0}},t)\ \text{\ensuremath{\forall}}\ t\leq\bar{t}\}
\end{eqnarray*}
\end{defn}
\begin{lem}
\label{lem:left_containment}Let $T_{l}$ be the minimum time such
that $R_{e}^{-}({\bf e_{0}},R_{p}^{l},T_{l})=\emptyset$. If $d_{pe}^{0}\geq2r_{p}+2\pi r_{p}(v_{e_{m}}/v_{p_{m}})$
then $T_{l}<\infty$.
\end{lem}
\begin{IEEEproof}
See Appendix.
\end{IEEEproof}
\begin{lem}
\label{lem:right_containment}Let $T_{r}$ be the minimum time such
that $R_{e}^{-}({\bf e_{0}},R_{p}^{r},T_{r})=\emptyset$. If $d_{pe}^{0}\geq2r_{p}+2\pi r_{p}(v_{e_{m}}/v_{p_{m}})$
then $T_{r}<\infty$.
\end{lem}
\begin{lem}
\label{lem:complete_containment}Let $T_{o}$ be the minimum time
such that $R_{e}^{-}({\bf e_{0}},R_{p},T_{o})=\emptyset$. If $d_{pe}^{0}\geq2r_{p}+2\pi r_{p}(v_{e_{m}}/v_{p_{m}})$
then $T_{o}<\infty$.
\end{lem}
\begin{IEEEproof}
Follows from Lemma \ref{lem:left_containment} and Lemma \ref{lem:right_containment}.
\end{IEEEproof}
\begin{rem}
It would seem that the condition $R_{e}^{-}({\bf e_{0}},R_{p},\bar{t})=\emptyset$,
would be necessary and sufficient for capture. Let $T_{o}=\inf\{t\in\mathbb{R}:R_{e}^{-}({\bf e_{0}},R_{p},t)=\emptyset\}$.
It is shown in \cite{tsiotras2017reachabilitysets}, for pursuer and
evader which can turn instantaneously, that such a condition is indeed
necessary and sufficient for capture if we consider open-loop strategies
and capture occurs at $T_{o}$. However, we demonstrate through the
following counter-examples that the claim does not hold if we consider
feedback strategies for the pursuer and evader.
\end{rem}
\begin{example}
\label{exa:evader_behind}Refer to Figure \ref{fig:invalid_containment}
where the evader is exactly behind the pursuer and oriented away from
the pursuer. The pursuer's initial position is ${\bf p_{0}}=[0,0,\pi/2]$
while that of the evader is ${\bf e_{0}}=[-11,0,-\pi/2]$. Assume
$v_{p_{m}}=2$, $w_{p_{m}}=0.2$, $v_{e_{m}}=1.15$ and $w_{e_{m}}=0.18$.
At time $T_{o}=31.34$, the evader's reachable set (dotted curve)
is contained in pursuer's reachable set (solid line). If $R_{e}^{-}({\bf e_{0}},R_{p},T_{o})=\emptyset$
was the criterion for capture, then the evader would have traveled
straight and the pursuer would have intercepted it along the $CS$
path exactly at the point marked by star in Figure \ref{fig:invalid_containment}.
However, the optimal feedback strategies obtained by numerical simulation
(using the algorithms in \cite{Raivio2000}) and the reachable set
for such a strategy are as shown in Figure \ref{fig:valid_containment}.
Such a situation happens at time $T=35.9$. Hence, $R_{e}^{-}({\bf e_{0}},R_{p},T_{o})=\emptyset$
is not a sufficient criterion for capture. 
\end{example}
\begin{example}
\label{exa:evader_front}An alternate hypothesis might be that capture
occurs when either the left reachable set or the right reachable set
completely contains the evader's reachable region as shown in Figure
\ref{fig:valid_containment}. But consider the case when evader is
in front of the pursuer as shown in Figure \ref{fig:valid_front}.
Again the game is solved numerically using algorithms in \cite{Raivio2000}.
At the point of capture, marked by the star, neither the left nor
the right reachable set of the pursuer contain the evader's reachable
set.
\end{example}
Examples \ref{exa:evader_behind} and \ref{exa:evader_front} indicate
that reachable sets themselves do not characterize capture and a novel
geometric interpretation is required to understand the nature of optimal
trajectories. This is achieved next by introducing the notion of continuous
subsets of reachable sets.

\noindent 
\begin{figure*}
\begin{minipage}[t]{0.32\textwidth}%
\begin{center}
\includegraphics[scale=0.35]{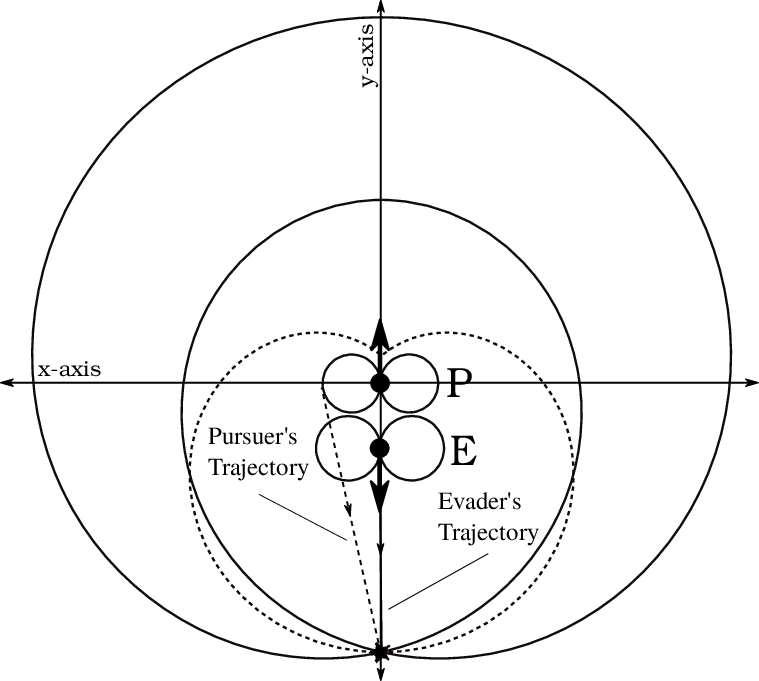}
\par\end{center}
\caption{\label{fig:invalid_containment}Invalid containment of evader's reachable
set}
\end{minipage}\hfill{}%
\begin{minipage}[t]{0.3\textwidth}%
\begin{center}
\includegraphics[scale=0.35]{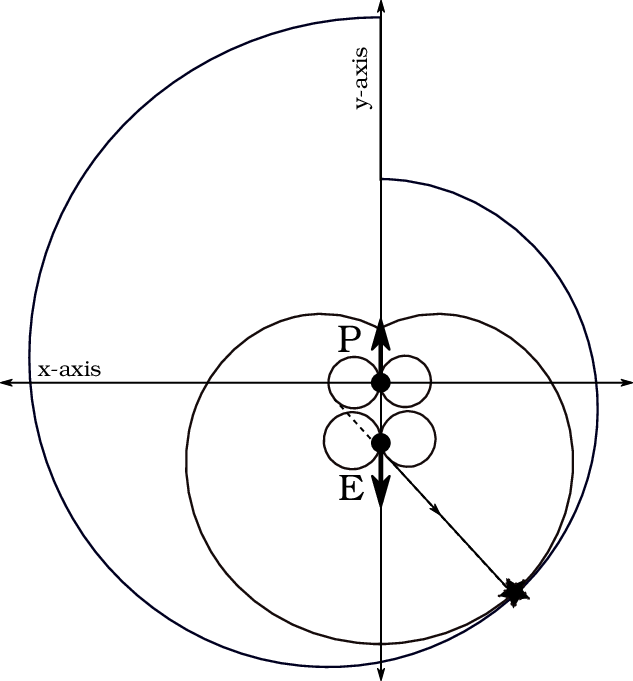}
\par\end{center}
\caption{\label{fig:valid_containment}Valid containment of evader's reachable
set}
\end{minipage}\hfill{}%
\begin{minipage}[t]{0.3\textwidth}%
\begin{center}
\includegraphics[scale=0.4]{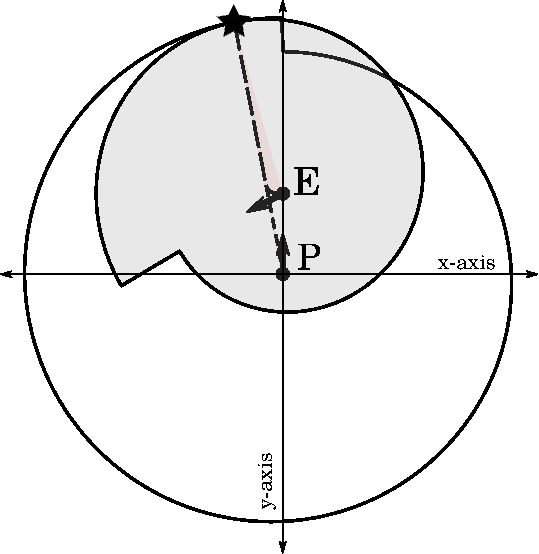}
\par\end{center}
\caption{\label{fig:valid_front}Left reachable set does not contain evader's
reachable set. Yet capture occurs at this time.}
\end{minipage}
\end{figure*}

\section{\label{sec:continuous_subsets}Continuous subsets of reachable set}

Recall that $T_{o}=\inf\{t\in\mathbb{R}:R_{e}^{-}({\bf e_{0}},R_{p},t)=\emptyset\}$
and consider a point ${\bf z}\in\partial R_{e}({\bf e}_{0},T_{o})\cap\partial R_{p}({\bf p}_{0},T_{o})$.
In Example \ref{exa:evader_behind} it was shown that the capture
did not occur at point ${\bf z}$. We try to explain this phenomena
using small deviations around the evader input signal. This will result
in a small deviation in the final point ${\bf e}_{|\mathbb{R}^{2}}(T_{o})={\bf z}$. 

\textbf{\textit{Observation:}}\textit{ If capture is to occur at ${\bf z}={\bf e}_{|\mathbb{R}^{2}}(T_{o})$,
using feedback pursuer strategies, every feasible variation in ${\bf e}$
should be traceable by the pursuer using small variations of its own
input signal.} 

We make the notion of admissible variations more concrete in this
section by defining the continuous subsets of the reachable set for
the Dubins vehicle.

\subsubsection*{Definition of Continuous Subsets of Reachable Sets }

Let $R_{p}({\bf p_{0}},\bar{t})$ be the reachable set of the pursuer
at time $\bar{t}$ starting from initial position ${\bf p}_{0}\in\mathbb{R}^{3}$. 
\begin{enumerate}
\item Let ${\bf x,y}\in R_{p}({\bf p_{0}},\bar{t})$ and let $\overline{{\bf xy}}$
be any curve from ${\bf x}$ to ${\bf y}$ such that all points on
$\overline{{\bf xy}}$ belong to $R_{p}({\bf p_{0}},\bar{t})$ as
shown in Figure \ref{fig:Continuum-of-Trajectories}. 
\item Let $\mathfrak{T}_{\overline{{\bf x}{\bf y}}}({\bf r},{\bf p_{0}},\bar{t})$
be the set of all the trajectories from ${\bf p_{0}}$ to the point
${\bf r}$ on $\overline{{\bf x}{\bf y}}$, which reach the point
${\bf r}\in R_{p}({\bf p_{0}},\bar{t})$ at time $t\leq\bar{t}$ starting
from ${\bf p_{0}}$.
\begin{eqnarray*}
\mathfrak{T}_{\overline{{\bf x}{\bf y}}}({\bf r},{\bf p_{0}},\bar{t}) & = & \{{\bf p}\in\mathcal{C}^{0}(\mathbb{R}^{+},\mathbb{R}^{3}):\exists\,\,{\bf u_{p}}\in\mathcal{U}_{p}\,\\
 &  & \text{ with }{\bf p}_{|\mathbb{R}^{2}}(t)={\bf r}\\
 &  & \,\,\,\text{ and }{\bf p}(0)={\bf p_{0}}\ \text{for }\text{some }t\leq\bar{t}\}
\end{eqnarray*}
\item Let $\mathfrak{L}(\overline{{\bf x}{\bf y}},{\bf p_{0}},\bar{t})$
be the collection of all possible trajectories from ${\bf p_{0}}$
which reach line $\overline{{\bf x}{\bf y}}$ at time $t\leq\bar{t}$
.
\begin{eqnarray*}
\mathfrak{L}(\overline{{\bf x}{\bf y}},{\bf p_{0}},\bar{t}) & = & \Big\{{\bf p}\in\mathcal{C}^{0}(\mathbb{R}^{+},\mathbb{R}^{3}):\exists\ {\bf r}\in\overline{{\bf x}{\bf y}}\ \text{s.t.}\\
 &  & \ {\bf p}\in\mathfrak{T}_{\overline{{\bf x}{\bf y}}}({\bf r},{\bf p_{0}})\Big\}
\end{eqnarray*}
\item Now we look at functions $\mathfrak{f}:\overline{{\bf x}{\bf y}}\rightarrow\mathfrak{L}(\overline{{\bf x}{\bf y}},{\bf p_{0}},\bar{t})$
such that $\mathfrak{f}({\bf r})\in\mathfrak{T}_{\overline{{\bf x}{\bf y}}}({\bf r},{\bf p_{0}},\bar{t})$
$\forall\ {\bf r}\in\overline{{\bf xy}}$. We construct a set $\mathfrak{C}(\overline{{\bf x}{\bf y}},{\bf p_{0}},\bar{t})$
such that
\begin{eqnarray*}
\mathfrak{C}(\overline{{\bf x}{\bf y}},{\bf p_{0}},\bar{t}) & = & \{\mathfrak{f}|\mathfrak{f}:\overline{{\bf x}{\bf y}}\rightarrow\mathfrak{L}(\overline{{\bf x}{\bf y}},{\bf p_{0}},\bar{t})\ \text{and }\\
 &  & \mathfrak{f}({\bf r})\in\mathfrak{T}_{\overline{{\bf x}{\bf y}}}({\bf r},{\bf p_{0}},\bar{t})\forall\ {\bf r}\in\overline{{\bf xy}}\}
\end{eqnarray*}
The set $\mathfrak{C}(\overline{{\bf x}{\bf y}},{\bf p_{0}},\bar{t})$
is the set of all the functions which map some point ${\bf r}\in\overline{{\bf x}{\bf y}}$
to a trajectory ${\bf p}\in\mathfrak{L}(\overline{{\bf x}{\bf y}},{\bf p_{0}},\bar{t})$
and the trajectory ${\bf p}$ is such that ${\bf p}_{|\mathbb{R}^{2}}(\bar{t})={\bf r}$.
\item Let the metric be defined on set $\overline{{\bf x}{\bf y}}$ by the
standard two norm. The metric on the set $\mathfrak{L}(\overline{{\bf x}{\bf y}},{\bf p_{0}},\bar{t})$
is defined by the $\mathcal{L}_{\infty}$ norm.
\end{enumerate}
\begin{defn}
\label{def:continuum_set}We say that the curve $\overline{{\bf x}{\bf y}}$
is a \textbf{\textit{continuum set}} if there exists a continuous
one-one map $\mathfrak{f}^{c}\in\mathfrak{C}(\overline{{\bf x}{\bf y}},{\bf p_{0}},\bar{t})$. 
\end{defn}
\begin{defn}
The range of $\mathfrak{f}^{c}$ will be called the \textbf{\textit{continuum
of trajectories}}. 
\end{defn}
\begin{rem}
Since, $\mathfrak{f}^{c}$ is a continuous function, the trajectories
in the continuum of trajectories of $\mathfrak{f}^{c}$ are such that
for any two points ${\bf w,z}\in\overline{{\bf x}{\bf y}}$ we have
$||\mathfrak{f}^{c}({\bf w})-\mathfrak{f}^{c}({\bf z})||_{\mathcal{L}_{\infty}}\rightarrow0$
as $||{\bf w}-{\bf z}||_{2}\rightarrow0$.
\end{rem}
\begin{defn}
A \textbf{\textit{continuous subset}} of the reachable set at time
$\bar{t}$, $R_{p}^{c}({\bf p_{0}},\bar{t})\subseteq R_{p}({\bf p_{0}},\bar{t})$
is a connected set of all the points s.t. 
\end{defn}
\begin{enumerate}
\item For all ${\bf x,y}\in R_{p}^{c}({\bf p_{0}},\bar{t})$ and every curve
$\overline{{\bf x}{\bf y}}\in R_{p}^{c}({\bf p_{0}},\bar{t})$, the
set $\overline{{\bf x}{\bf y}}$ is a continuum set.
\item For every point ${\bf r}\in R_{p}^{c}({\bf p_{0}},\bar{t})$ there
exists a trajectory ${\bf p}\in\mathcal{C}^{0}(\mathbb{R}^{+},\mathbb{R}^{3})$
such that ${\bf p}(0)={\bf p_{0}},{\bf p}(\tilde{t})={\bf r}$ for
some $0\leq\tilde{t}\leq\bar{t}$ and ${\bf p}(t)\in R_{p}^{c}({\bf p_{0}},\bar{t})$
for all $t\leq\bar{t}$.
\end{enumerate}
From here on we will denote a continuous subset of pursuer's reachable
set by $CSP$. A continuous subset of the reachable set of the evader
is defined analogously and will be denoted by $CSE$. From the definition
it is clear that the continuous subsets of reachable set are not unique.
The collection of all the $CSP$ ($CSE$) at time $\bar{t}$ from
initial position ${\bf p_{0}}$ (${\bf e_{0}}$) is denoted by $\mathfrak{R}_{p}({\bf p_{0}},\bar{t})$
($\mathfrak{R}_{e}({\bf e_{0}},\bar{t})$).
\begin{eqnarray*}
\mathfrak{R}_{e}({\bf e_{0}},\bar{t}) & = & \{R_{e}^{c}({\bf e_{0}},\bar{t}):R_{e}^{c}({\bf e_{0}},\bar{t})\subseteq R_{e}({\bf e_{0}},\bar{t})\}\\
\mathfrak{R}_{p}({\bf p_{0}},\bar{t}) & = & \{R_{p}^{c}({\bf p_{0}},\bar{t}):R_{p}^{c}({\bf p_{0}},\bar{t})\subseteq R_{p}({\bf p_{0}},\bar{t})\}
\end{eqnarray*}
The continuous safe region of a evader's continuous subset $R_{e}^{c}({\bf e_{0}},\bar{t})\in\mathfrak{R}_{e}({\bf e_{0}},\bar{t})$
with respect to a continuous subset of the pursuer $R_{p}^{c}({\bf p_{0}},\bar{t})\in\mathfrak{R}_{p}({\bf p_{0}},\bar{t})$
is defined next.
\begin{defn}
The \textbf{\textit{continuous safe region}} $R_{e}^{c-}({\bf e_{0}},R_{p}^{c},\bar{t})$
of a $CSE$ $R_{e}^{c}({\bf e_{0}}$$,\bar{t})$, at time $\bar{t}$,
with respect to a $CSP$ $R_{p}^{c}({\bf p_{0}},\bar{t})$ is defined
as
\begin{eqnarray*}
R_{e}^{c-}({\bf e_{0}},R_{p}^{c},\bar{t}) & = & \{{\bf z}\in\mathbb{R}^{2}:{\bf z}\in R_{e}^{c}({\bf e_{0}},\bar{t})\backslash R_{p}^{c}({\bf p_{0}},\bar{t})\ \\
 &  & \text{and }\exists\ \text{an evader}\text{ trajectory with}\\
 &  & {\bf e}_{|\mathbb{R}^{2}}(t_{1})={\bf z}\text{ for some}\ t_{1}\leq\bar{t}\text{ and }\\
 &  & {\bf e}_{|\mathbb{R}^{2}}(t)\notin R_{p}^{c}({\bf p_{0}},t)\ \text{at each}\ t\leq t_{1}\}
\end{eqnarray*}
\end{defn}
When the safe region is an empty set we say that the $CSP$ $R_{p}^{c}({\bf p_{0}},\bar{t})$
contains $CSE$ $R_{e}^{c}({\bf e_{0}},\bar{t})$ continuously.
\begin{defn}
\textbf{\textit{Continuous containment}}: If the continuous safe region
of a $CSE$ $R_{e}^{c}({\bf e_{0}},\bar{t})$, at time $\bar{t}$,
with respect to $CSP$ $R_{p}^{c}({\bf p_{0}},\bar{t})$ be such that
$R_{e}^{c-}({\bf e_{0}},R_{p}^{c},\bar{t})=\emptyset$, then we say
that $R_{p}^{c}({\bf p_{0}},\bar{t})$ contains $R_{e}^{c}({\bf e_{0}},\bar{t})$
continuously.
\end{defn}
It turns out that optimal capture occurs when every continuous subset
of the evader is contained inside some continuous subset of the pursuer.
The time $T$ at which such a situation occurs is defined next. Let
\begin{eqnarray}
T & := & \inf\{t\in\mathbb{R}:\forall\ R_{e}^{c}({\bf e_{0}},t)\in\mathfrak{R}_{e}({\bf e_{0}},t)\ \label{eq:feedback_time_to_capture}\\
 &  & \text{\ensuremath{\exists}}\ R_{p}^{c}({\bf p_{0}},t)\in\mathfrak{R}_{p}({\bf p_{0}},t)\ \text{s.t. }R_{e}^{c-}({\bf e_{0}},R_{p}^{c},t)=\emptyset\}\nonumber 
\end{eqnarray}

\section{\textcolor{blue}{\label{sec:Main-Results}}Main Results}

In this section we state the main contributions of the paper. These
have been proved in subsequent sections.

\subsection{Capture characterization using continuous subsets}

\noindent The next theorem states that the minimum time at which continuous
containment occurs ($T$ as defined by (\ref{eq:feedback_time_to_capture}))
is the same as the $\min-\max$ time ($T^{*}$ as per Definition \ref{def:saddle_point})
$T=T^{*}$. This implies that for all time $t<T=T^{*}$ there exists
an evader feedback policy such that irrespective of any feedback policy
of the pursuer, the evader is able to avoid capture. Also, there exists
a feedback pursuer policy such that irrespective of any feedback policy
of the evader the pursuer is able to capture it at some time $t<T=T^{*}$.
\begin{thm}
\label{thm:capture_strategy_pursuer_evader}Let $T$ and $T^{*}$
be as per (\ref{eq:feedback_time_to_capture}) and Definition \ref{def:saddle_point}.
$T$ is the $\min-\max$ time to capture i.e. $T=T^{*}$.
\end{thm}
\begin{IEEEproof}
See Section \ref{sec:imp_thm_1}.
\end{IEEEproof}
Theorem \ref{thm:capture_strategy_pursuer_evader} demonstrates that
continuous containment is a complete characterization of optimal capture.
However, to completely understand the nature of the optimal trajectories
we investigate the specific continuous subsets of the pursuer and
evader reachable sets that actually achieve the inf $T$ in (\ref{eq:feedback_time_to_capture}).
Through a careful study of these special $CSE/CSP's$ we are able
to prove the following theorem that describes the optimal trajectories.
\begin{thm}
\label{thm:cs-type}If the initial distance $d_{pe}^{0}\geq2r_{e}+2\pi r_{e}(v_{p_{m}}/v_{e_{m}})$,
then the saddle-point strategies of the evader and the pursuer are
of the type $CS$ that is a circle and straight line.
\end{thm}
\begin{IEEEproof}
See Section \ref{sec:imp_theorem_2}.
\end{IEEEproof}
The trajectories of both the pursuer and the evader are of the type
$CS$ and since it is guaranteed that capture will occur, the straight
lines are coincident to the capture point. This observation allows
us to prove the next geometric result.
\begin{thm}
\label{thm:one_valid_tangent}If $d_{pe}^{0}\geq2r_{e}+2\pi r_{e}(v_{p_{m}}/v_{e_{m}})$
then the saddle-point strategies of the pursuer and the evader result
in pursuer and evader trajectories being coincident to circles in
an $PE$-pairs and one of the common tangents of that pair. 
\end{thm}
\begin{IEEEproof}
See Section \ref{sec:impt_theorem_3}.
\end{IEEEproof}
Clearly Theorem \ref{thm:one_valid_tangent} effectively solves the
subsidiary Problem \ref{prob:reachable_characterization} listed in
Section \ref{sec:problem_formulation}. The simple geometry of optimal
curves defined by Theorem \ref{thm:one_valid_tangent} lets us propose
an algorithm to immediately compute the feedback control for both
the pursuer and the evader at each instant. This solution to Problem
\ref{prob:implementable_feedback_law} is described next\textcolor{blue}{.}

\subsection{\label{sec:feedback_law_numerical_simulations}Feedback law using
geometry}

In this section we design algorithms to select an appropriate tangent
which describes the feedback saddle-point strategies. As discussed
in Section \ref{subsec:Minimum-radius-turning} each $PE$-pair has
four common tangents. Since there are four such $PE$-pairs we will
have 16 directed tangents in total. First we show that at any time
$t$, only one directed tangent corresponding to each $PE$-pair in
the set $PE(t)$ is a valid tangent along which the saddle-point trajectories
may occur. 
\begin{lem}
\label{lem:valid_tangent}At each time $t$, corresponding to each
pair of $PE$-circles in the set $PE(t)$ there is only one valid
common tangent with which saddle point strategies can coincide.
\end{lem}
\begin{IEEEproof}
We prove this on a case by case basis. Consider a circle pair $\{A_{p}(t),C_{e}(t)\}\in PE(t)$
as shown in Figure \ref{fig:ApCe}. A directed straight line $\overrightarrow{O_{p}O_{e}}$
is drawn from the center of circle $A_{p}$ to the center of circle
$C_{e}$. Clearly, the tangents which end on the evader to the right
of $\overrightarrow{O_{p}O_{e}}$ are not feasible as the direction
of the tangents do not match with the orientation of the evader. Similarly,
the tangents which end on $C_{e}(t)$ to the left of the line $\overrightarrow{O_{p}O_{e}}$
are also not feasible. Hence, we can eliminate tangents $T_{1}$,
$T_{2}$ and $T_{4}$. Thus, only the tangent $T_{3}$ (dashed line)
is a feasible one. Similarly, the claim can be proven for other pairs
in the set $PE(t)$. 
\end{IEEEproof}
Algorithm \ref{alg:Valid-Tangent} is designed to compute the valid
tangent for each $PE$-pair. The common tangents of all the $PE$-pairs
have been shown in Figures \ref{fig:ApCe}, \ref{fig:ApAe}, \ref{fig:CpAe},
and \ref{fig:CpCe}. In each case the valid tangent has been shown
by a dashed line. 
\begin{algorithm}
\begin{enumerate}
\item For the tangent under consideration let $T_{p}$ be its intersection
point with the pursuer circle under consideration and $T_{e}$ be
the intersection point with the evader circle under consideration.
\item The following observations can be seen from Figures \ref{fig:ApCe},
\ref{fig:ApAe}, \ref{fig:CpAe}, and \ref{fig:CpCe} for valid tangent.
\begin{enumerate}
\item For $A_{p}(t)$ and $A_{e}(t)$ the angle between the valid tangent
and $\overrightarrow{O_{p}T_{p}}$ and $\overrightarrow{O_{e}T_{e}}$
translated to $T_{p}$ and $T_{e}$ respectively is $\pi/2$ in anti-clockwise
direction.
\item For $C_{p}(t))$ and $C_{e}(t)$ the angle between the valid tangent
and $\overrightarrow{O_{p}T_{p}}$ and $\overrightarrow{O_{e}T_{e}}$
translated to $T_{p}$ and $T_{e}$ respectively is $-\pi/2$ in anti-clockwise
direction.
\end{enumerate}
\item For a given directed tangent $\overrightarrow{T}$ in a $PE$-pair
if the angles with $\overrightarrow{O_{p}T_{p}}$ and $\overrightarrow{O_{e}T_{e}}$
satisfy the conditions above then it is a valid tangent.
\end{enumerate}
\caption{\label{alg:Valid-Tangent}Valid Tangent}
\end{algorithm}

It was shown, using geometry, in Lemma \ref{lem:valid_tangent} that
each $PE$-pair has only one valid tangent. Thus there are four valid
tangents (one corresponding to each $PE$-pair) with which the saddle-point
strategies of the pursuit-evasion game may coincide. Next we formulate
a matrix game at each instant of time to design feedback saddle-point
strategies for the pursuit-evasion game.

Recall that the clockwise circle $C_{p}(t)$ is traversed for $w_{p}(t)=-w_{p_{m}}$
and anticlockwise circle $A_{p}(t)$ for $w_{p}(t)=+w_{p_{m}}$. Similarly,
for the evader the clockwise circle $C_{e}(t)$ is traversed for $w_{e}(t)=-w_{e_{m}}$
and anticlockwise circle $A_{e}(t)$ for $w_{e}(t)=+w_{e_{m}}$. Selecting
$w_{p}(t)=+w_{p_{m}}$ and $w_{e}(t)=-w_{e_{m}}$ is equivalent to
selecting the valid tangent of the pair $\{A_{p}(t),C_{e}(t)\}$ along
which saddle-point strategies for the pursuit-evasion game will occur.
The computation of time to capture at time $t$, $T_{ac}(t)$, for
the valid tangent of the pair $\{A_{p}(t),C_{e}(t)\}$, shown in Figure
\ref{fig:ApCe}, is given in Algorithm \ref{alg:tangent_time}.
\begin{algorithm}
Input: Valid tangent for the pair $\{A_{p}(t),C_{e}(t)\}$

\rule[0.5ex]{1\columnwidth}{0.5pt}

Let $ET_{e}$ be the arc subtended between $O_{e}E$ and $O_{e}T_{e}$
in clockwise direction and let $PT_{p}$ be the arc subtended by $O_{p}P$
and $O_{p}T_{p}$ in anticlockwise direction as shown in Figure \ref{fig:ApCe}.
Compute the length of the arcs $PT_{p}=l_{ap}$ and $ET_{E}=l_{a_{e}}$
and define $t_{p}:=l_{ap}/v_{p_{m}}$ and $t_{e}:=l_{ae}/v_{e_{m}}$.
Also, let the distance between $T_{p}T_{e}$ be denoted by $d$.
\begin{enumerate}
\item If $t_{p}>t_{e}$ the evader will come out of the circle and onto
the tangent earlier than the evader.
\begin{enumerate}
\item $\tilde{t}:=t_{p}-t_{e}$. Thus the evader will travel a distance
$d_{e}:=v_{e_{m}}\tilde{t}$ on the straight line before the pursuer
comes onto the tangent.
\item Thus at the time $t_{p}$ the distance between the pursuer and evader
will be $\tilde{d}:=d+d_{e}$. Now the time to capture from this point
will be $\bar{t}:=\tilde{d}/(v_{p_{m}}-v_{e_{m}})$.
\item Thus the time to capture will be $T_{ac}(t)=t_{p}+\bar{t}$.
\end{enumerate}
\item If $t_{p}\leq t_{e}$ the pursuer will come on straight line earlier.
\begin{enumerate}
\item $\tilde{t}:=t_{e}-t_{p}$. Thus the pursuer will travel a distance
$d_{p}:=v_{p_{m}}\tilde{t}$ on the straight line before the pursuer
comes on the straight line.
\item Thus at the time $t_{e}$ the distance between the pursuer and evader
will be $\tilde{d}:=d-d_{p}$. Now the time to capture from this point
will be $\bar{t}:=\tilde{d}/(v_{p_{m}}-v_{e_{m}})$.
\item Thus the time to capture will be $T_{ac}(t)=t_{e}+\bar{t}$.
\end{enumerate}
\end{enumerate}
\rule[0.5ex]{1\columnwidth}{0.5pt}

Output: Time to capture along a valid tangent for $\{A_{p}(t),C_{e}(t)\}$
pair $=$ $T_{ac}(t)$ \caption{\label{alg:tangent_time}Algorithm to compute time to capture along
a valid tangent for $\{A_{p}(t),C_{e}(t)\}$ pair}
\end{algorithm}

\begin{figure*}
\begin{minipage}[t]{0.33\textwidth}%
\begin{center}
\includegraphics[scale=0.85]{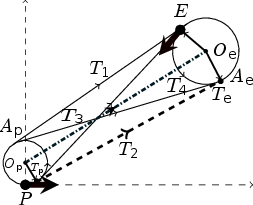}\caption{\label{fig:ApAe} $\{A_{p}(t),A_{e}(t)\}$}
\par\end{center}%
\end{minipage}\hfill{}%
\begin{minipage}[t]{0.33\textwidth}%
\begin{center}
\includegraphics[scale=0.85]{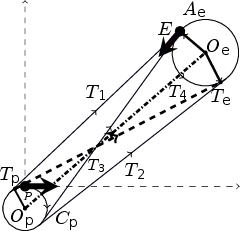}
\par\end{center}
\caption{\label{fig:CpAe}$\{C_{p}(t),A_{e}(t)\}$}
\end{minipage}\hfill{}%
\begin{minipage}[t]{0.33\textwidth}%
\begin{center}
\includegraphics[scale=0.85]{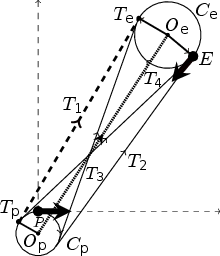}
\par\end{center}
\caption{\label{fig:CpCe} $\{C_{p}(t),C_{e}(t)\}$}
\end{minipage}
\end{figure*}

We calculate the times corresponding to each circle pairs and hence
each input pairs. At an given time instant, say $t$, let $A_{p}(t)$,
$C_{p}(t)$, and $A_{e}(t)$, $C_{e}(t)$ be the pursuer and evader
circles respectively. Let $T_{aa}(t)$, $T_{ac}(t)$, $T_{ca}(t)$
and $T_{cc}(t)$ be the times corresponding to valid tangents of circle-pairs
$\{A_{p}(t),A_{e}(t)\}$, $\{A_{p}(t),C_{e}(t)\}$, $\{C_{p}(t),C_{e}(t)\}$,
and $\{C_{p}(t),C_{e}(t)\}$ respectively. For example, if the pursuit-evasion
saddle-point occurs on the $PE$-pair $\{A_{p}(t),C_{e}(t)\}$ then
at $t$ we must have $w_{p}(t)=w_{p_{m}}$ and $w_{e}(t)=-w_{e_{m}}$
initially. Similarly we have, 
\begin{enumerate}
\item $\{A_{p}(t),A_{e}(t)\}\Rightarrow w_{p}(t)=+w_{p_{m}},\,w_{e}(t)=+w_{e_{m}}$ 
\item $\{A_{p}(t),C_{e}(t)\}$ $\Rightarrow$ $w_{p}(t)=+w_{p_{m}}$, $w_{e}(t)=-w_{e_{m}}$
\item $\{C_{p}(t),A_{e}(t)\}$ $\Rightarrow$ $w_{p}(t)=-w_{p_{m}}$, $w_{e}(t)=+w_{e_{m}}$
\item $\{C_{p}(t),C_{e}(t)\}$ $\Rightarrow$ $w_{p}(t)=-w_{p_{m}}$, $w_{e}(t)=-w_{e_{m}}$
\end{enumerate}
until the time that the trajectory leaves the corresponding circle
and starts on the straight line path along the common tangent. Thus
corresponding to pursuer and evader inputs at time $t$, we obtain
times of capture along each of the valid tangents. Using these times
we formulate a matrix game as shown in Table \ref{tab:matrixgame}.
The valid tangent on which the pursuit-evasion game occurs constitutes
the saddle point strategies for the pursuit-evasion differential game.
Thus, starting at time $t$ with the configuration ${\bf p}(t),\,{\bf e}(t)$,
the saddle-point solution of the matrix game at each instant $t$
will give the common tangent, with which the open-loop representation
of the feedback saddle-point strategies is coincident.
\begin{table}
\centering{}%
\begin{tabular}{c|c|c|}
\multicolumn{1}{c}{P\textbackslash E} & \multicolumn{1}{c}{$w_{e}(t)=+w_{e_{m}}$} & \multicolumn{1}{c}{$w_{e}(t)=-w_{e_{m}}$}\tabularnewline
\cline{2-3} \cline{3-3} 
$w_{p}(t)=+w_{p_{m}}$ & $T_{aa}(t)$ & $T_{ac}(t)$\tabularnewline
\cline{2-3} \cline{3-3} 
$w_{p}(t)=-w_{p_{m}}$ & $T_{ca}(t)$ & $T_{cc}(t)$\tabularnewline
\cline{2-3} \cline{3-3} 
\end{tabular}\caption{\label{tab:matrixgame}Matrix game at time instant $t$}
\end{table}
Thus the policy of the evader would be $\max-\min$ solution of the
matrix game while that of the evader would be $\min-\max$ solution
of the matrix game at each instant $t$. From this discussion the
next theorem follows.
\begin{thm}
\label{thm:matrix_law}If $d_{pe}^{0}\geq2r_{e}+2\pi r_{e}(v_{p_{m}}/v_{e_{m}})$,
the tangents selected by the saddle-point equilibrium in the matrix
game described by Table \ref{tab:matrixgame} will be coincident with
the open-loop representation of feedback saddle-point strategies at
each time $t\geq0$.
\end{thm}
If the solution is computed at each instant of time $t$ and the input
value corresponding to time instant $t$ i.e. $w_{p}(t)$ and $w_{e}(t)$
is applied then it constitutes a feedback law. Using Theorem \ref{thm:matrix_law},
such a feedback law $\text{F}({\bf x}(t)):{\bf x}(t)\in\mathbb{R}^{6}\rightarrow[{\bf u_{p}^{\top}}(t)\,\,{\bf u_{e}^{\top}}(t)]^{\top}\in\mathbb{R}^{4}$
can be computed and this provides solution for Problem\textcolor{blue}{{}
\ref{prob:implementable_feedback_law}}. For further details the interested
reader can refer to Algorithm 3 in \cite{aditya2019ECC} for details. 

\subsection{Numerical Simulations\label{sec:Numerical_Simulations}}

In \cite{Raivio2000}, an algorithm is proposed to solve pursuit-evasion
games numerically. This is achieved by first solving the $\min$ problem
of the pursuer and then the $\max$ problem of the evader. These optimal
problems are solved iteratively to obtain the $\min-\max$ solution
of the differential game. We use their algorithm to numerically solve
the game of two cars and the $\min$ and the $\max$ problems are
solved using direct numerical optimal control methods \cite{betts}
using IPOPT \cite{Wachter2006}. Clearly, such numerical techniques
are not practical for computing feedback solution in real time due
to time complexity and convergence issues of numerical optimization
methods. However, we do these simulations in order to verify the correctness
of the feedback law proposed in Theorem \ref{thm:matrix_law}.

The parameters for the pursuer and evader used for simulations are
$v_{p_{m}}=2$, $w_{p_{m}}=2$, $v_{e_{m}}=1$ and $w_{e_{m}}=1$.
The simulations were performed using numerical techniques (NT) from
\cite{Raivio2000} as well as the proposed feedback law (FL) for following
initial states of the pursuer and the evader: 
\begin{enumerate}
\item ${\bf {\bf p_{0}}}=[0,0,\pi/2],{\bf e_{0}}=[-3,-6,\pi/2]$ (ML: Figure
\ref{fig:Case1_TangentLaw}, NT: Figure \ref{fig:Case1_Numerical})
\item ${\bf p_{0}}=[0,0,\pi/2],{\bf e_{0}}=[6,3,\pi/2]$ (ML: Figure \ref{fig:Case2_TangentLaw},
NT: Figure \ref{fig:Case2_Numerical})
\item ${\bf p_{0}}=[0,0,\pi/2],{\bf e_{0}}=[0,-6,\pi/2]$ (ML: Figure \ref{fig:Case3_TangentLaw},
NT: Figure \ref{fig:Case3_Numerical})
\item ${\bf p_{0}}=[0,0,\pi/2],{\bf e_{0}}=[0,6,\pi/2+\pi/6]$ (ML: Figure
\ref{fig:Case4_TangentLaw}, NT: Figure \ref{fig:Case4_Numerical})
\end{enumerate}
In all the figures the pursuer trajectory is shown by the dashed curve
whereas the evader trajectory is shown by the dotted curve. The comparison
of matrix law and the numerical simulation show that the trajectories
are identical.

\begin{figure*}
\begin{centering}
\begin{minipage}[t]{0.45\textwidth}%
\begin{center}
\includegraphics[scale=0.3]{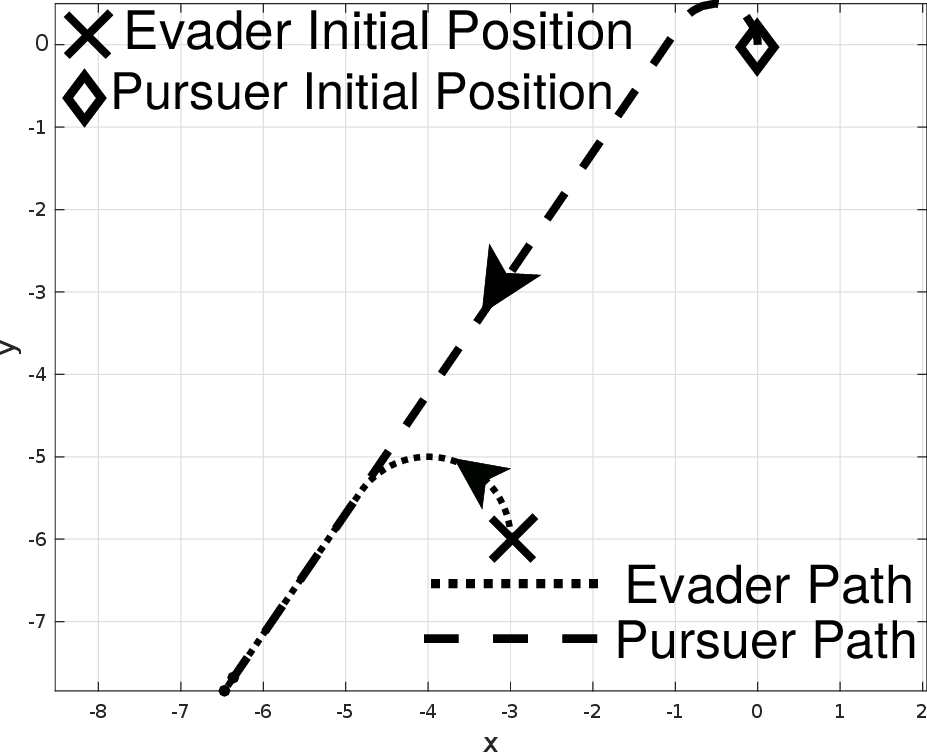}
\par\end{center}
\begin{center}
\caption{\label{fig:Case1_TangentLaw}Feedback law-1}
\par\end{center}%
\end{minipage}\hfill{}%
\begin{minipage}[t]{0.45\textwidth}%
\begin{center}
\includegraphics[scale=0.3]{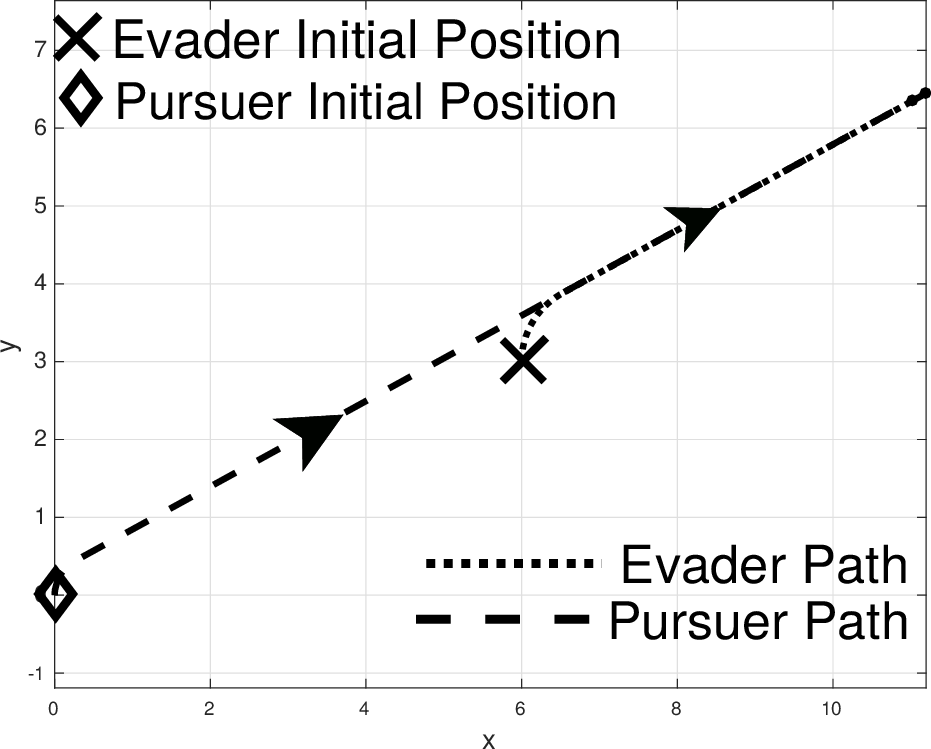}
\par\end{center}
\caption{\label{fig:Case2_TangentLaw}Feedback law-2}
\end{minipage}
\par\end{centering}
\centering{}%
\begin{minipage}[t]{0.45\textwidth}%
\begin{center}
\includegraphics[scale=0.3]{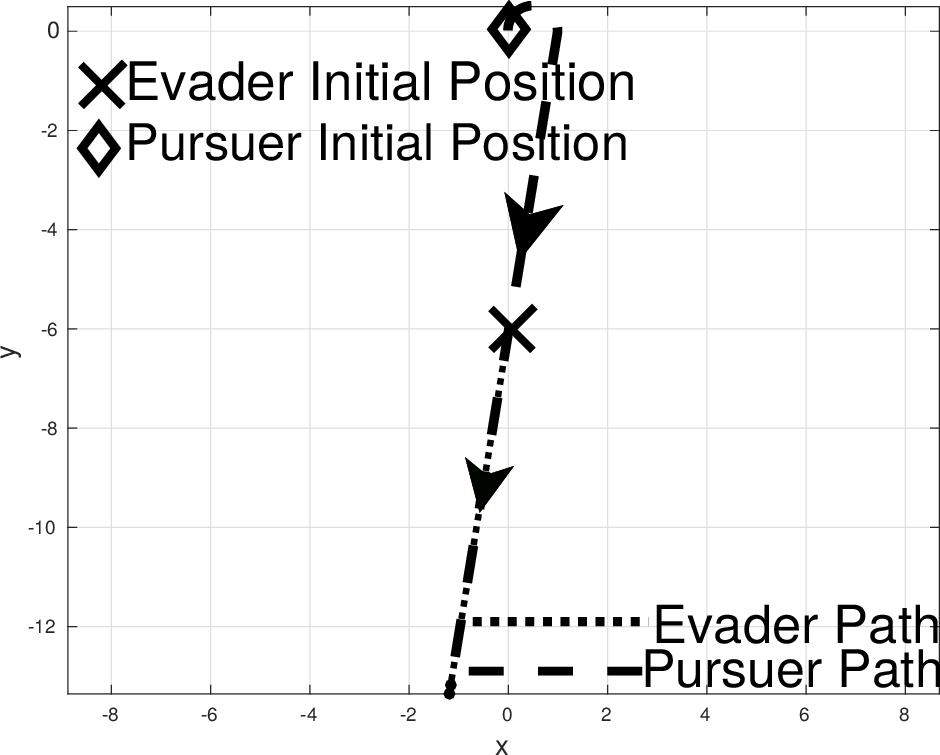}
\par\end{center}
\caption{\label{fig:Case3_TangentLaw}Feedback law-3}
\end{minipage}\hfill{}%
\begin{minipage}[t]{0.45\textwidth}%
\begin{center}
\includegraphics[scale=0.3]{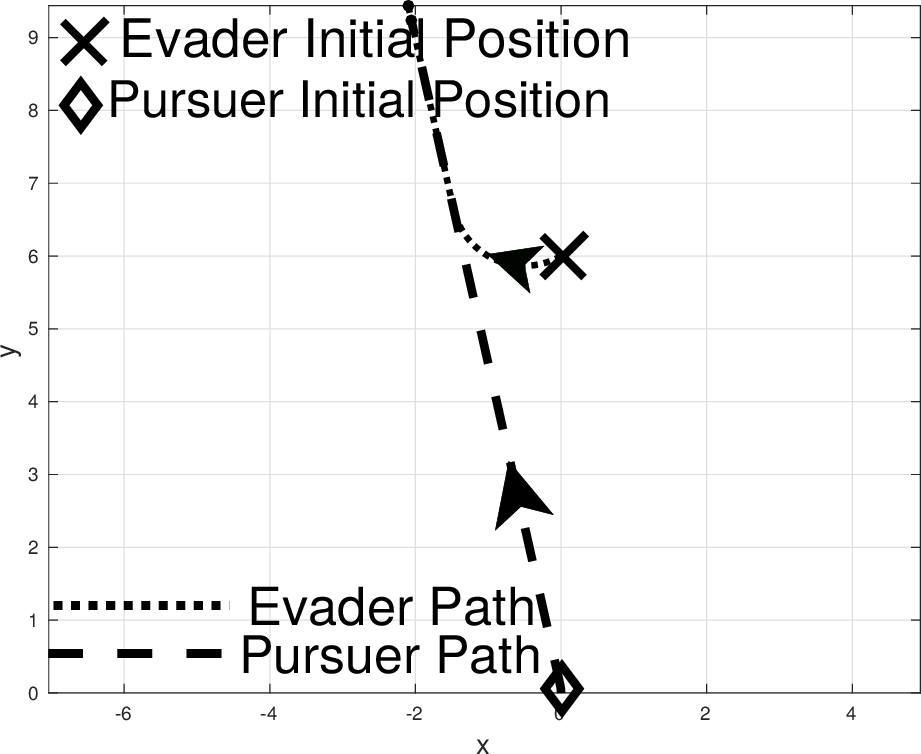}
\par\end{center}
\begin{center}
\caption{\label{fig:Case4_TangentLaw}Feedback law-4}
\par\end{center}%
\end{minipage}
\end{figure*}

\begin{figure*}
\begin{centering}
\begin{minipage}[t]{0.45\textwidth}%
\begin{center}
\includegraphics[scale=0.3]{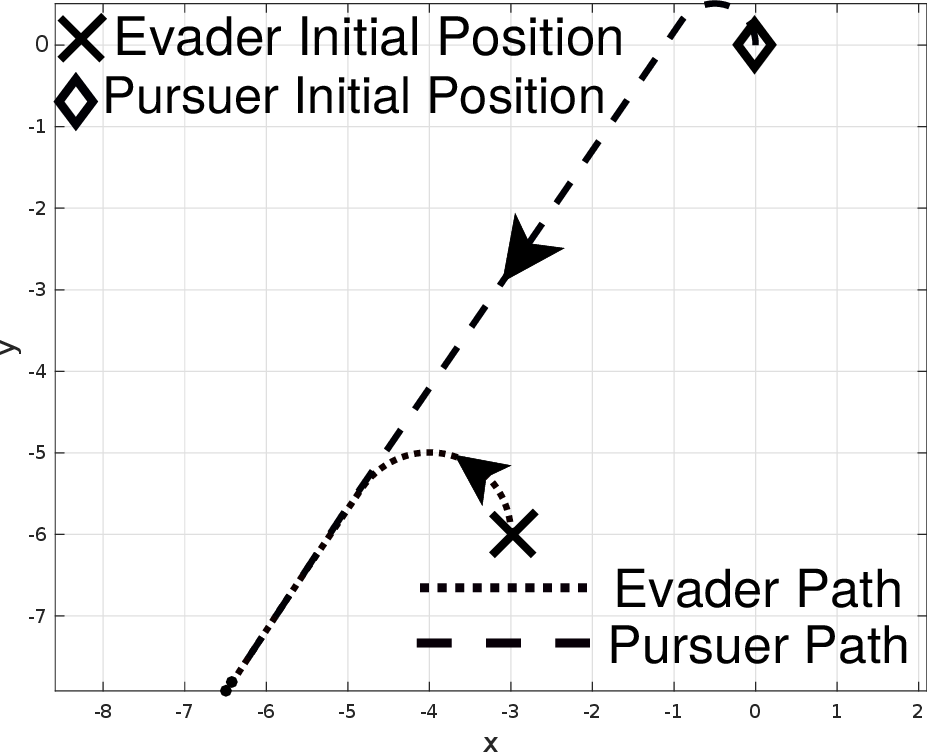}
\par\end{center}
\caption{\label{fig:Case1_Numerical}\cite{Raivio2000}Numerical -1}
\end{minipage}\hfill{}%
\begin{minipage}[t]{0.45\textwidth}%
\begin{center}
\includegraphics[scale=0.3]{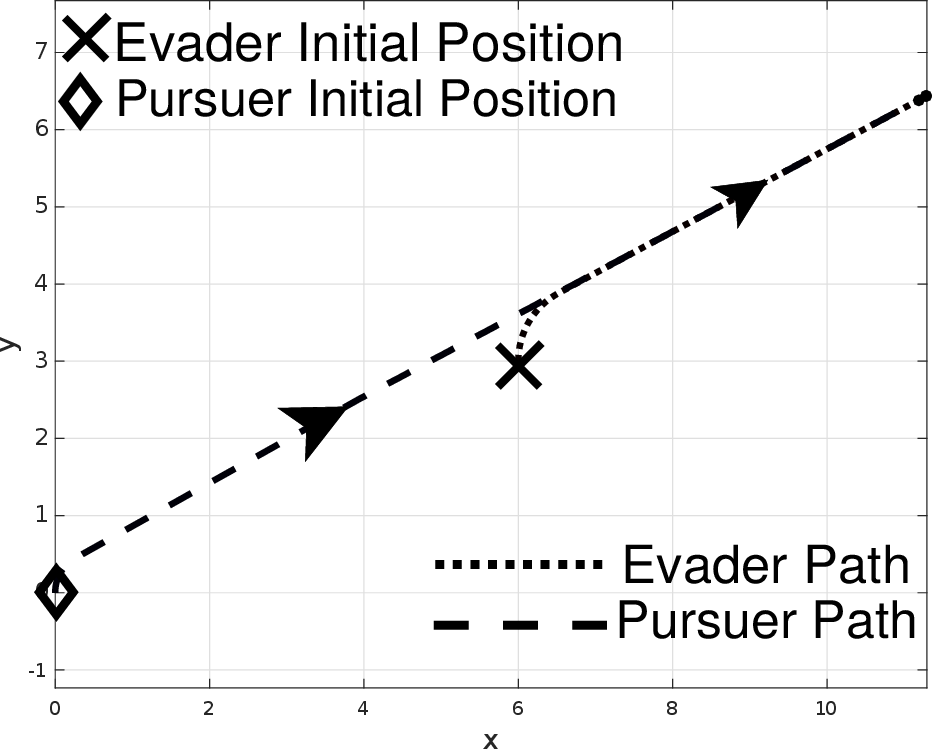}
\par\end{center}
\begin{center}
\caption{\label{fig:Case2_Numerical}\cite{Raivio2000}Numerical -2}
\par\end{center}%
\end{minipage}
\par\end{centering}
\centering{}%
\begin{minipage}[t]{0.45\textwidth}%
\begin{center}
\includegraphics[scale=0.3]{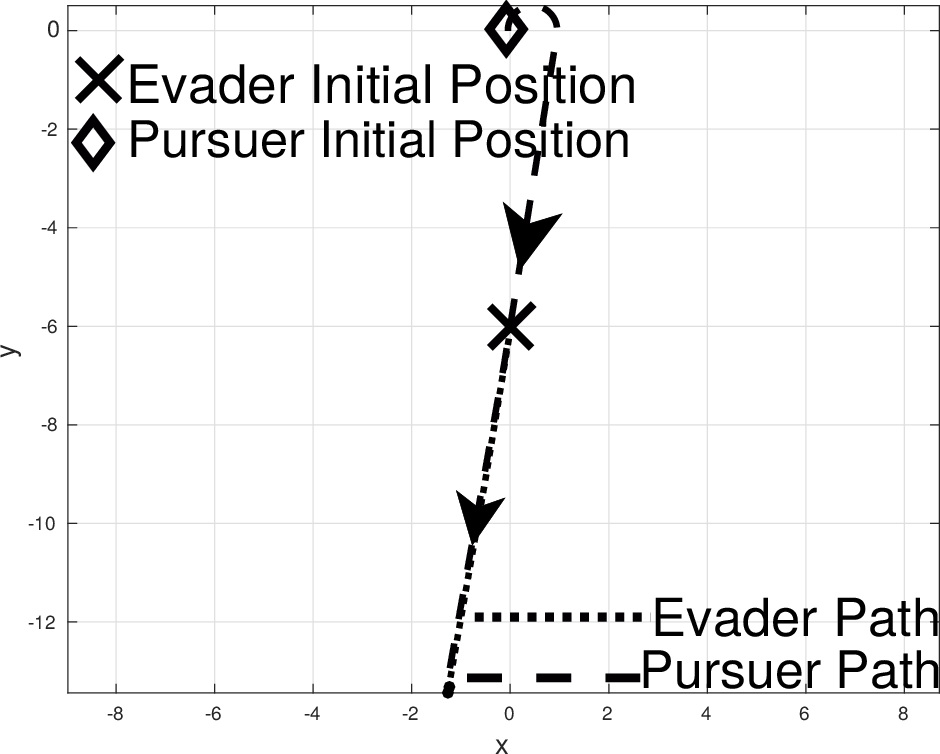}
\par\end{center}
\caption{\label{fig:Case3_Numerical}\cite{Raivio2000}Numerical -3}
\end{minipage}\hfill{}%
\begin{minipage}[t]{0.45\textwidth}%
\begin{center}
\includegraphics[scale=0.3]{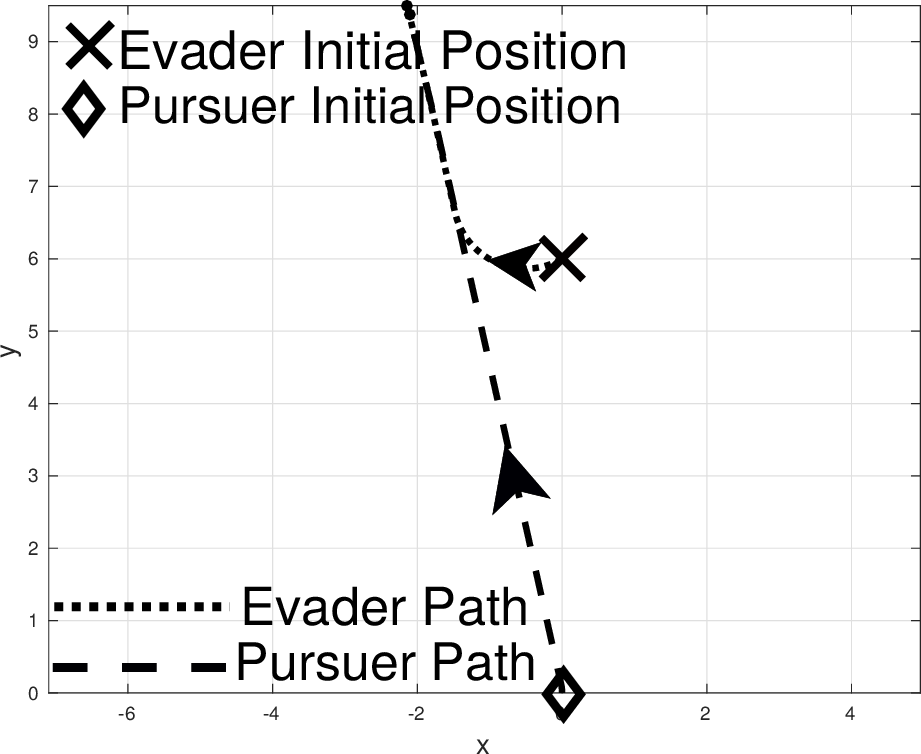}
\par\end{center}
\caption{\label{fig:Case4_Numerical}\cite{Raivio2000}Numerical -4}
\end{minipage}
\end{figure*}

\section{\label{sec:continuity_proofs}Proofs of continuity for specific subsets
of the reachable set of Dubins vehicle}

In this section we discuss the important continuous subsets of the
reachable set of Dubins vehicle necessary to characterize the saddle-point
trajectories for the game of two cars.
\begin{defn}
Let $\overline{{\bf xy}}$ be a continuum set. If a continuum of trajectories
for the curve $\overline{{\bf xy}}$ are of the type $LS$ then we
say that $\overline{{\bf xy}}$ is a \textbf{\textit{$LS$ continuum
set}}. $RS,CS$ continuum sets are defined analogously.
\end{defn}
\begin{lem}
\label{lem:curve_LS_not_continuous}Let $\bar{t}\geq2\pi r_{p}/v_{p_{m}}$
and consider Figure \ref{fig:Left-Reachable-Set}. Any curve $\overline{{\bf xy}}\in R_{p}^{l}({\bf p_{0}},\bar{t})$
such that $\overline{{\bf xy}}$ crosses the line $PA$, is not a
$LS$ continuum set.
\end{lem}
\begin{IEEEproof}
Recall that the switching time $t_{1}$ for $LS$ type trajectory
is less than $2\pi r_{p}/v_{p_{m}}$. The portion of the curve $\overline{{\bf xy}}$
on the left of $PA$ can be reached by $LS$ trajectories which have
switching time $t_{1}\leq\pi r_{p}/v_{p_{m}}$ while the portion on
right of $PA$ can be reached by curves with switching time $t_{1}\geq(3\pi/2)r_{p}/v_{p_{m}}$.
Hence, there do not exist $LS$ trajectories which can make $\overline{{\bf xy}}$
a continuum set. 
\end{IEEEproof}
\begin{lem}
\label{lem:curve_LS}Let $\bar{t}\geq2\pi r_{p}/v_{p_{m}}$ and consider
Figure \ref{fig:Continuum-of-Trajectories}. Any curve $\overline{{\bf xy}}\in R_{p}^{l}({\bf p_{0}},\bar{t})\backslash PA$
is a $LS$ continuum set.
\end{lem}
\begin{IEEEproof}
Recall that the switching time $t_{1}$ for a $LS$ trajectory is
less than $2\pi r_{p}/v_{p_{m}}$. Consider a curve $\overline{{\bf xy}}$
in the reachable set of the kinematic point, between points ${\bf x},{\bf y}\in R_{p}^{l}({\bf p_{0}},\bar{t})$.
Let ${\bf p_{x}}$ be a time optimal trajectory given by (\ref{eq:left_reachable_set})
and ${\bf x}={\bf p_{x}}_{|\mathbb{R}^{2}}(t)=[x_{f}(t)\ y_{f}(t)]^{\top}$
for some $t\leq\bar{t}$. Let ${\bf x'}$ be a point in neighborhood
of ${\bf x}$ such that ${\bf x'}={\bf x}+\delta{\bf x}$, ${\bf x'}\in R_{p}^{l}({\bf p_{0}},\bar{t})$
and ${\bf x'}={\bf p'_{x}}_{|\mathbb{R}^{2}}(t')=[x_{f}(t')\ y_{f}(t')]^{\top}$
for some $t'\leq\bar{t}$. Now define a function say ${\bf \bar{g}}:\mathbb{R}^{2}\rightarrow\mathbb{R}^{2}$
given by (\ref{eq:left_reachable_set}). The domain of the function
is $t\times\theta\in\mathbb{R}^{2}$ where, $t\in[0,\bar{t}]$ and
$\theta\in[0,2\pi]$ and the range is the set of points in the reachable
set $R_{p}^{l}({\bf p_{0}},\bar{t})$. Now, the Jacobian of the function
is,
\[
J=\left[\begin{array}{cc}
v_{p_{m}}\cos\tilde{\theta} & \ \ -v_{p_{m}}^{2}w_{p_{m}}(t-t_{1})\sin\tilde{\theta}\\
v_{p_{m}}\sin\tilde{\theta} & \ \ v_{p_{m}}^{2}w_{p_{m}}(t-t_{1})\cos\tilde{\theta}
\end{array}\right]
\]
Since (\ref{eq:left_reachable_set}) is defined such that $t>t_{1}>0$,
$J$ is non-singular. Hence, by the inverse function theorem there
exists an inverse map ${\bf \bar{g}}^{-1}:\mathbb{R}^{2}\rightarrow\mathbb{R}^{2}$
such that ${\bf \bar{g}}^{-1}$ is continuous. Thus, as $||\delta{\bf x}||_{2}\rightarrow0$
we have $||[\delta t\ \delta\theta]'||_{2}\rightarrow0$ and thus
$||{\bf p_{x}}_{|\mathbb{R}^{2}}-{\bf p'}_{|\mathbb{R}^{2}}||_{\mathcal{L_{\infty}}}\rightarrow0$.
Thus, for any two nearby points the trajectories will also be nearby. 

Define a function $\mathfrak{f}$ such that it matches each point
${\bf x}$ to the time-optimal trajectory given by (\ref{eq:left_reachable_set})
and ${\bf x}={\bf p_{x}}_{|\mathbb{R}^{2}}(t_{a})=[x_{fl}(t_{a})\ y_{fl}(t_{a})]$
for some $t_{a}\leq\bar{t}$. Thus the function $\mathfrak{f}$ would
be a one-one function by definition. Also, the function is also continuous
by the analysis presented above ($||\delta{\bf x}||_{2}\rightarrow0\Rightarrow||{\bf p_{x}}_{|\mathbb{R}^{2}}-{\bf p'}_{|\mathbb{R}^{2}}||_{\mathcal{L_{\infty}}}\rightarrow0$).
Thus $\overline{{\bf xy}}$ is a continuum set.
\end{IEEEproof}
The following lemma is analogous to Lemma \ref{lem:curve_LS}.
\begin{lem}
\label{lem:curve_RS}Let $\bar{t}\geq2\pi r_{p}/v_{p_{m}}$. Any curve
$\overline{{\bf xy}}\in R_{p}^{r}({\bf p_{0}},\bar{t})\backslash PA$
is a $RS$ continuum set.
\end{lem}
Next, we define a series of special continuous subsets of reachable
sets, which will be useful to characterize the saddle point capture
condition in Section \ref{subsec:lcse} and Section \ref{subsec:ccsp}.
\begin{defn}
Let $\bar{t}\geq2\pi r_{p}/v_{p_{m}}$. Then, the \textbf{\textit{central
reachable set}} of the pursuer (see Figure \ref{fig:Central-Reachable-Region})
is defined as $R_{p}^{CR}({\bf p_{0}},\bar{t})=R_{p}({\bf p_{0}},\bar{t})\backslash\{\text{int}(C_{p}(0))\cup\text{int}(A_{p}(0))\cup\overline{PB}\}$,
where $\text{int}(C_{p}(0))$ and $\text{int}(A_{p}(0))$ denote the
interiors of pursuer circles at time $t=0$ and $\overline{PB}=\{PB\backslash\{P\}\}$.

The central reachable set of the evader $R_{e}^{CR}({\bf e_{0}},\bar{t})$
is defined analogously for $\bar{t}\geq2\pi r_{e}/v_{e_{m}}$. The
central reachable set is shown in Figure \ref{fig:Central-Reachable-Region}.
The points on the line segment $\overline{PB}$ and the area inside
the minimum turning radius circles do not form a part of central reachable
set.
\end{defn}
\begin{lem}
\label{lem:CRS_continuous}The central reachable set of pursuer/evader
is a continuous subset for all time $t\geq2\pi\frac{r_{p}}{v_{p_{m}}}$
($t\geq2\pi\frac{r_{e}}{v_{e_{m}}}$ ). 
\end{lem}
\begin{IEEEproof}
Consider Figure \ref{fig:Central-Reachable-Region}. Let $D_{l}$
denote the set of points enclosed by the curve $DABPD$ and $D_{r}$
be the set of points enclosed by the curve $DCBPD$. Any curve contained
completely in the set $D_{l}\backslash\left(\overline{PB}\cup\text{int}(A_{p}(0))\right)\subseteq R_{p}^{l}({\bf p_{0}},\bar{t})$
is a continuum set by Lemma \ref{lem:curve_LS}. Also, any curve contained
completely in the set $D_{r}\backslash\left(\overline{PB}\cup\text{int}(A_{p}(0))\right)$\\$\subseteq R_{p}^{r}({\bf p_{0}},\bar{t})$
is a continuum set by Lemma \ref{lem:curve_RS}. Thus we need to consider
only those type of curves which cross line $DP$ such as the curve
${\bf x-z}-{\bf y}$ shown in Figure \ref{fig:Central-Reachable-Region}.
The curve $\overline{{\bf xz}}$ in Figure \ref{fig:Central-Reachable-Region}
is a $LS$ continuum set. Let this continuum set of trajectories be
denoted by ${\bf T^{l}}$. Similarly, any curve $\overline{{\bf zy}}$
in the subset $DCBPD\subseteq R_{p}^{r}({\bf p_{0}},\bar{t})$ as
shown in Figure \ref{fig:Central-Reachable-Region} is a continuous
subset with $RS$ type of trajectories. Let this continuum of trajectories
be denoted by ${\bf T^{r}}$. Now, the trajectory in ${\bf T^{l}}$
and ${\bf T^{r}}$ which reaches the point ${\bf z}$ is the same
trajectory along the straight line $PD$. Thus, ${\bf T^{l}}\cup{\bf T^{r}}$
forms a continuum of trajectories for the curve ${\bf x-z}-{\bf y}$
and hence it is a continuum set. Thus any curve in the central reachable
set is a continuum set with $CS$ type of trajectories. Further, the
continuum of trajectories ${\bf T^{l}}\cup{\bf T^{r}}$ is completely
contained in $R_{p}^{CR}({\bf p_{0}},\bar{t})$. Hence, the central
reachable set it is a continuous subset.
\end{IEEEproof}
\begin{figure*}
\begin{minipage}[t]{0.3\textwidth}%
\begin{center}
\includegraphics[scale=0.17]{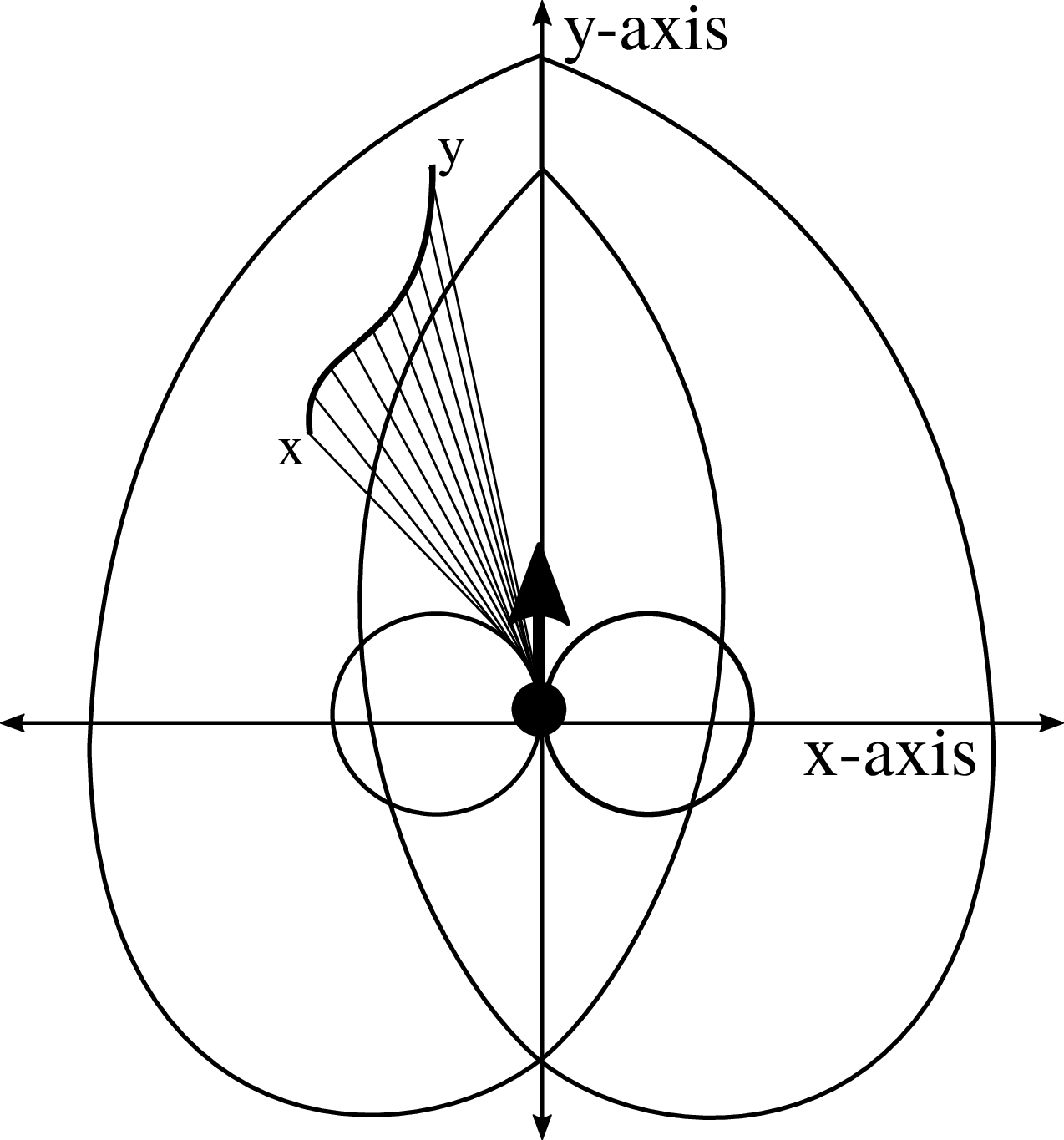}
\par\end{center}
\caption{\label{fig:Continuum-of-Trajectories}Continuum of trajectories}
\end{minipage}\hfill{}%
\begin{minipage}[t]{0.3\textwidth}%
\begin{center}
\includegraphics[scale=0.17]{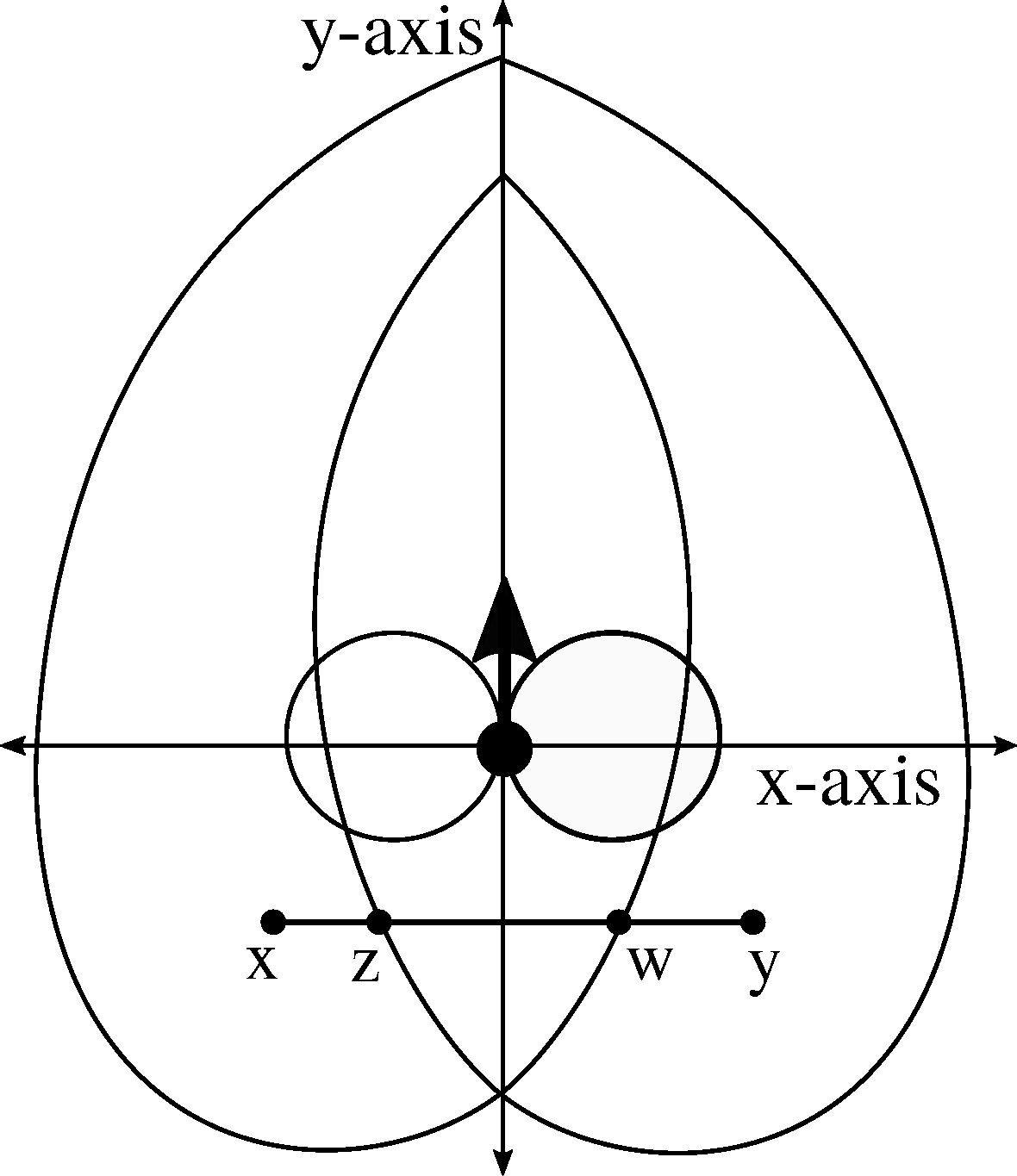}
\par\end{center}
\caption{\label{fig:unoin_not_continuous}$R_{p}^{r}({\bf p_{0}},\bar{t})\cup R_{p}^{l}({\bf p_{0}},\bar{t})$
is not a $CSP$.}
\end{minipage}\hfill{}%
\begin{minipage}[t]{0.3\textwidth}%
\begin{center}
\includegraphics[scale=0.17]{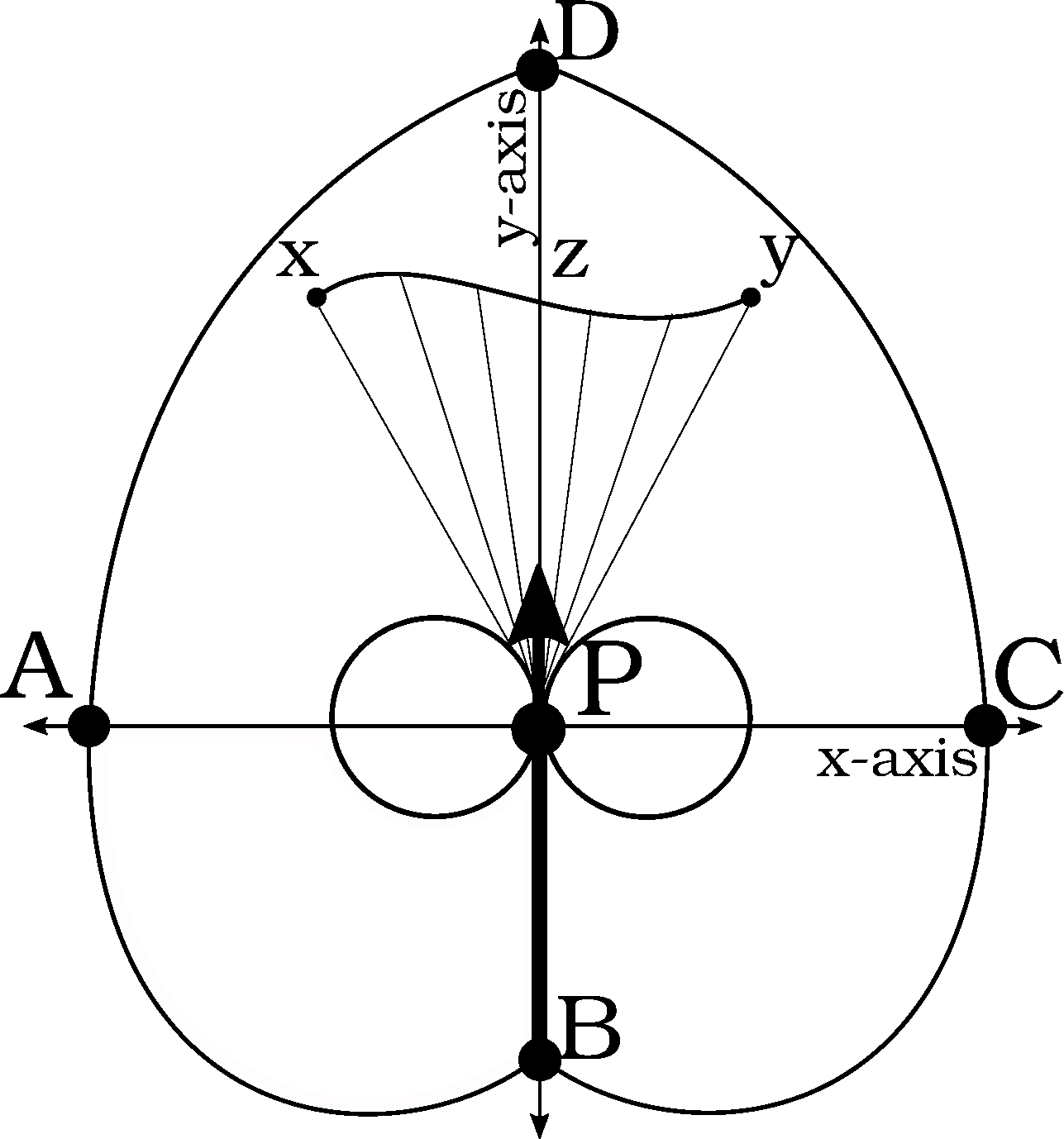}
\par\end{center}
\caption{\label{fig:Central-Reachable-Region}Central reachable set is a continuous
subset}
\end{minipage}
\end{figure*}

\begin{defn}
\label{def:truncated_sets}Consider Figures \ref{fig:Left-Reachable-Set}
and \ref{fig:Right-Reachable-Set} and let $\bar{t}\geq2\pi r_{p}/v_{p_{m}}$.
The \textbf{\textit{truncated left reachable set}} is defined as $R_{p}^{l_{t}}({\bf p_{0}},\bar{t}):=R_{p}^{l}({\bf p_{0}},\bar{t})\backslash PA$.
Similarly, the \textbf{\textit{truncated right reachable set}} is
defined as $R_{p}^{r_{t}}({\bf p_{0}},\bar{t})=R_{p}^{r}({\bf p_{0}},\bar{t})\backslash PA$. 
\end{defn}
The proofs of Lemma \ref{lem:truncated_left_continuous} and Lemma
\ref{lem:truncated_right_continuous} follow from Lemma \ref{lem:curve_LS}
and Lemma \ref{lem:curve_RS}. 
\begin{lem}
\label{lem:truncated_left_continuous}Pursuer's (Evader's) truncated
left reachable set $R_{p}^{l_{t}}({\bf p_{0}},\bar{t})$ ($R_{e}^{l_{t}}({\bf e_{0}},\bar{t})$)
is a continuous subset for all $t\geq2\pi\frac{r_{p}}{v_{p_{m}}}$
($t\geq2\pi\frac{r_{e}}{v_{e_{m}}}$ ).
\end{lem}
\begin{lem}
\label{lem:truncated_right_continuous}Pursuer's (Evader's) truncated
right reachable set $R_{p}^{r_{t}}({\bf p_{0}},\bar{t})$ ($R_{e}^{r_{t}}({\bf e_{0}},\bar{t})$)
is a continuous subset for all $t\geq2\pi\frac{r_{p}}{v_{p_{m}}}$
($t\geq2\pi\frac{r_{e}}{v_{e_{m}}}$ ).
\end{lem}
\begin{defn}
\textit{\label{def:blocking_set}}\textbf{\textit{Blocking set $B_{p}({\bf p_{0}},{\bf e_{0}},\bar{t})$}}:Consider
Figure \ref{fig:construction_blocking_1} and $\bar{t}\geq2\pi r_{p}/v_{p_{m}}$.
Construct a line segment $EP$, joining the initial position of the
evader ($E$) to the initial position of the pursuer (P). Draw tangents
to the pursuer circles parallel to $EP$ directed from the evader
to pursuer. Let the tangents be called $T_{l}^{1},T_{l}^{2},T_{r}^{1},T_{r}^{2}$.
Only one tangent in the pair $\{T_{l}^{1},T_{l}^{2}\}$ has same orientation
as the left pursuer circle and we denote it by $T_{l}^{v}$. Let $\tilde{T}_{l}^{v}$
be the curve obtained by concatenation of the arc $PA$ and $T_{l}^{v}$
as shown in Figure \ref{fig:construction_blocking_2}. Similarly,
the tangent in $\{T_{r}^{1},T_{r}^{2}\}$, having same orientation
as the right pursuer circle, will be denoted by $T_{r}^{v}$. Also,
let $\tilde{T}_{r}^{v}$ be the curve obtained by concatenation of
the arc $PC$ and $T_{r}^{v}$ as shown in Figure \ref{fig:construction_blocking_2}.
The curves $\tilde{T}_{l}^{v}$ is the curve of type $LS$ while $\tilde{T}_{r}^{v}$
is the curve of type $RS$ up to the time $\bar{t}$. Hence, $\tilde{T}_{r}^{v}$
ends at point $D$ while $\tilde{T}_{l}^{v}$ ends at point $B$.
The portion of $R_{p}({\bf p_{0}},\bar{t})$ shaded by sloped lines
between $\tilde{T}_{l}^{v}$ and $\tilde{T}_{r}^{v}$, containing
the initial position of the evader (marked by $E$ in Figure \ref{fig:construction_blocking_2}),
is defined to be the blocking set $B_{p}({\bf p_{0}},{\bf e_{0}},\bar{t})$. 
\end{defn}
The shaded region in Figure \ref{fig:blocking_evader_behind} shows
the blocking set when evader is behind the pursuer on the right side
of the pursuer. Similarly, Figure \ref{fig:blocking_evader_front}
shows the blocking set when evader is in front of the pursuer. 
\begin{lem}
\label{lem:blocking_continuous}$B_{p}({\bf p_{0}},{\bf e_{0}},\bar{t})$
is a $CSP$ for $\bar{t}\geq2\pi\frac{r_{e}}{v_{e_{m}}}$.
\end{lem}
\begin{IEEEproof}
$B_{p}({\bf p_{0}},{\bf e_{0}},\bar{t})$ is an union of two parts,
one which forms the part of left reachable set say part $A$ and another
which forms part of right reachable set say part $B$. Thus any curve
$\overline{{\bf xy}}$ in $B_{p}({\bf p_{0}},{\bf e_{0}},\bar{t})$
can be divided into two parts one which belongs to part $A$ and the
other which belongs to part $B$. Thus by using arguments similar
to that in Lemma \ref{lem:CRS_continuous}, we can show that the curve
is a $CS$ continuum set. 
\end{IEEEproof}
\begin{lem}
\label{lem:union_not_continuous}The union of left reachable set and
right reachable set is not a continuous subset.
\end{lem}
\begin{IEEEproof}
Consider the curve $\overline{{\bf xy}}$ in Figure \ref{fig:unoin_not_continuous}.
All the trajectories of the type $CS$, which reach points on the
segment $\overline{{\bf xz}}$ must start on the left circle of the
pursuer. Similarly, all the trajectories of the type $CS$ which reach
points on the segment $\overline{{\bf yw}}$ must start on the right
circle of the pursuer. Hence the line $\overline{{\bf xy}}$ is not
a continuum set.
\end{IEEEproof}

\section{\label{sec:imp_thm_1}Proof of Theorem \ref{thm:capture_strategy_pursuer_evader} }

In this section we will prove the Theorem \ref{thm:capture_strategy_pursuer_evader}.
We do this by first proving two intermediate theorems: Theorem \ref{thm:capture_sufficient_condition}
and \ref{thm:capture_necessary_condition}. Theorem \ref{thm:capture_sufficient_condition}
proves that there exists a feedback policy for pursuer which leads
to capture in time $t\leq T$ irrespective of the evasion policy chosen
by the evader.
\begin{thm}
\label{thm:capture_sufficient_condition}Let $T$ be as given by (\ref{eq:feedback_time_to_capture}).
For every trajectory ${\bf e}$ of the evader corresponding to feedback
policy ${\bf u_{e}}\in\mathcal{U}_{e}$ there exists a pursuer trajectory
corresponding to feedback strategy ${\bf u_{p}}\in\mathcal{U}_{p}$
s.t. ${\bf e}_{|\mathbb{R}^{2}}(t)={\bf p}_{|\mathbb{R}^{2}}(t)$
for some $t\leq T$.
\end{thm}
\begin{IEEEproof}
Since, for all $\,R_{e}^{c}({\bf e_{0}},T)\in\mathfrak{R}_{e}({\bf e_{0}},T),\,\text{there exists an}\,R_{p}^{c}({\bf p_{0}},T)\in\mathfrak{R}_{p}({\bf p_{0}},T)$
such that $R_{e}^{c-}({\bf e_{0}},R_{p}^{c}$$,T)$$=$$\emptyset$,
regardless of the input strategy ${\bf u_{e}}\in\mathcal{U}_{e}$
that the evader selects, there exists an open-loop pursuer strategy
${\bf u_{p}}\in\mathcal{U}_{p}$ such that the evader will be captured
at a point say ${\bf c}\in\mathbb{R}^{2}$. However, if the evader
is allowed feedback policies, it can execute deviations about ${\bf u_{e}}$
continuously by using the correct information about the pursuit trajectory
${\bf p}$ chosen by the pursuer. Let the evader input deviate by
$\delta{\bf u_{e}}$ and the evader trajectory by ${\bf \delta e}$.
Let the new evader trajectory ${\bf e}'={\bf e}+{\bf \delta e}$ be
such that ${\bf e}'_{|\mathbb{R}^{2}}(T)={\bf c}'\in\mathbb{R}^{2}$.
Let $\overline{{\bf cc'}}$ be any curve between points ${\bf c}$
and ${\bf c'}$ such that $\overline{{\bf cc'}}\in R_{e}^{c}({\bf e_{0}},T)$
for some $CSE$ $R_{e}^{c}({\bf e_{0}},T)$. Since the curve $\overline{{\bf cc'}}$
is in some $CSE$ $R_{e}^{c}({\bf e_{0}},T)$, and $R_{e}^{c-}({\bf e_{0}},R_{p}^{c},T)=\emptyset$
we have $\overline{{\bf cc'}}\in R_{p}^{c}({\bf p_{0}},T)$ for some
$CSP$ $R_{p}^{c}({\bf p_{0}},T)$. Thus, there exists a pursuer input
${\bf u'_{p}}={\bf u_{p}}+\delta{\bf u_{p}}$ such that the ${\bf p}'_{|\mathbb{R}^{2}}(T)={\bf c'}$
for the trajectory ${\bf p}'$ corresponding to the input ${\bf u'_{p}}$
and some time $t\leq T$.
\end{IEEEproof}
Next we show that there exists an evasion policy which can ensure
that even with the best pursuit policy, capture happens only at $T$
and no sooner.

\begin{figure}

\begin{centering}
\includegraphics[scale=0.25]{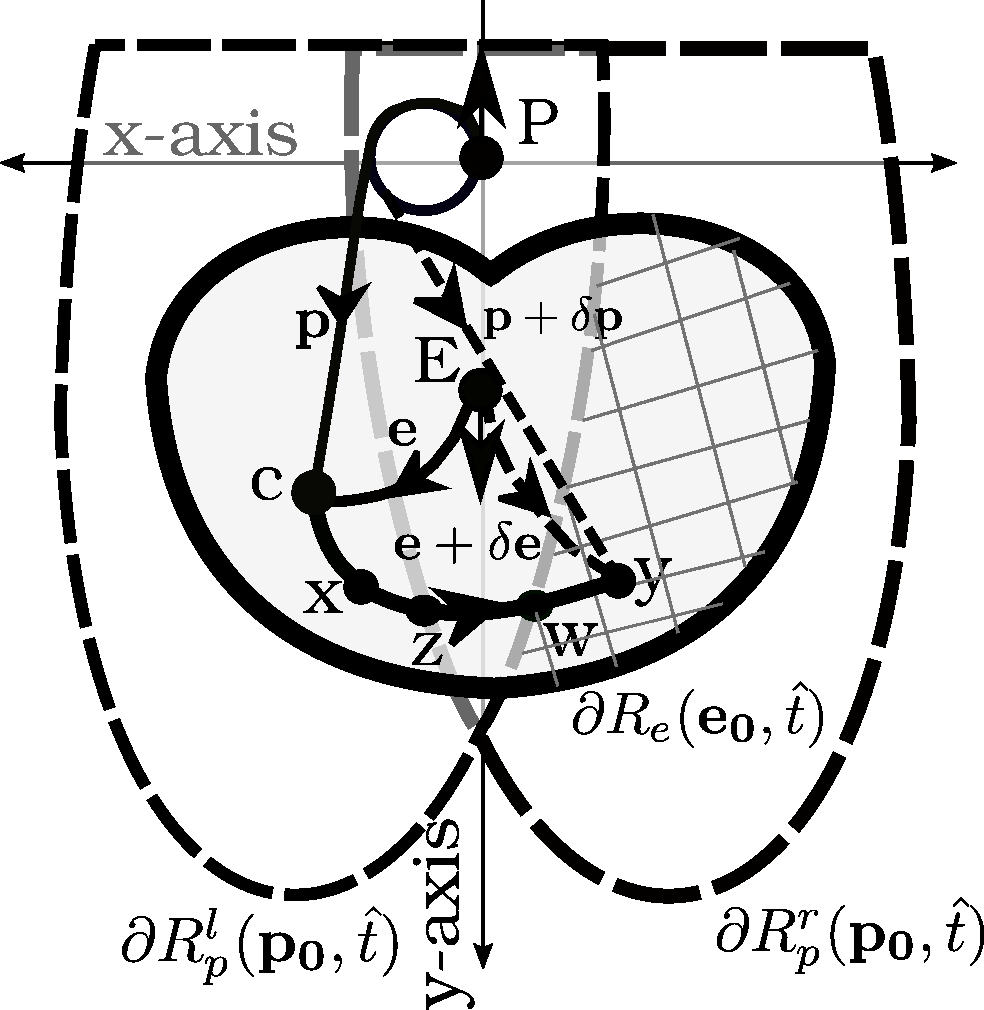}
\par\end{centering}
\caption{\label{fig:no_capture}Only Containment: Pursuer cannot follow all
the variations of the evader}

\end{figure}

\begin{defn}
Let $\mathcal{L}_{e}({\bf e_{0}},T)$ be a $CSE$. Let $\mathcal{L}_{e}^{-}({\bf e_{0}},R_{p}^{c},T-\delta)\neq\emptyset$
for all $\delta>0$ and for all $R_{p}^{c}({\bf p_{0}},T-\delta)\in\mathfrak{R}_{p}$.
Let $\mathcal{C}_{p}({\bf p_{0}},T)$ be a $CSP$ such that $\mathcal{L}_{e}^{-}({\bf e_{0}},\mathcal{C}_{p},T)=\emptyset$.
We call the $\mathcal{L}_{e}({\bf e_{0}},T)$ the \textbf{last $CSE$}
($LCSE$) and $\mathcal{C}_{p}({\bf p_{0}},T)$ as \textbf{capturing}
$CSP$ ($CCSP$). Note that a pair of a $LCSE$ and a $CCSP$ is not
unique.
\end{defn}
\begin{thm}
\label{thm:capture_necessary_condition}Let $T$ be as defined by
(\ref{eq:feedback_time_to_capture}). For every pursuer trajectory
${\bf p}$ corresponding to feedback policy ${\bf u_{p}}\in\mathcal{U}_{p}$
there exists an evader trajectory ${\bf e}$ corresponding to feedback
strategy ${\bf u_{e}}\in\mathcal{U}_{e}$ s.t. ${\bf e}_{|\mathbb{R}^{2}}(t)\neq{\bf p}_{|\mathbb{R}^{2}}(t)$
for all $t<T$
\end{thm}
\begin{IEEEproof}
This situation can be divided into two cases: 

\noindent \textbf{\textit{Case 1}}: Consider a time instant $\tilde{t}<T$
such that $\mathcal{L}_{e}^{-}({\bf e_{0}},R_{p}^{c},\tilde{t})\neq\emptyset$
for all $R_{p}^{c}({\bf p_{0}},\tilde{t})\in\mathfrak{R}_{p}$ and
$R_{e}^{-}({\bf e_{0}},R_{p},\tilde{t})\neq\emptyset$ for $R_{p}({\bf p_{0}},\tilde{t})$.
Since $R_{e}^{-}({\bf e_{0}},R_{p},\tilde{t})\neq\emptyset$, there
exists an evader strategy ${\bf u_{e}}\in\mathcal{U}_{e}$ such that
${\bf e}_{|\mathbb{R}^{2}}(t)\notin R_{p}({\bf p}_{0},t)$ for each
$t\leq\tilde{t}$. Thus the evader can always escape capture at each
$t\leq\tilde{t}$. 

\noindent \textbf{\textit{Case 2}}: Consider a time instant $\hat{t}<T$
such that $\mathcal{L}_{e}^{-}({\bf e_{0}},R_{p}^{c},\hat{t})\neq\emptyset$
for all $R_{p}^{c}({\bf p_{0}},\hat{t})\in\mathfrak{R}_{p}$ but $R_{e}^{-}({\bf e_{0}},R_{p},\hat{t})=\emptyset$.
Since $R_{e}^{-}({\bf e_{0}},R_{p},\hat{t})=\emptyset$, every trajectory
${\bf e}$ of the evader, corresponding to open loop input ${\bf u_{e}}$,
can be intercepted by some pursuer strategy ${\bf p}$ corresponding
to a suitable pursuer input ${\bf u_{p}}$. This is shown in Figure
\ref{fig:no_capture} where the solid curve (labeled ${\bf e}$) denotes
an evader trajectory ${\bf e}_{|\mathbb{R}^{2}}$ and the solid curve
(labeled ${\bf p}$) denotes the pursuer trajectory ${\bf p}_{|\mathbb{R}^{2}}$.
Let the pursuer strategy ${\bf p}$ intercept the evader, using the
input ${\bf u_{e}}$, at a point ${\bf c}\in\mathcal{L}_{e}({\bf e_{0}},\hat{t})$.
Since, $\mathcal{L}_{e}^{-}({\bf e_{0}},R_{p}^{c},\hat{t})\neq\emptyset$
$\forall$ $R_{p}^{c}({\bf p_{0}},\hat{t})\in\mathfrak{R}_{p}$, $\exists$
a curve $\overline{{\bf x}{\bf y}}\in\mathcal{L}_{e}({\bf e_{0}},\hat{t})$
(depicted by the solid curve from ${\bf x}$ to ${\bf y}$ in Figure
\ref{fig:no_capture} ) which is not a continuum set for the pursuer.
Extend the curve $\overline{{\bf xy}}$ from point ${\bf x}$ to ${\bf c}\in\mathcal{L}_{e}({\bf e_{0}},\hat{t})$
by a curve $\overline{{\bf xc}}\in\mathcal{L}_{e}({\bf e_{0}},\hat{t})$.
Call the new curve $\overline{{\bf cy}}$. Since the curve ${\bf \overline{xy}}\in\mathcal{L}_{e}^{-}({\bf e_{0}},R_{p}^{c},\hat{t})$
is not a continuum set for the pursuer, the curve ${\bf \overline{cy}}$
is also not a continuum set for the pursuer. However, the curve ${\bf \overline{cy}}\in\mathcal{L}_{e}^{-}({\bf e_{0}},R_{p}^{c},\hat{t})$
is a continuum set for the evader. Thus the evader with a variation
$\delta{\bf u_{e}}$ in its strategy can change the end point to ${\bf y}$
from point ${\bf c}$ (The evader trajectory corresponding to input
${\bf u_{e}+\delta u_{e}}$ is shown by a dashed curve and labeled
${\bf e}+\delta{\bf e}$ ). However, since ${\bf \overline{cy}}$
is not a continuum set for the pursuer, the pursuer cannot catch the
evader in time $\hat{t}$ using an admissible variation $\delta{\bf u_{p}}$.
Thus the evader can always avoid capture for all time $t\leq\hat{t}<T$. 
\end{IEEEproof}
\noindent \textbf{Proof of Theorem \ref{thm:capture_strategy_pursuer_evader}:}
By Theorem \ref{thm:capture_necessary_condition}, in order to maximize
the time-to-capture the evader should use input strategy ${\bf u_{e}^{*}}\in\mathcal{U}_{e}$
such that, ${\bf e}_{|\mathbb{R}^{2}}(T)={\bf z}\in\mathcal{L}_{e}({\bf e_{0}},T)$,
and ${\bf e}_{|\mathbb{R}^{2}}(t)\notin\mathcal{C}_{p}({\bf p_{0}},t)$
$\forall$ $t<T$. On the other hand, the pursuer should use time-optimal
input strategy ${\bf u_{p}^{*}}\in\mathcal{U}_{p}$ such that ${\bf p}_{|\mathbb{R}^{2}}(T)={\bf z}$
and ${\bf p}_{|\mathbb{R}^{2}}(t)\in\mathcal{C}_{p}({\bf p_{0}},T)$
$\forall\ t\leq T$ in order to minimize the time to capture the evader.
Thus if the evader uses ${\bf u_{e}^{*}}\in\mathcal{U}_{e}$ and pursuer
uses ${\bf u_{p}^{*}}\in\mathcal{U}_{p}$ the time to capture $T_{c}({\bf u_{p}^{*}},{\bf u_{e}^{*}})=T$.
If the evader deviates from ${\bf u_{e}^{*}}$ by ${\bf \delta u_{e}}$
then by Theorem \ref{thm:capture_sufficient_condition} it will be
captured at $T_{c}({\bf u_{p}^{*}},{\bf u_{e}^{*}}+{\bf \delta u_{e}})\leq T$.
Similarly, if the pursuer deviates from ${\bf u_{p}^{*}}$ by ${\bf \delta u_{p}}$
it will not be able to reach point ${\bf e}_{|\mathbb{R}^{2}}(T)={\bf z}\in\mathcal{C}_{p}({\bf e_{0}},T)$
at time $T$ and we have $T\leq T_{c}({\bf u_{p}^{*}}+{\bf \delta u_{p}},{\bf u_{e}^{*}})$.
Thus, 
\begin{equation}
T_{c}({\bf u_{p}^{*}},{\bf u_{e}^{*}}+{\bf \delta u_{e}})\leq T=T_{c}({\bf u_{p}^{*}},{\bf u_{e}^{*}})\leq T_{c}({\bf u_{p}^{*}}+{\bf \delta u_{p}},{\bf u_{e}^{*}})\label{eq:saddle_point_capture}
\end{equation}
and hence ${\bf u_{e}^{*}}$ and ${\bf u_{p}^{*}}$ are the feedback
saddle-point strategies. Since, $T$ satisfies (\ref{eq:saddle_point_capture})
we have $T=T^{*}$.

\section{\label{sec:imp_theorem_2}Proof of Theorem \ref{thm:cs-type}}

In this section we explicitly characterize the continuous subset of
evader's reachable set which will be the $LCSE$. For the pursuer
we show that $CCSP$ will be a set from among three particular continuous
subsets of the pursuer's reachable set. Using $LCSE$ and $CCSP$
we show that the saddle point trajectories are of the type $CS$.

\subsection{\label{subsec:lcse}Characterization of $LCSE$}

In this section we characterize the $LCSE$ $\mathcal{L}_{e}({\bf e_{0}},T^{*})$.
Recall that $r_{p}$ and $r_{e}$ are the minimum turning radii of
the pursuer and the evader respectively, while $d_{pe}^{0}$ denotes
the initial distance between the pursuer and the evader. If the capture
time $T^{*}\geq2\pi r_{e}/v_{e_{m}}$, then we show that the boundary
of $LCSE$ is the external boundary of evader's reachable set. This
simplifies the analysis of the game significantly. Hence, in order
to ensure that the capture of evader takes at least time $2\pi r_{e}/v_{e_{m}}$
we add an assumption on the initial distance between the pursuer and
the evader. 
\begin{lem}
\label{lem:evader_circles_forbidden}If $d_{pe}^{0}\geq2r_{e}+2\pi r_{e}(v_{p_{m}}/v_{e_{m}})$
then capture can occur only at time $t\geq2\pi r_{e}/v_{e_{m}}$.
\end{lem}
\begin{IEEEproof}
Let the evader be located at $E$, at a distance $d_{pe}^{0}$ away
from the pursuer located at $P$, as shown in Figure \ref{fig:evaders_circles_forbidden}.
Let $PA$ be the line from pursuer's position to the center of anti-clockwise
evader circle denoted by point $A$. Also, let $PA$ intersect the
anticlockwise circle at point $O$. The least time for the pursuer
to reach point $O$ is 
\begin{equation}
t_{m}:=\text{len}(PO)/v_{p_{m}}=(\text{len}(PA)-r_{e})/v_{p_{m}}\label{eq:interm-1}
\end{equation}
Also from triangle inequality we have $\text{len}(PA)\geq\text{len}(PE)-r_{e}$.
Thus (\ref{eq:interm-1}) becomes
\[
t_{m}\geq(\text{len}(PE)-r_{e}-r_{e})/v_{p_{m}}=[d_{pe}^{0}-2r_{e}]/v_{p_{m}}
\]
Let, $t_{p}:=(d_{pe}^{0}-2r_{e})/v_{p_{m}}$ and $t_{e}:=2\pi r_{e}/v_{e_{m}}$.
Thus if $t_{p}\geq t_{e}$ the pursuer requires at least $t_{e}$
time to reach any point on evader circles. Thus the evader, by using
$v_{e}(t)=0$ and $u_{e}(t)=0$ $\forall\ t$ can avoid capture up
to time $t_{e}$ (since it remains at point $E$ $\forall\ t$ and
pursuer takes at least time $t_{e}$ to reach it). Further, if it
uses time optimal evasion strategy then the time to capture will certainly
be greater than $t_{e}$. Thus, if $d_{pe}^{0}\geq2r_{e}+2\pi r_{e}(v_{p_{m}}/v_{e_{m}})$
we have $t_{p}=(d_{pe}^{0}-2r_{e})/v_{p_{m}}\geq2\pi r_{e}/v_{e_{m}}$.
This implies $t_{p}\geq t_{e}$ and the claim follows.
\end{IEEEproof}
\begin{figure}

\begin{centering}
\includegraphics[scale=0.25]{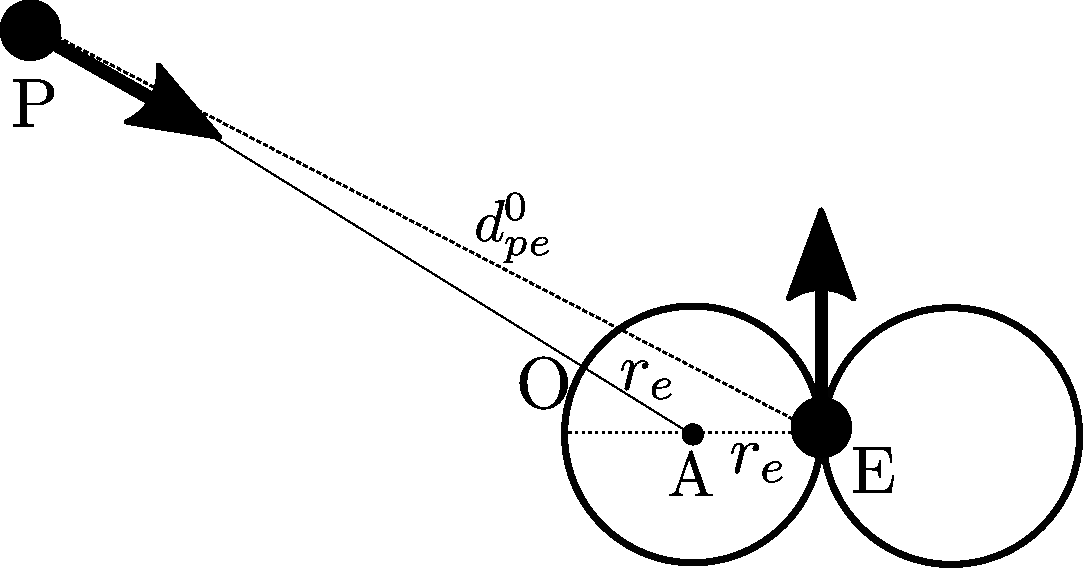}
\par\end{centering}
\caption{\label{fig:evaders_circles_forbidden}Evader capture after $t\protect\geq2\pi r_{e}/v_{e_{m}}$}

\end{figure}
Recall that, by Lemma \ref{lem:CRS_continuous}, the central reachable
set $R_{e}^{CR}({\bf e_{0}},T^{*})$ is a $CSE$. In next theorem
we prove that the $LCSE$ is in fact $R_{e}^{CR}({\bf e_{0}},T^{*})$
if $d_{pe}^{0}\geq2r_{e}+2\pi r_{e}(v_{p_{m}}/v_{e_{m}})$. Further,
since the boundary of $R_{e}^{CR}({\bf e_{0}},T^{*})$ and the external
boundary of $R_{e}({\bf e_{0}},T^{*})$ are the same, it follows that
the whole reachable set of the evader must be contained in some continuous
subset of the evader's reachable set. In the next theorem we specialize
the Theorems \ref{thm:capture_sufficient_condition} and \ref{thm:capture_necessary_condition}
for the assumption that $d_{pe}^{0}\geq2r_{e}+2\pi r_{e}(v_{p_{m}}/v_{e_{m}})$,
and give a necessary and sufficient condition for capture under feedback
trajectories.
\begin{thm}
\label{thm:complete_containment_continuous_necessary_sufficient}Let
$d_{pe}^{0}\geq2r_{e}+2\pi r_{e}(v_{p_{m}}/v_{e_{m}})$. For every
evader trajectory ${\bf e}_{|\mathbb{R}^{2}}$ generated by feedback
policy ${\bf u_{e}}\in\mathcal{U}$ there exists a pursuer trajectory
${\bf p}_{|\mathbb{R}^{2}}$ corresponding to feedback strategy ${\bf u_{p}}\in\mathcal{U}_{p}$
s.t. ${\bf e}_{|\mathbb{R}^{2}}(t)={\bf p}_{|\mathbb{R}^{2}}(t)$
for some $t\leq T^{*}$ if and only if $R_{e}^{-}({\bf e_{0}},\mathcal{C}_{p},T^{*})=\emptyset$
for a $CCSP$ $\mathcal{C}_{p}({\bf p_{0}},T^{*})$.
\end{thm}
\begin{IEEEproof}
First we show the $only\ if$ part. As seen in Lemma \ref{lem:evader_circles_forbidden},
if $d_{pe}^{0}\geq2r_{e}+2\pi r_{e}(v_{p_{m}}/v_{e_{m}})$ the pursuer
can capture the evader only after time $t_{e}=2\pi r_{e}/v_{e_{m}}$.
The evader's central reachable set after time $t\geq t_{e}=2\pi r_{e}/v_{e_{m}}$
is shown in Figure \ref{fig:Central-Reachable-Region}. It can be
easily seen that the external boundary of the evader's reachable set
and that of the central reachable set are the same i.e. $\partial R_{e}({\bf e_{0}},T^{*})=\partial R_{e}^{CR}({\bf e_{0}},T^{*})$.
Now, $R_{e}^{CR}({\bf e_{0}},T^{*})$ is a $CSE$. By Theorems \ref{thm:capture_necessary_condition}
and \ref{thm:capture_sufficient_condition} $R_{e}^{CR}({\bf e_{0}},T^{*})$
must be contained inside some $CCSP$ $\mathcal{C}_{p}({\bf p_{0}},T^{*})$
for capture to occur at time $t\leq T^{*}$. Since, $\partial R_{e}({\bf e_{0}},T^{*})=\partial R_{e}^{CR}({\bf e_{0}},T^{*})$,
it is necessary that $R_{e}({\bf e_{0}},T^{*})$ must be contained
inside some $\mathcal{C}_{p}({\bf p_{0}},T^{*})$ for the capture
to occur at $t\leq T^{*}$.

Now we prove the $if$ part. Since $R_{e}^{c}({\bf e_{0}},T^{*})\subseteq R_{e}({\bf e_{0}},T^{*})$
for all $R_{e}^{c}({\bf e_{0}},T^{*})\in\mathfrak{R}({\bf e_{0}},T^{*})$,
the condition $R_{e}^{-}({\bf e_{0}},\mathcal{C}_{p},T^{*})=\emptyset$
implies that every $R_{e}^{c}({\bf e_{0}},T^{*})$$\in$$\mathfrak{R}_{e}({\bf e_{0}},T^{*})$
is such that $R_{e}^{c-}({\bf e_{0}},\mathcal{C}_{p},T^{*})=\emptyset$
for some corresponding $CCSP$ $\mathcal{C}_{p}({\bf p_{0}},T^{*})$.
Thus by Theorems \ref{thm:capture_sufficient_condition} and \ref{thm:capture_necessary_condition}
we can conclude that capture will occur at $t\leq T^{*}$.
\end{IEEEproof}

\subsection{\label{subsec:ccsp}Characterization of $CCSP$}

Next, we will characterize $CCSP$ $\mathcal{C}_{p}({\bf p_{0}},T^{*})$
in order to determine the nature of feedback strategies. We will show
that the $CCSP$ is either the blocking set $B_{p}({\bf p_{0}},{\bf e_{0}},T^{*})$
or the truncated left reachable set $R_{p}^{l_{t}}({\bf p_{0}},T^{*})$
or the truncated right reachable set $R_{p}^{r_{t}}({\bf p_{0}},T^{*})$.
First we define the concept of blocking $LS$ and $RS$ trajectories
of the pursuer.
\begin{defn}
\textbf{\textit{Blocking trajectory of the pursuer }}: If the initial
position of the evader (${\bf e_{0}}$) and pursuer (${\bf p_{0}}$)
is such that all the points on a pursuer trajectory are reached by
the pursuer earlier than the evader then we say that such a trajectory
is a blocking trajectory. Further, if the blocking trajectory is of
the type $LS\ (RS)$ then we call it $LS\ (RS)$ blocking trajectory.
\end{defn}
This concept of blocking trajectory will be used to analyze the containment
of evader's reachable set by $B_{p}({\bf p_{0}},{\bf e_{0}},T^{*})$,
$R_{p}^{l_{t}}({\bf p_{0}},T^{*})$, and $R_{p}^{r_{t}}({\bf p_{0}},T^{*})$.

\subsection*{Continuous capture by blocking set}

We show that the blocking set (defined in Definition \ref{def:blocking_set}
and shown to be a continuous subset in Lemma \ref{lem:blocking_continuous})
will contain the evader's reachable set for all initial conditions
of the pursuer and the evader such that $d_{pe}^{0}\geq2r_{e}+2\pi r_{e}(v_{p_{m}}/v_{e_{m}})$.
Recall that $\tilde{T}_{l}^{v}$ and $\tilde{T}_{r}^{v}$ are the
trajectories of the type $LS$ and $RS$ as defined in Definition
\ref{def:blocking_set}. The next lemma follows easily from the construction
of $B_{p}({\bf p_{0}},{\bf e_{0}},\bar{t})$ for some time $\bar{t}\geq2\pi r_{p}/v_{p_{m}}$. 
\begin{lem}
\label{lem:blocking_curves_blocking_set}$\tilde{T}_{l}^{v}$ is a
blocking $LS$ trajectory while $\tilde{T}_{r}^{v}$ is a blocking
$RS$ trajectory if $d_{pe}^{0}\geq2r_{e}+2\pi r_{e}(v_{p_{m}}/v_{e_{m}})$.
\end{lem}
\begin{rem}
By Lemma \ref{lem:blocking_curves_blocking_set}, $\tilde{T}_{l}^{v}$
and $\tilde{T}_{r}^{v}$ block the evader from entering the region
shown by crisscross shading in Figure \ref{fig:construction_blocking_2}.
Some examples of blocking sets have been shown in Figures \ref{fig:blocking_evader_behind},
\ref{fig:blocking_evader_front}, and \ref{fig:blocking_front_extreme}.
As can be seen from these examples, at most one curve out of $\tilde{T}_{l}^{v}$
and $\tilde{T}_{r}^{v}$ can end on the left internal boundary or
the right internal boundary for any position of the pursuer and evader
after time $\bar{t}\geq2\pi r_{p}/v_{p_{m}}$. 
\end{rem}
Next we show that for all possible initial positions of the evader,
the blocking set will contain the evader's reachable set at some time
$T_{b}<\infty$.
\begin{lem}
\label{lem:blocking_set_finite_time}For every evader and pursuer
initial positions such that $d_{pe}^{0}\geq2r_{e}+2\pi r_{e}(v_{p_{m}}/v_{e_{m}})$,
the set $R_{e}^{-}({\bf e_{0}},B_{p},T_{b})=\emptyset$ for some time
$T_{b}<\infty$.
\end{lem}
\begin{IEEEproof}
Recall that, if $d_{pe}^{0}\geq2r_{e}+2\pi r_{e}(v_{p_{m}}/v_{e_{m}})$
then each of the left reachable set and the right reachable set contain
the evader's reachable set at times $T_{l}$ and $T_{r}$ respectively
(by Lemma \ref{lem:left_containment} and Lemma \ref{lem:right_containment}).
Now, since either the blocking curve $\tilde{T}_{l}^{v}$ or the blocking
curve $\tilde{T}_{r}^{v}$ may end on the internal boundary, without
loss of generality, let $\tilde{T}_{l}^{v}$ end on the left internal
boundary (see Definition \ref{def:internal_boundary}) as shown in
Figure \ref{fig:construction_blocking_1}. Now consider the left reachable
set of the pursuer. Some portion of the left reachable set is not
a part of the blocking set. Thus, at time $t_{m}=\max(T_{l},T_{r})$,
$B_{p}({\bf p_{0}},{\bf e_{0}},t_{m})$ may not contain some part
of $R_{e}({\bf e_{0}},t_{m})$ even if the left reachable set contains
the reachable set completely. However, since $d_{pe}^{0}\geq2r_{e}+2\pi r_{e}(v_{p_{m}}/v_{e_{m}})$
the evader cannot enter this part as $\tilde{T}_{l}^{v}$ and $\tilde{T}_{r}^{v}$
are blocking curves (by Lemma \ref{lem:blocking_curves_blocking_set}).
Apart from this portion the left reachable set is a part of $B_{p}({\bf p_{0}},{\bf e_{0}},t_{m})$.
Thus, the remaining part of $R_{e}({\bf e_{0}},t_{m})$ is contained
in the blocking set and we have $R_{e}^{-}({\bf e_{0}},B_{p},t_{m})=\emptyset$.
Hence, for some time $T_{b}\leq t_{m}$ we will have $R_{e}^{-}({\bf e_{0}},B_{p},T_{b})=\emptyset$
. 
\end{IEEEproof}
\begin{figure*}
\begin{minipage}[t]{0.3\textwidth}%
\begin{center}
\includegraphics[scale=0.18]{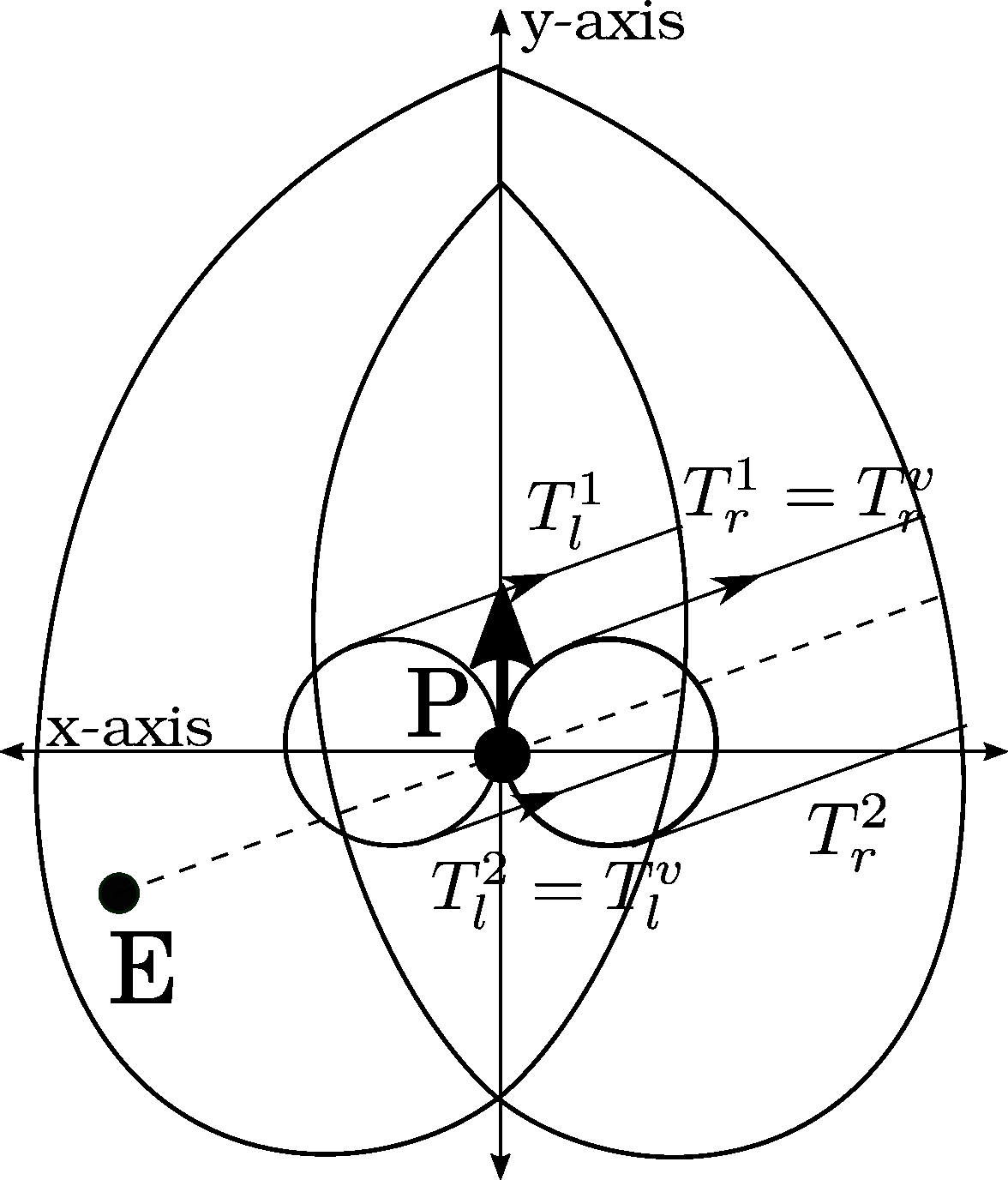}
\par\end{center}
\caption{\label{fig:construction_blocking_1}Construction of $B_{p}({\bf p_{0}},{\bf e_{0}},\bar{t})$}
\end{minipage}\hfill{}%
\begin{minipage}[t]{0.3\textwidth}%
\begin{center}
\includegraphics[scale=0.18]{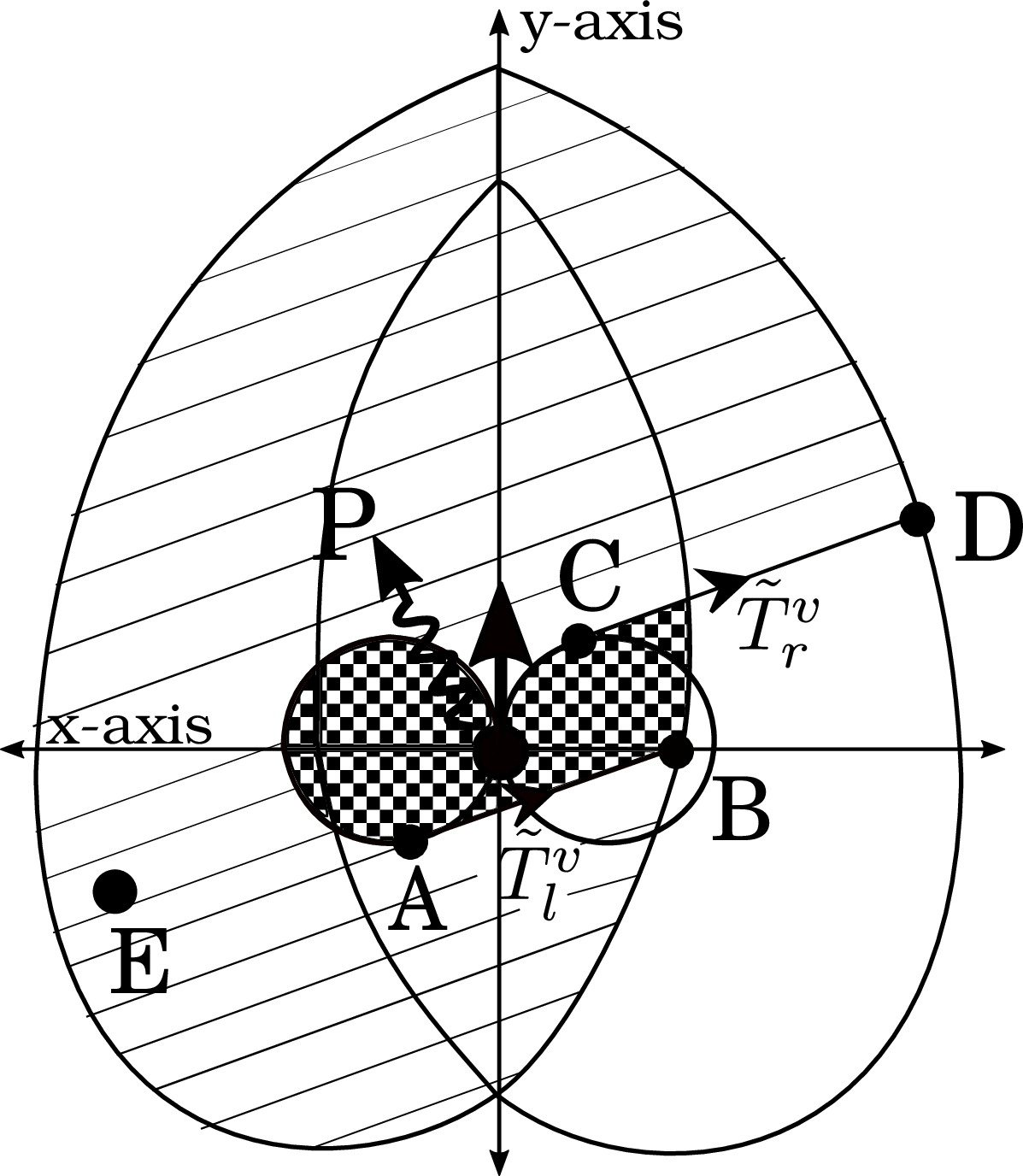}
\par\end{center}
\caption{\label{fig:construction_blocking_2}$B_{p}({\bf p_{0}},{\bf e_{0}},\bar{t})$:
Evader behind on left side (shaded region)}
\end{minipage}\hfill{}%
\begin{minipage}[t]{0.3\textwidth}%
\begin{center}
\includegraphics[scale=0.1]{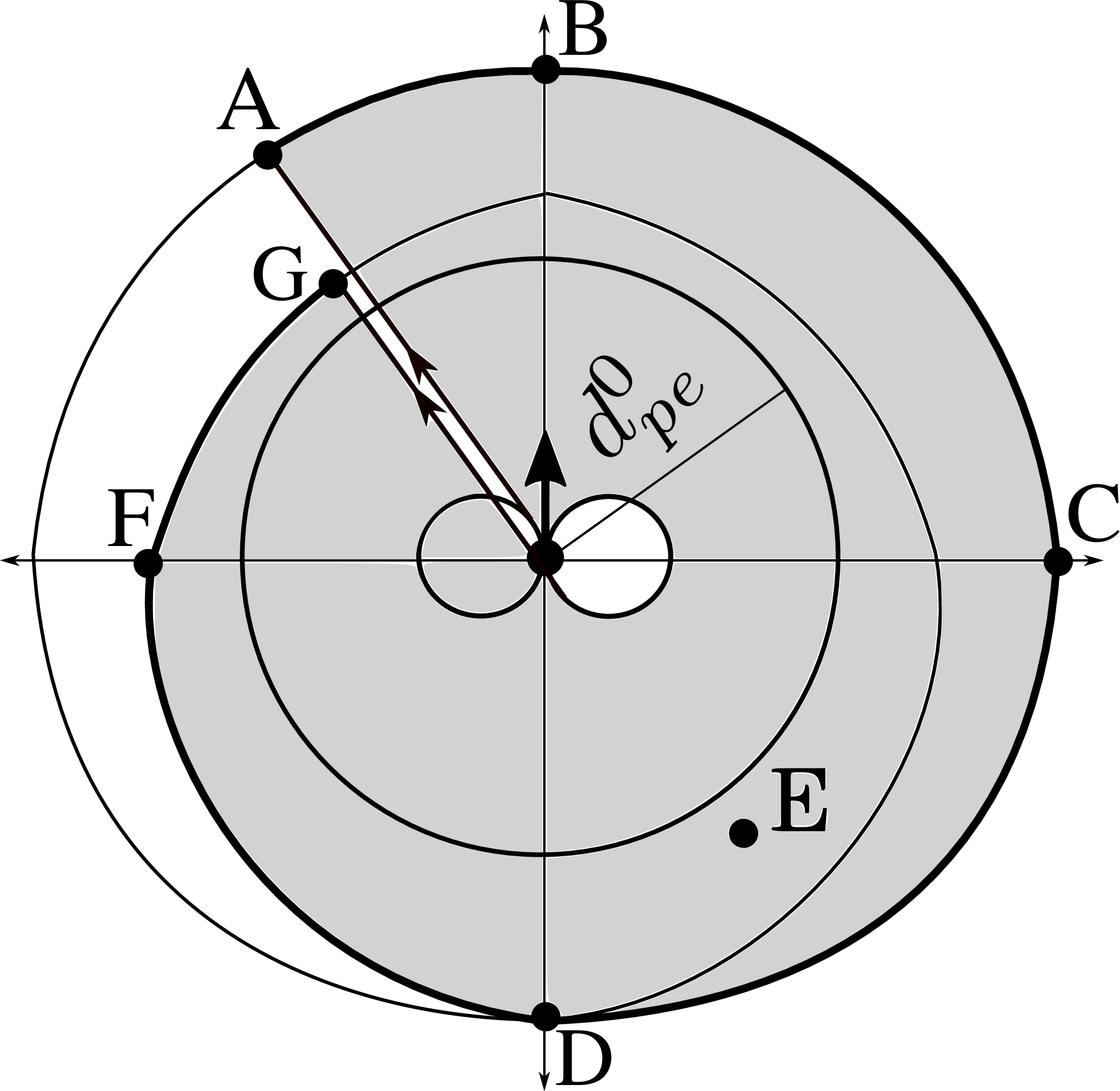}
\par\end{center}
\caption{\label{fig:blocking_evader_behind}$B_{p}({\bf p_{0}},{\bf e_{0}},\bar{t})$:
Evader behind on right side (shaded region)}
\end{minipage}
\end{figure*}

\subsection*{Continuous containment by truncated left reachable set}

In this section we analyze the conditions under which $R_{p}^{l_{t}}({\bf p_{0}},\bar{t})$
and $R_{p}^{r_{t}}({\bf p_{0}},\bar{t})$ can contain the evader's
reachable set at some time $\bar{t}$. First we define two $LS$ ($RS$)
blocking trajectories for the truncated left (right) reachable set.

Let $PA$ be the curve defined by (\ref{eq:left_reachable_set}) with
$t_{1}=0$ and is shown in Figure \ref{fig:Left-Reachable-Set}. We
call this curve $BL^{1}$. Similarly, let $PCB$ be the curve defined
by (\ref{eq:left_reachable_set}) with $t_{1}=2\pi r_{p}/v_{p_{m}}$
and is shown in Figure \ref{fig:Left-Reachable-Set}. We call this
curve $BL^{2}$. Analogously, $BR^{1}$ and $BR^{2}$ are defined
for the truncated right reachable set.
\begin{rem}
Now, the left reachable set $R_{p}^{l}({\bf p_{0}},T_{l})$ contains
the evader's reachable set at time $T_{l}$ (by Lemma \ref{lem:left_containment}).
Since the truncated reachable set does not contain $PA$ and the anti-clockwise
circle $A_{p}(0)$, hence if evader's reachable set intersects the
line $PA$ or the anti-clockwise circle, continuous containment by
$R_{p}^{l}({\bf p_{0}},T_{l})$ will not happen. However, such a situation
is avoided since $BL^{1}$ and $BL^{2}$ act as blocking curves for
some initial conditions of the pursuer and the evader, and block the
evader from entering the points in $PA$ and $A_{p}(0)$ analogous
to what $\tilde{T}_{l}^{v}$ and $\tilde{T}_{r}^{v}$ achieve for
the blocking set. This is shown in Proposition \ref{prop:left_right_trunc_continuous_behind}. 
\end{rem}
First we define some terminology. Consider Figure \ref{fig:BL2_blocking}.
Let $OPQ$ be the line passing through the pursuer position $P$ and
perpendicular to the orientation of the pursuer. If the evader is
in the closed half plane in the direction opposite to the orientation
of the pursuer then we say that the evader is behind the pursuer.
If the evader is located in the open half plane in the direction of
the orientation of the pursuer then we say that the evader is located
in the front of the pursuer. For $R_{p}^{l_{t}}({\bf p_{0}},T_{l})$,
let $PF$ be the line as shown in Figure \ref{fig:BL2_blocking} at
an angle $\theta$ with respect to $PS$. The shaded region between
$PO$ and $PF$, for $\theta=\arccos\left(2\pi r_{p}v_{e_{m}}/(v_{p_{m}}d_{pe}^{0})\right)$,
is denoted by $S^{l}(T_{l})$. Similarly, for $R_{p}^{r_{t}}({\bf p_{0}},T_{r})$
let $PF$ be the line as shown in Figure \ref{fig:BR2_blocking} at
an angle $\theta$ with respect to $PS$. The shaded region between
$PO$ and $PF$ for, $\theta=\arccos\left(2\pi r_{p}v_{e_{m}}/(v_{p_{m}}d_{pe}^{0})\right)$,
is denoted by $S^{r}(T_{r})$.

\noindent 
\begin{figure*}
\begin{minipage}[t]{0.3\textwidth}%
\begin{center}
\includegraphics[scale=0.4]{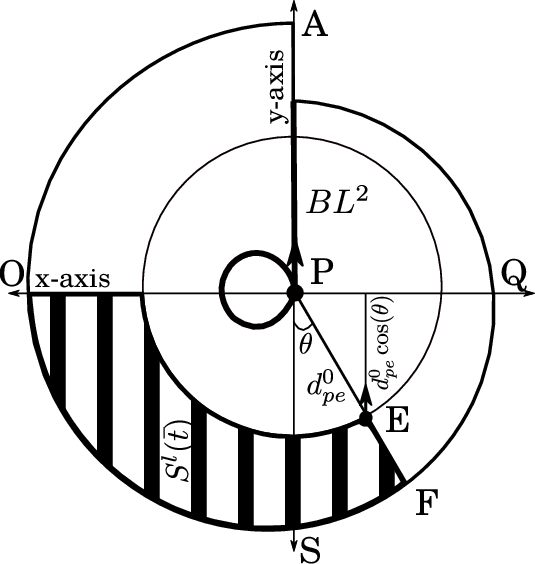}
\par\end{center}
\caption{\label{fig:BL2_blocking}$BL^{2}$: Blocking $LS$ trajectory}
\end{minipage}\hfill{}%
\begin{minipage}[t]{0.3\textwidth}%
\begin{center}
\includegraphics[scale=0.4]{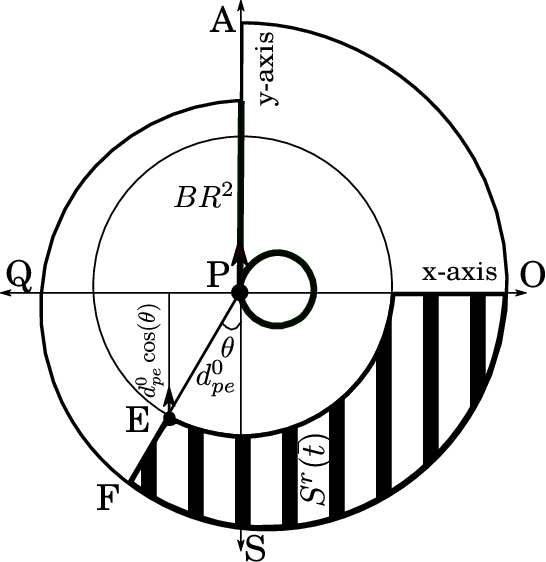}
\par\end{center}
\caption{\label{fig:BR2_blocking}$BR^{2}$: Blocking $RS$ trajectory}
\end{minipage}\hfill{}%
\begin{minipage}[t]{0.3\textwidth}%
\begin{center}
\includegraphics[scale=0.09]{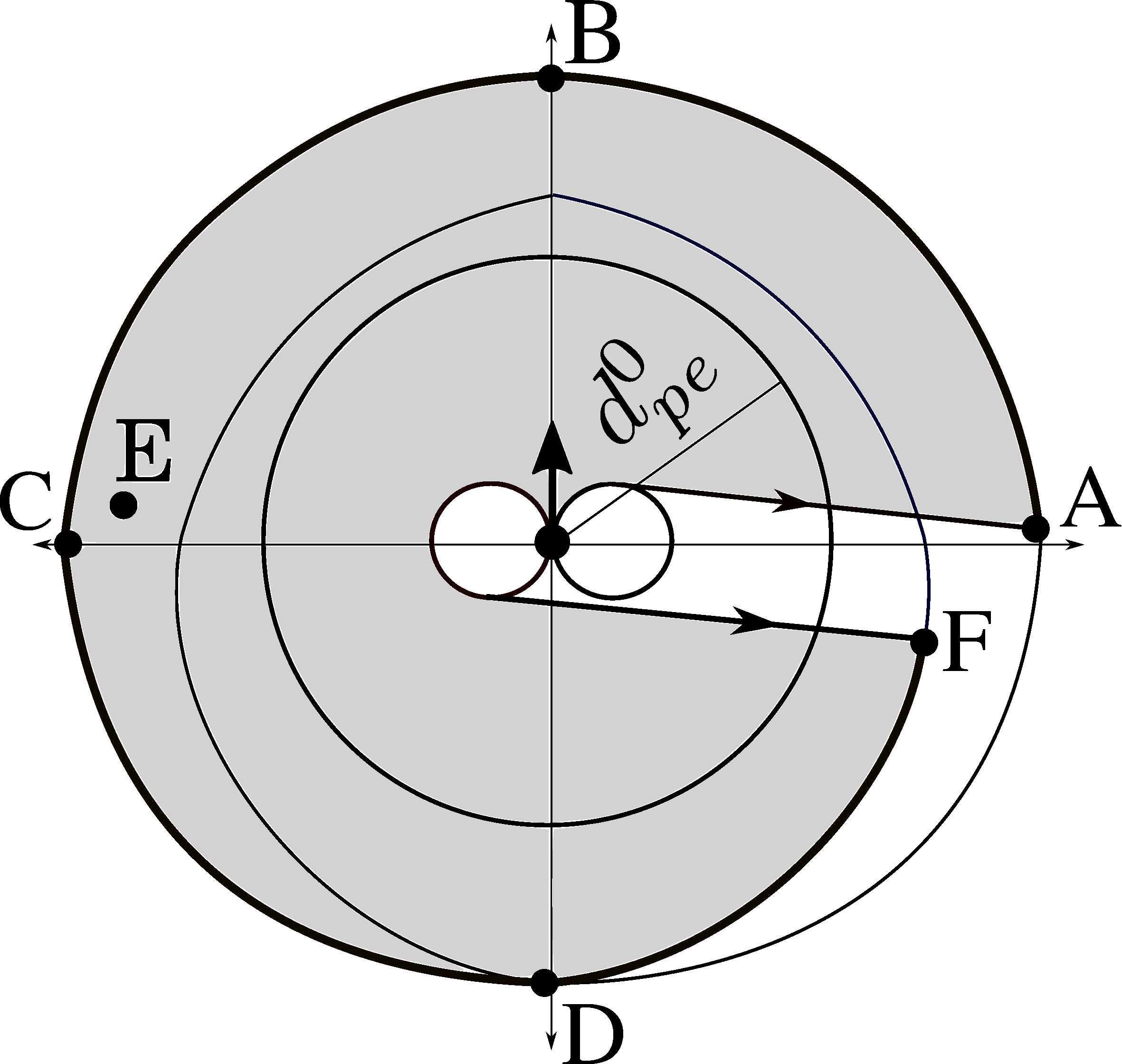}
\par\end{center}
\caption{\label{fig:blocking_front_extreme}$B_{p}({\bf p_{0}},{\bf e_{0}},T_{b})$
as active set}
\end{minipage}
\end{figure*}

\begin{lem}
\label{lem:blocking_curves_left_truncated_set}Let $d_{pe}^{0}\geq2r_{e}+2\pi r_{e}(v_{p_{m}}/v_{e_{m}})$
and let the evader's initial position be in $S^{l}(T_{l})$. Then,
$R_{p}^{l_{t}}({\bf p_{0}},T_{l})$ will contain the evader's reachable
set continuously.
\end{lem}
\begin{IEEEproof}
\textbf{\textit{$BL^{1}$ as blocking trajectory}} : Since evader
is in $S^{l}(T_{l})$, it is behind the pursuer. All the points on
the trajectory $BL^{1}$ can be reached by the pursuer before the
evader can reach them. Hence, it becomes the blocking $LS$ trajectory.
Thus $BL^{1}$ will prevent the evader from reaching any point on
$PA$ from the left side of the pursuer.

\textbf{\textit{$BL^{2}$ as blocking trajectory}} : If $d_{pe}^{0}\geq2r_{e}+2\pi r_{e}(v_{p_{m}}/v_{e_{m}})$
then the pursuer can reach all the points on the anti-clockwise circle
of the pursuer before the evader (This follows from analysis similar
to Lemma \ref{lem:evader_circles_forbidden}). Since the boundary
of $A_{p}(0)$ forms a part of $BL^{2}$ we can conclude that $BL^{2}$
prevents the evader from entering $A_{p}(0)$. 

Now, if the evader tries to reach any point on line $PA$ from the
right, while the pursuer is initially traveling on $A_{p}(0)$, we
will show that $BL^{2}$ will intercept the evader. If the pursuer
follows $BL^{2}$, it will reach point $P$ again after encircling
the anticlockwise circle once, for which it will require time $t_{p}:=2\pi r_{p}/v_{p_{m}}$.
Thus if the evader requires time $t_{e}\geq t_{p}$ to reach the line
$OPQ$ then $BL^{2}$ will be a blocking $LS$ trajectory. Let the
evader be located at an angle $\theta$ with respect to line $PS$
as shown in Figure \ref{fig:BL2_blocking}. At this position the evader
requires minimum time to reach a point on line $OPQ$ if it is pointing
upwards. This minimum time is $d_{pe}^{0}\cos(\theta)/v_{e_{m}}$.
Thus $t_{e}\geq d_{pe}^{0}\cos(\theta)/v_{e_{m}}$. Hence $BL^{2}$
is the blocking trajectory if 
\[
d_{pe}^{0}\cos(\theta)/v_{e_{m}}\geq2\pi r_{p}/v_{p_{m}};\quad\cos(\theta)\geq2\pi r_{p}v_{e_{m}}/(v_{p_{m}}d_{pe}^{0})
\]
This implies that for $\theta\leq\arccos\left(2\pi r_{p}v_{e_{m}}/(v_{p_{m}}d_{pe}^{0})\right)$,
$BL^{2}$ will act as blocking trajectory. 

\noindent The left reachable set contains the evader's reachable set
at time $T_{l}$ (by Lemma \ref{lem:left_containment}). Further,
$BL^{1}$ and $BL^{2}$ act as blocking curves and block the evader
from entering the points in $PA$ and $A_{p}(0)$ and the claim follows.
\end{IEEEproof}
\begin{defn}
Let $T_{l}^{c}$ be the minimum time such that $R_{e}({\bf e_{0}},T_{l}^{c})\subseteq R_{p}^{l_{t}}({\bf p_{0}},T_{l}^{c})$
continuously i.e. $R_{e}^{-}({\bf e_{0}},R_{p}^{l_{t}},T_{l}^{c})=\emptyset$.
If for some initial positions of the pursuer and the evader, $R_{p}^{l_{t}}({\bf p_{0}},T_{l}^{c})$
does not contain $R_{e}({\bf e_{0}},T_{l}^{c})$ continuously then
we say $R_{e}^{-}({\bf e_{0}},R_{p}^{l_{t}},T_{l}^{c})=\emptyset$
at $T_{l}^{c}=\infty$.
\end{defn}
\begin{rem}
From the above analysis it follows that if $T_{l}^{c}<\infty$ then
$T_{l}^{c}=T_{l}$. However, if $T_{l}^{c}=\infty$ this implies that
the truncated left reachable set cannot contain evader's reachable
set continuously.
\end{rem}
\begin{lem}
\label{lem:blocking_curves_right_truncated_set}Let $d_{pe}^{0}\geq2r_{e}+2\pi r_{e}(v_{p_{m}}/v_{e_{m}})$
and let the evader's initial position be located in $S^{r}(T_{r})$.
Then, $R_{p}^{r_{t}}({\bf p_{0}},T_{r})$ will contain the evader's
reachable set continuously at some finite time $\bar{t}\geq2\pi r_{p}/v_{p_{m}}$.
\end{lem}
\begin{rem}
The minimum time for continuous containment by the truncated right
reachable set is denoted by $T_{r}^{c}$. If $T_{r}^{c}<\infty$ then
$T_{r}^{c}=T_{r}$. However, if $T_{r}^{c}=\infty$ this implies that
the truncated right reachable set cannot contain evader's reachable
set continuously.
\end{rem}
\begin{prop}
\label{prop:left_right_trunc_continuous_behind}Let $d_{pe}^{0}\geq2r_{e}+2\pi r_{e}(v_{p_{m}}/v_{e_{m}})$.
If the evader is behind the pursuer then the evader's reachable set
is contained continuously in the left reachable set or the right reachable
set or both at time $t_{lr}=\min(T_{l}^{c},T_{l}^{r})$.
\end{prop}
\begin{IEEEproof}
If the evader is behind the pursuer at a distance greater than $d_{pe}^{0}\geq2r_{e}+2\pi r_{e}(v_{p_{m}}/v_{e_{m}})$
it is either located in $S^{l}(T_{l}^{c})$ or $S^{r}(T_{r}^{c})$
or both and the claim follows from Lemma \ref{lem:blocking_curves_left_truncated_set}
and Lemma \ref{lem:blocking_curves_right_truncated_set}.
\end{IEEEproof}

\subsection*{Representative set and frontier boundaries}

Next, we define the representative set at time $t\geq2\pi r_{p}/v_{p_{m}}$
which comprises of the blocking set, truncated left reachable set,
and the truncated right reachable set. 
\begin{defn}
\textbf{\textit{Representative set }}at time $t\geq2\pi r_{p}/v_{p_{m}}$
is 
\begin{eqnarray*}
\mathfrak{R}({\bf p_{0}},{\bf e_{0}},t): & = & \{B_{p}({\bf p_{0}},{\bf e_{0}},t),R_{p}^{l_{t}}({\bf p_{0}},t),R_{p}^{r_{t}}({\bf p_{0}},t)\}
\end{eqnarray*}

Let $T_{a}$ be the time such that for some $\tilde{R}_{p}^{c}\in\mathfrak{R}({\bf p_{0}},{\bf e_{0}},T_{a})$
we have $R_{e}^{-}({\bf e_{0}},\tilde{R}_{p}^{c},T_{a})=\emptyset$
and for all $t<T_{a}$, $R_{e}^{-}({\bf e_{0}},R_{p}^{c},t)\neq\emptyset$
for all $R_{p}^{c}\in\mathfrak{R}({\bf p_{0}},{\bf e_{0}},t)$ i.e.
$T_{a}=\min\{T_{b},T_{l}^{c},T_{r}^{c}\}$. We denote this set $\tilde{R}_{p}^{c}\in\mathfrak{R}({\bf p_{0}},{\bf e_{0}},T_{a})$
by $\mathcal{A}_{p}(T_{a})$ and call it the \textbf{\textit{active
pursuer set}}.
\end{defn}
Note that, since $R_{e}^{-}({\bf e_{0}},B_{p},T_{b})=\emptyset$ for
some time $T_{b}<\infty$ (Lemma \ref{lem:blocking_set_finite_time}),
we have $T_{a}<\infty$. 
\begin{rem}
The capture criterion, which could not be explained using only containment
of reachable sets can now be explained using the active set and $LCSE$.
Consider the position of the pursuer and evader as shown in the Figure
\ref{fig:invalid_containment}. In this case the active set is the
 left reachable set as shown in Figure \ref{fig:left_capture_behind}
and capture occurs at time $T_{l}^{c}=T_{l}$. Further, at time $t<T_{l}$,
the safe region of evader is non-empty (shown in Figure \ref{fig:left_safe_behind}).
Next consider the situation when the evader is located ahead of the
pursuer as shown in Figure \ref{fig:valid_front}. In this case the
blocking set is the active set. The capture is shown in Figure \ref{fig:blocking_evaders_ahead_valid}
and occurs at time $T_{b}$. For any time $t<T_{b}$ the safe region
of the evader is non-empty (as shown in Figure \ref{fig:blocking_evader_front_safe_region})
and the evader can escape capture using feedback strategies. In the
above examples the active set explains capture under feedback strategies.
In fact, we will show that if the evader is in front of the pursuer
then it suffices to consider blocking set as $CCSP$ and if the evader
is behind the pursuer then it suffices to consider either the left
reachable set or the right reachable set as $CCSP$.
\end{rem}
Next we define frontier boundary for all the sets in the representative
set. 
\begin{defn}
\textbf{\textit{Frontier boundary}}\textit{:} 
\begin{enumerate}
\item The \textbf{\textit{frontier boundary of $R_{p}^{l_{t}}({\bf p_{0}},t)$}}
is the portion of the external boundary denoted by $ADEFB$ (start
point $A$ to end point $B$) as shown in Figure \ref{fig:Left-Reachable-Set}
by dashed curve i.e. we exclude the portion $BA$ from $\partial R_{p}^{l}({\bf p_{0}},t)$.
\item The \textbf{\textit{frontier boundary of $R_{p}^{r_{t}}({\bf p_{0}},t)$}}
is defined analogously. 
\item The \textbf{\textit{frontier boundary of the blocking set}} is the
portion of external boundary of the blocking set excluding the blocking
curves $\tilde{T}_{l}^{v}$ and $\tilde{T}_{r}^{v}$. For example,
the frontier boundary of blocking set in Figure \ref{fig:blocking_evader_behind}
is denoted by $ABCDFG$. 
\end{enumerate}
\end{defn}
\begin{rem}
The frontier boundaries of all the sets expand outwards radially with
time. Also the external boundary of the evader reachable set expands
out radially. As a result, for the active set, the capture occurs
on the frontier boundary.
\end{rem}
\begin{figure*}
\begin{minipage}[t]{0.3\textwidth}%
\begin{center}
\includegraphics[scale=0.097]{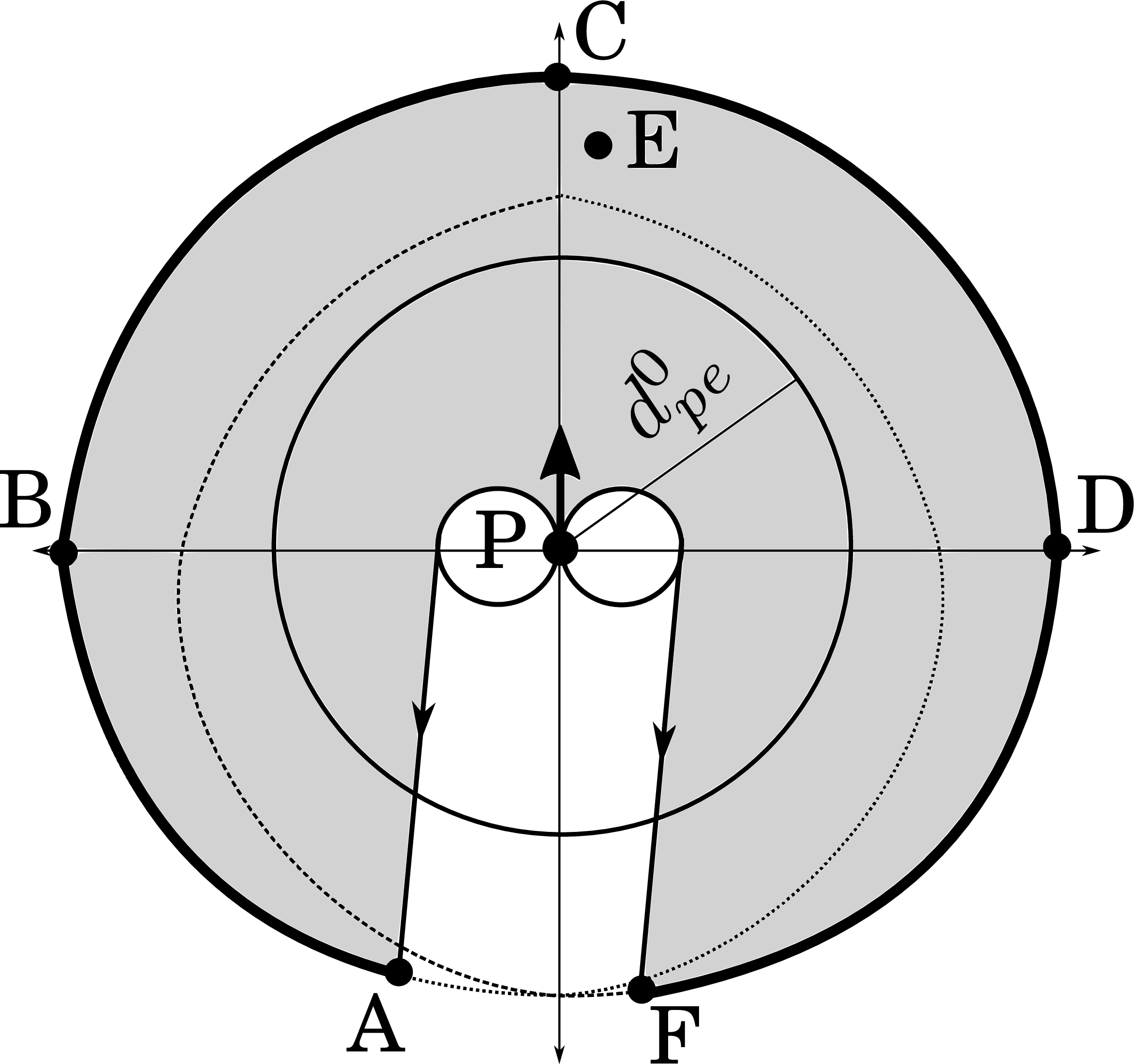}
\par\end{center}
\caption{\label{fig:blocking_evader_front}$B_{p}({\bf p_{0}},{\bf e_{0}},t_{m})$:
Evader in front (shaded region)}
\end{minipage}\hfill{}%
\begin{minipage}[t]{0.3\textwidth}%
\begin{center}
\includegraphics[scale=0.4]{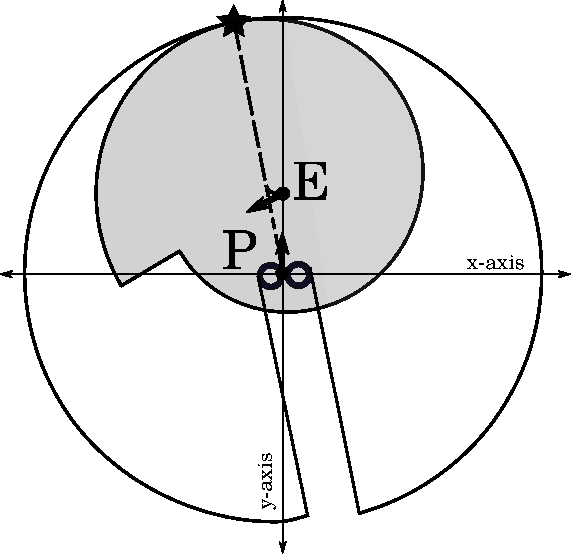}
\par\end{center}
\caption{\label{fig:blocking_evaders_ahead_valid}Active Set: Blocking set}
\end{minipage}\hfill{}%
\begin{minipage}[t]{0.3\textwidth}%
\begin{center}
\includegraphics[scale=0.4]{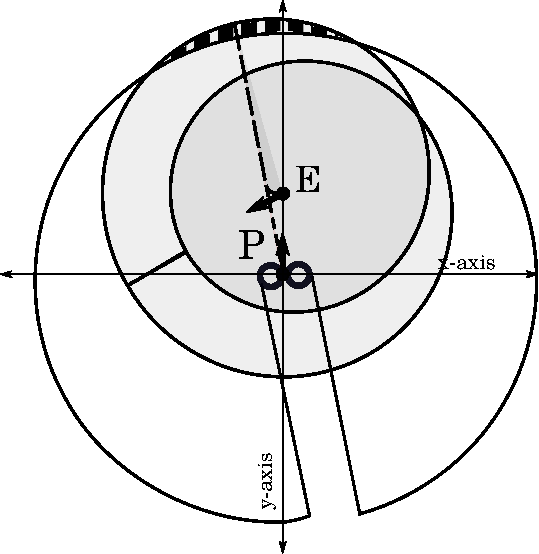}
\par\end{center}
\caption{\label{fig:blocking_evader_front_safe_region}Evader safe region}
\end{minipage}
\end{figure*}

Next we demonstrate that if $d_{pe}^{0}\geq2r_{e}+2\pi r_{e}(v_{p_{m}}/v_{e_{m}})$
the evader can always escape capture for all time $t<T_{a}$. We do
this in by considering a particular set in representative set to be
a active set in the Propositions \ref{prop:left_reachable_continuous_containment},
\ref{prop:right_reachable_continuous_containment}, and \ref{prop:blocking_continuous_containment}.
\begin{prop}
\label{prop:left_reachable_continuous_containment}Let $d_{pe}^{0}\geq2r_{e}+2\pi r_{e}(v_{p_{m}}/v_{e_{m}})$.
If $\mathcal{A}_{p}(T_{a})=R_{p}^{l_{t}}({\bf p_{0}},T_{a})$ and
$T_{a}=T_{l}^{c}<\infty$, then the evader can always escape capture
for $t<T_{a}$ by using feedback strategy. 
\end{prop}
\begin{IEEEproof}
First we divide the frontier boundary of $R_{p}^{l_{t}}({\bf p_{0}},T_{a})$
into two parts $ADE$ and $EGFHB$ as shown in Figure \ref{fig:left_reachable_active_set}.
Let ${\bf z}\in\mathbb{R}^{2}$ be a last point of $R_{e}({\bf e_{0}},T_{a})$
covered by $R_{p}^{l_{t}}({\bf p_{0}},T_{a})$ at time $T_{a}$. Since
the frontier boundary of $R_{p}^{l_{t}}({\bf p_{0}},t)$ and the external
boundary $R_{e}({\bf e_{0}},t)$ both grow out radially with time
$t$, ${\bf z}$ will lie on the external boundary of $R_{e}({\bf e_{0}},T_{a})$.
Also, the point in $R_{p}^{l_{t}}({\bf p_{0}},T_{a})$ which covers
${\bf z}$ will lie on the frontier boundary of $R_{p}^{l_{t}}({\bf p_{0}},T_{a})$.
Now the situation can be divided into two cases: 

Case 1: Let ${\bf z}\in ADE$. Then at time $T_{a}'<T_{a}$ we have
$R_{e}^{-}({\bf e_{0}},R_{p}^{l},T_{a}^{'})\neq\emptyset$ i.e. the
pursuer's reachable set does not cover the evader's reachable set
and the evader can always escape capture. 

Case 2: Let ${\bf z}\in EGFHB$. We further divide the portion $EGFHB$
into two parts namely $EGF$ and $FHB$ as shown in Figure \ref{fig:left_reachable_active_set}.
Now let ${\bf z}\in FHB$. In Figure \ref{fig:left_reachable_active_set}
the right reachable set is shown by dotted curve. Since $d_{pe}^{0}\geq2r_{e}+2\pi r_{e}(v_{p_{m}}/v_{e_{m}})$,
the evader is located outside the circle of radius $d_{pe}^{0}$.
It is clear from geometry that if ${\bf z}\in FHB$ then the right
external boundary would have covered the evader's reachable set earlier.
Thus the right reachable set would have been the active set and this
contradicts the assumption that $R_{p}^{l_{t}}({\bf p_{0}},T_{a})$
is the active set. 

Next, let ${\bf z}\in EGF$. Let $T_{a}^{'}<T_{a}$ s.t. $R_{e}^{-}({\bf e_{0}},R_{p},T_{a}^{'})=\emptyset$
( If $R_{e}^{-}({\bf e_{0}},R_{p},T_{a}^{'})$$\neq\emptyset$ then
the pursuer's reachable set does not cover the evader's reachable
set and the evader can always escape capture). At time $T_{a}^{'}$,
$R_{p}^{l_{t}}({\bf p_{0}},T_{a}^{'})$ will not cover evader's reachable
set completely. Since $R_{e}^{-}({\bf e_{0}},R_{p},T_{a}^{'})=\emptyset$,
this uncovered region must form a part of the right reachable set
of the pursuer. Further, there must exist a point ${\bf x}$ in the
evader's reachable set such that ${\bf x}\in\left[R_{p}^{l_{t}}({\bf p_{0}},T_{a}^{'})\backslash R_{p}^{r_{t}}({\bf p_{0}},T_{a}^{'})\right]\cap R_{e}({\bf e_{0}},T_{a}^{'})$.
(For if such a point does not exist then the evader's reachable set
at $T_{a}^{'}$ would be entirely contained by the the right reachable
set. Since, $R_{p}^{r_{t}}({\bf p_{0}},T_{a}^{'})$ is in the representative
set this would contradict the assumption that $R_{p}^{l_{t}}({\bf p_{0}},T_{a})$
is the active set). Also, there exists a point ${\bf y}\in\left[R_{p}^{r_{t}}({\bf p_{0}},T_{a}^{'})\backslash R_{p}^{l_{t}}({\bf p_{0}},T_{a}^{'})\right]\cap R_{e}({\bf e_{0}},T_{a}^{'})$
(Otherwise we would have $R_{e}^{-}({\bf e_{0}},R_{p}^{l},T_{a}^{'})=\emptyset$
and this will contradict the assumption that the minimum time at which
$R_{e}^{-}({\bf e_{0}},R_{p}^{l_{t}},t)=\emptyset$ is $T_{a}$ since
$T_{a}^{'}<T_{a}$). Now consider a curve ${\bf \overline{xy}}$ in
the evader's reachable set as shown in Figure \ref{fig:unoin_not_continuous}.
Since ${\bf z}\in EGF$, such a curve is behind the pursuer and as
shown in the proof of Lemma \ref{lem:union_not_continuous} it is
not a continuum curve for the pursuer. Hence, at time $T_{a}^{'}<T_{a}$,
$R_{e}^{-}({\bf e_{0}},R_{p}^{c},T_{a}^{'})\neq\emptyset$ for all
$CSP$ $R_{p}^{c}({\bf p_{0}},T_{a}^{'})$ and the evader can escape
capture. 
\end{IEEEproof}
\begin{figure*}
\begin{minipage}[t]{0.3\textwidth}%
\begin{center}
\includegraphics[scale=0.135]{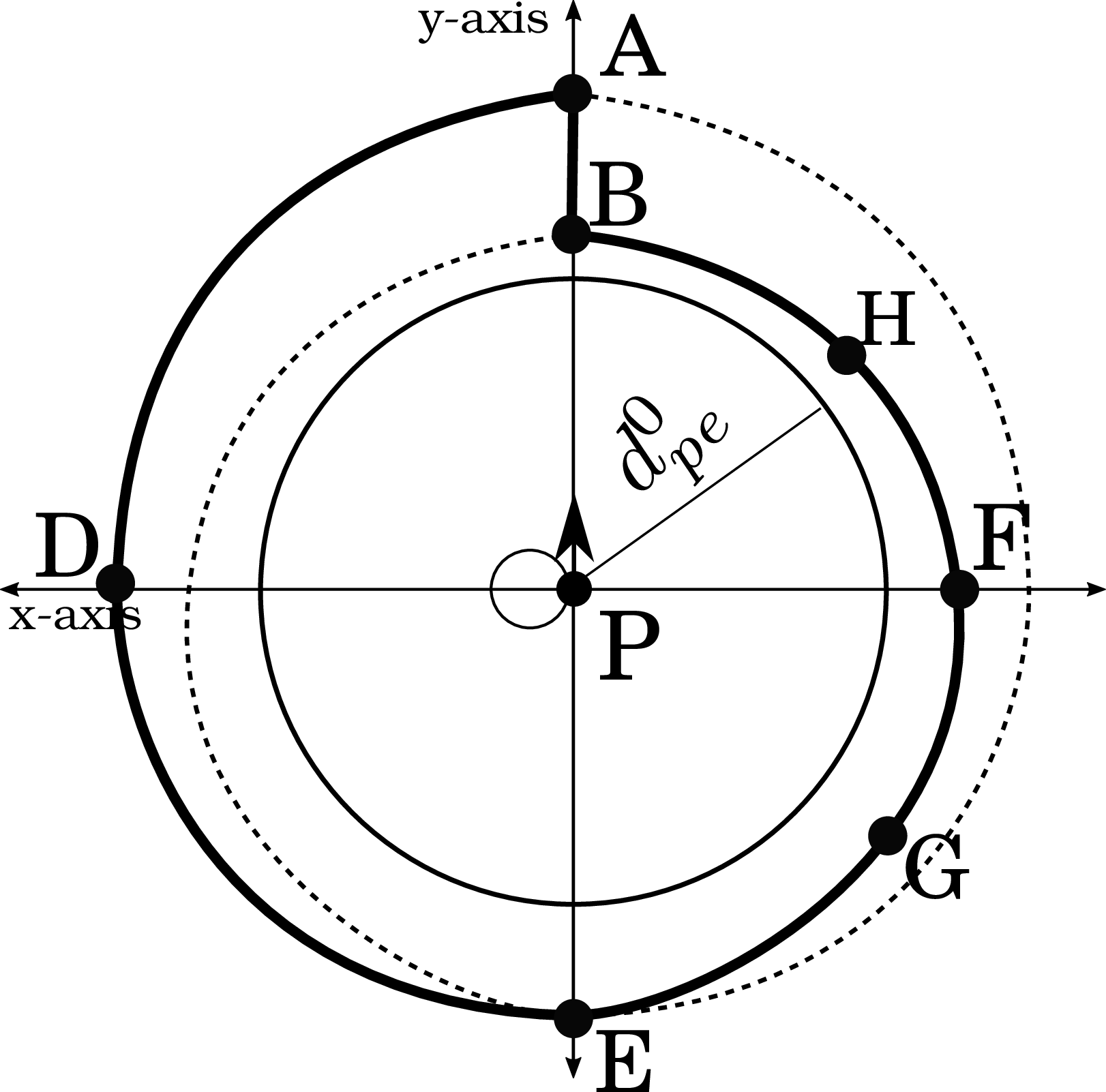}
\par\end{center}
\caption{\label{fig:left_reachable_active_set}Left reachable set as active
set}
\end{minipage}\hfill{}%
\begin{minipage}[t]{0.3\textwidth}%
\begin{center}
\includegraphics[scale=0.38]{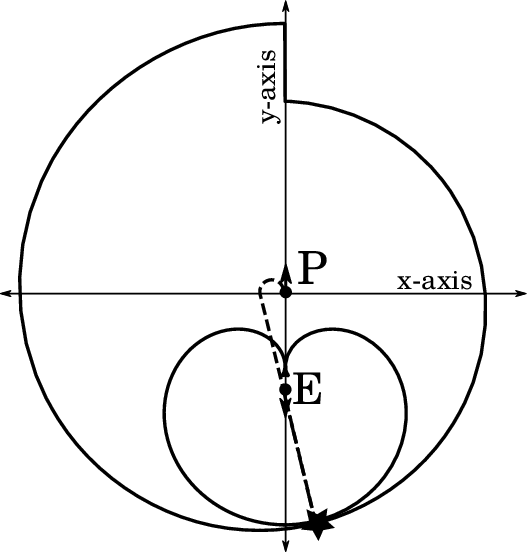}
\par\end{center}
\caption{\label{fig:left_capture_behind}Continuous containment by left reachable
set}
\end{minipage}\hfill{}%
\begin{minipage}[t]{0.35\textwidth}%
\begin{center}
\includegraphics[scale=0.35]{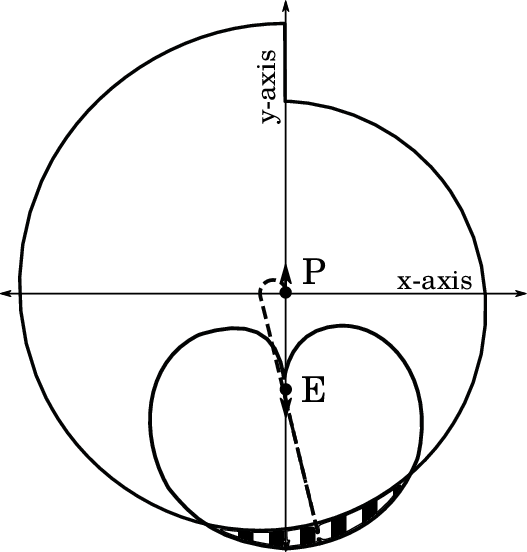}
\par\end{center}
\caption{\label{fig:left_safe_behind}Safe region: Continuous containment by
left reachable set}
\end{minipage}
\end{figure*}

The proof of Proposition \ref{prop:right_reachable_continuous_containment}
is similar to that of Proposition \ref{prop:left_reachable_continuous_containment}.
\begin{prop}
\label{prop:right_reachable_continuous_containment}Let $d_{pe}^{0}\geq2r_{e}+2\pi r_{e}(v_{p_{m}}/v_{e_{m}})$.
If $\mathcal{A}_{p}(T_{a})=R_{p}^{r_{t}}({\bf p_{0}},T_{a})$ and
$T_{a}=T_{r}^{c}<\infty$, then evader can always escape capture for
$t<T_{a}$ by using feedback strategy. 
\end{prop}
\begin{rem}
Consider a situation where the active set is the blocking set and
either the truncated left reachable set or the truncated right reachable
set contains the evader's reachable set continuously. In such a situation
either the truncated left reachable set or the truncated right reachable
set is also an active set along with the blocking set. This is established
in the next lemma.
\end{rem}
\begin{lem}
\label{lem:behind_time_same}Let $T_{lr}=\min(T_{l}^{c},T_{r}^{c})$
and $d_{pe}^{0}\geq2r_{e}+2\pi r_{e}(v_{p_{m}}/v_{e_{m}})$. If $T_{lr}<\infty$
and $B_{p}({\bf p_{0}},{\bf e_{0}},T_{b})$ is the active set such
that $R_{e}^{-}({\bf e_{0}},B_{p},T_{b})=\emptyset$ then either $R_{e}^{-}({\bf e_{0}},R_{p}^{l_{t}},T_{b})=\emptyset$
or $R_{e}^{-}({\bf p_{0}},R_{p}^{r_{t}},T_{b})=\emptyset$ and $T_{b}=T_{lr}$.
\end{lem}
\begin{IEEEproof}
Either the internal boundary of left reachable set or the internal
boundary of the right reachable set is part of the frontier boundary
of the blocking set. Without loss of generality assume that the internal
boundary of the right reachable set forms a part of the frontier boundary
of the blocking set as shown in Figure \ref{fig:blocking_evader_behind}.
Since the evader is behind the pursuer and $d_{pe}^{0}\geq2r_{e}+2\pi r_{e}(v_{p_{m}}/v_{e_{m}})$,
from Proposition \ref{prop:left_right_trunc_continuous_behind}, $BR^{1}$
and $BR^{2}$ are blocking curves. Thus the right reachable set will
contain evader's reachable set continuously. Hence, for the blocking
set, the capture will occur on $BCDFG$ w.r.t Figure \ref{fig:blocking_evader_behind}.
The $AB$ part of the frontier boundary is redundant. Thus the capture
by right reachable set and the blocking set will be at the same point
of the frontier boundary and the claim follows.
\end{IEEEproof}
\begin{note}Proposition \ref{prop:left_right_trunc_continuous_behind}
and Lemma \ref{lem:behind_time_same} allow us to consider $R_{p}^{l_{t}}({\bf p_{0}},T_{l}^{c})$
or $R_{p}^{r_{t}}({\bf p_{0}},$ $T_{r}^{c})$ as the active set when
the evader is behind the pursuer. Also it suffices to consider only
blocking set as be the active set when the evader is in the front
of pursuer. \end{note}
\begin{prop}
\label{prop:blocking_continuous_containment}Let $d_{pe}^{0}\geq2r_{e}+2\pi r_{e}(v_{p_{m}}/v_{e_{m}})$
and let the evader be in the front of the pursuer. If $\mathcal{A}_{p}(T_{a})=B_{p}({\bf p_{0}},{\bf e_{0}},T_{a})$,
then the evader can always escape capture for $t<T_{a}$ by using
feedback strategy.. 
\end{prop}
\begin{IEEEproof}
Consider the situation shown in Figure \ref{fig:blocking_front_extreme}
where the evader (denoted by point $E$) is in front of the pursuer.
Let ${\bf z}$ be the last point of $R_{e}({\bf e_{0}},T_{a})$ covered
by $B_{p}({\bf p_{0}},{\bf e_{0}},T_{a})$ at time $T_{a}$. The frontier
boundary of $B_{p}({\bf p_{0}},{\bf e_{0}},t)$ and the external boundary
$R_{e}({\bf e_{0}},t)$ both grow out radially with time $t$. Hence,
the point ${\bf z}$ will lie on the external boundary of $R_{e}({\bf e_{0}},T_{a})$.
Also, the point covering ${\bf z}$ will lie on the frontier boundary
of $B_{p}({\bf p_{0}},{\bf e_{0}},T_{a})$. From Figure \ref{fig:blocking_evader_behind}
it is clear that ${\bf z}\in ABCD$. Then at time $T_{a}'<T_{a}$
we have $R_{e}^{-}({\bf e_{0}},B_{p},T_{a}')\neq\emptyset$ i.e. the
pursuer's reachable set does not cover the evader's reachable set.
This situation is shown in Figure \ref{fig:blocking_evader_front_safe_region}.
Thus the evader can always escape capture at time $t<T_{a}$. 
\end{IEEEproof}
The next theorem characterizes the $CCSP$ and using it we determine
the saddle point strategies, point of capture and time of capture
in Theorem \ref{thm:cs-type}. Recall that $T^{*}$ is the $\min-\max$
time to capture under saddle-point strategies.
\begin{prop}
\label{thm:ApCCSP}$\mathcal{A}_{p}(T_{a})$ is a $CCSP$ and $T_{a}=T^{*}$.
\end{prop}
\begin{IEEEproof}
Proposition \ref{prop:left_reachable_continuous_containment}, Proposition
\ref{prop:right_reachable_continuous_containment} and Proposition
\ref{prop:blocking_continuous_containment} show that the active set
$\mathcal{A}_{p}(T_{a})$ is a $CCSP$. Thus from Theorem \ref{thm:complete_containment_continuous_necessary_sufficient}
we have $T^{*}=T_{a}$.
\end{IEEEproof}
\noindent \textbf{Proof of Theorem \ref{thm:capture_strategy_pursuer_evader}:
}Let $T^{*}$ and ${\bf z}$ be the time of capture and point of capture
respectively if the pursuer and evader use saddle point strategies.
From previous analysis, since the frontier boundaries grow out radially,
the point ${\bf z}$ belongs to the frontier boundary of the active
set $\mathcal{A}_{p}({\bf p_{0}},T^{*})$ and the external boundary
of the evader's reachable set. The time optimal trajectory for the
pursuer, which lies in $\mathcal{A}_{p}({\bf p_{0}},T^{*})$, to reach
point ${\bf z}$ is of the type $CS$. Thus, by Theorem \ref{thm:capture_strategy_pursuer_evader},
the saddle-point trajectory of the pursuer is of the type $CS$. Point
${\bf z}$ also lies on the external boundary of the evader's reachable
set. The evader trajectory which reaches ${\bf z}$ in minimum time
is also of the type $CS$. Thus evader's saddle-point trajectory is
of the type $CS$ (by Theorem \ref{thm:capture_strategy_pursuer_evader}).

\section{\label{sec:impt_theorem_3}Proof of Theorem \ref{thm:one_valid_tangent}}

One of the methods of computing the feedback pair $[\gamma_{p}^{*}({\bf x}),\gamma_{e}^{*}({\bf x})]$
is to use the Hamiltonian and Euler-Lagrange equations. We review
it here briefly. Consider a system beginning at $t=0$ given by the
equation 
\[
\dot{{\bf x}}(t)=f({\bf x}(t),{\bf u_{p}}(t),{\bf u_{e}}(t),t),\quad{\bf x}(0)={\bf x_{0}}
\]
where ${\bf x}:[0,t_{f}]\rightarrow\mathbb{R}^{n}$ is the state of
the system and ${\bf u_{p}}(t):[0,t_{f}]\rightarrow\mathbb{R}^{m_{p}}$
and ${\bf u_{e}}(t):[0,t_{f}]\rightarrow\mathbb{R}^{m_{e}}$ are the
inputs of the pursuer and evader respectively. Let the state constraints
at final time $t_{f}$ be given by 
\[
\psi({\bf x}(t_{f}))=0\in\mathbb{R}^{p}
\]
Further, define a generic performance criterion as
\[
J({\bf \gamma_{p}},{\bf \gamma_{e}})=\stackrel[0]{t_{f}}{\int}L({\bf x}(t),{\bf u_{p}}(t),{\bf u_{e}}(t))dt
\]
where  $L({\bf x}(t),{\bf u_{p}}(t),{\bf u_{e}}(t))$ is the recurring
cost. The pursuer tries to minimize the performance index while the
evader tries to maximize it using feedback strategies ${\bf u_{p}}={\bf \gamma_{p}}({\bf x})$
and ${\bf u_{e}}={\bf \gamma_{e}}({\bf x})$. The aim is to find,
if such a pair exists, $[{\bf \gamma_{p}^{*}}({\bf x}),{\bf \gamma_{e}^{*}}({\bf x})]$
such that

\[
J({\bf \gamma_{p}^{*}},{\bf \gamma_{e}^{*}})=\underset{\gamma_{e}}{\text{\ensuremath{\max}}}\,\,\underset{\gamma_{p}}{\text{\ensuremath{\min}}\,\,}J({\bf \gamma_{p}},{\bf \gamma_{e}})=\underset{\gamma_{p}}{\min}\,\,\underset{\gamma_{e}}{\text{\ensuremath{\max}}}\,\,J({\bf \gamma_{p}},{\bf \gamma_{e}})
\]
For a two player differential game described by the above equations,
let $[\gamma_{p}^{*}({\bf x}),\gamma_{e}^{*}({\bf x})]$ be the feedback
saddle-point solutions. Further, let ${\bf x^{*}}(t)$ denote the
corresponding state trajectory of the game and ${\bf u_{p}}(t)={\bf \gamma_{p}}({\bf x}(t))$,
${\bf u_{e}}(t)={\bf \gamma_{e}}({\bf x}(t))$ denote the open--loop
representations of the feedback strategies.
\begin{thm}
\label{thm:hamil}Minimum Principle \cite{bacsar1998dynamic},\cite{rufus},\cite{bryson1975applied}:
There exists a non-zero adjoint vector ${\bf {\bf \lambda}}(t):[0,t_{f}]\rightarrow\mathbb{R}^{n}$
satisfying following properties
\begin{eqnarray*}
\dot{{\bf x}}^{*} & = & f({\bf x}^{*}(t),{\bf u_{p}^{*}}(t),{\bf u_{e}^{*}}(t)),\quad{\bf x}(0)={\bf x_{0}}\\
\dot{{\bf \lambda}} & = & -\frac{\partial H}{\partial{\bf x}},\quad\quad{\bf \lambda}(t_{f})=\frac{\partial\Phi}{{\bf \partial x}}\big|_{t_{f}}
\end{eqnarray*}
\begin{equation}
\frac{\partial H}{\partial{\bf u_{p}}}=0,\,\frac{\partial H}{\partial{\bf u_{e}}}=0\quad or\quad H_{0}=\underset{{\bf u_{e}}}{\max}\,\,\underset{{\bf u_{p}}}{\min}H=\underset{{\bf u_{p}}}{\min}\,\underset{{\bf u_{e}}}{\max}H\label{eq:ham_controls}
\end{equation}
where $H=L+{\bf \lambda}^{\top}f$ is the system Hamiltonian, $\Phi(t_{f},{\bf x}(t_{f}))={\bf \mu}^{\top}\psi({\bf x}(t_{f}))$
and $\mu\in\mathbb{R}^{p}$ are undetermined multipliers at final
time.
\end{thm}
For the system given by (\ref{eq:pursuer_evader_dyn}) the Hamiltonian
is 
\begin{eqnarray}
H & = & 1+\lambda_{p_{x}}v_{p}\cos\theta_{p}+\lambda_{p_{y}}v_{p}\sin\theta_{p}+\lambda_{p_{\theta}}v_{p}w_{p}\label{eq:hamiltonian}\\
 &  & +\lambda_{e_{x}}v_{e}\cos\theta_{e}+\lambda_{e_{y}}v_{e}\sin\theta_{e}+\lambda_{e_{\theta}}v_{e}w_{e}\nonumber 
\end{eqnarray}
where $[\lambda_{p_{x}}\ \lambda_{p_{y}}\ \lambda_{p_{\theta}}]^{\top}$
denotes the adjoint vector corresponding to the pursuer. Also, let
$[\lambda_{e_{x}}\ \lambda_{e_{y}}\ \lambda_{e_{\theta}}]^{\top}$
denote the adjoint vector corresponding to the evader. Define, $\lambda_{p}=\sqrt{\lambda_{p_{x}}^{2}+\lambda_{p_{y}}^{2}}$,
$\lambda_{e}=\sqrt{\lambda_{e_{x}}^{2}+\lambda_{e_{y}}^{2}}$, $\phi_{p}=\tan^{-1}(\lambda_{p_{y}}/\lambda_{p_{x}})$
and $\phi_{e}=\tan^{-1}(\lambda_{e_{y}}/\lambda_{e_{x}})$. Thus we
can write (\ref{eq:hamiltonian}) as (see \cite{bui1994shortest}):
\begin{eqnarray}
H & = & 1+v_{p}\lambda_{p}\cos(\theta_{p}-\phi_{p})+\lambda_{p_{\theta}}v_{p}w_{p}\label{eq:hamiltonianconcise}\\
 &  & +v_{e}\lambda_{e}\cos(\theta_{e}-\phi_{e})+\lambda_{e_{\theta}}v_{e}w_{e}\nonumber 
\end{eqnarray}
By Theorem \ref{thm:guaranteedcapture} the capture time $T_{c}({\bf \gamma_{p}},{\bf \gamma_{e}})=T_{a}\leq\infty$.
Let $\mu=[\mu_{x}\,\,\mu_{y}]^{\top}$ be the undetermined constants
and $\psi({\bf x}(T_{a}))$ is defined by (\ref{eq:end_point_constatint}).
Thus, the complete final time constraint is $\Phi({\bf x}(T_{a}))=\mu^{\top}\psi({\bf x}(T_{a}))$
(see \cite{bryson1975applied}). The adjoint system for pursuer (evader)
is given as 
\begin{eqnarray}
\dot{\lambda}_{i_{x}}=0 &  & \dot{\lambda}_{i_{y}}=0\nonumber \\
\dot{\lambda}_{i_{\theta}} & = & -v_{i}\left[-\lambda_{i_{x}}\sin\theta_{i}+\lambda_{i_{y}}\cos\theta_{i}\right]\label{eq:adjoint_pursuer_evader}\\
 & = & v_{i}\lambda_{i}\sin(\theta_{i}-\phi_{i})\nonumber 
\end{eqnarray}
for $i\in\{p,e\}$.
\begin{lem}
Both the pursuer and evader use input policy ${\bf u}_{i}\in\mathcal{U}_{i}$
such that $v_{i}(t)=v_{i_{m}}\ \forall\ t\in[0,T_{a}]$.
\end{lem}
\begin{IEEEproof}
From Theorem \ref{thm:cs-type}, the capture occurs at the boundary
of the left reachable set or the right reachable set. From (\ref{eq:left_reachable_set})
it is seen that the boundary of the left and right reachable set is
characterized by the input policy ${\bf u}_{i}\in\mathcal{U}_{i}$
such that $v_{i}(t)=v_{i_{m}}\ \forall\ t$. 
\end{IEEEproof}
\begin{lem}
\label{lem:arcs_lines_purs}Any optimal path corresponding to the
saddle point strategy for the pursuer (evader) is the concatenation
of arcs of circles of radius $r_{p}$ ($r_{e}$) and line segments,
all parallel to some fixed direction $\phi_{p}$ ($\phi_{e}$). 
\end{lem}
\begin{IEEEproof}
The statement is derived using (\ref{eq:ham_controls}). The Hamiltonian
is affine in $w_{p}(t)$ and $w_{e}(t)$. If $\lambda_{p_{\theta}}(t)=0$
for all $t\in[t_{1},t_{2}]\subseteq[0,T_{a}]$ then from (\ref{eq:adjoint_pursuer_evader})
we must have $\dot{\lambda}_{p_{\theta}}(t)=0=v_{p}(t)\lambda_{p}(t)\sin(\theta_{p}(t)-\phi_{p})$.
Thus $\theta_{p}(t)=\phi_{p}$ or $\theta_{p}(t)=\phi_{p}+\pi$ for
all $t\in[t_{1},t_{2}]\subseteq[0,T_{a}]$ and the path is a line
segment with direction $\phi_{p}$. Thus $w_{p}(t)=0$ for all $t\in[t_{1},t_{2}]\subseteq[0,T_{a}]$.
If $|\lambda_{p_{\theta}}|>0$, this would imply that $w_{p}(t)=\pm w_{p_{m}}$
and the path would be an arc of circle $A_{p}(t)$ or $C_{p}(t)$.
Thus $H$ will be minimized with respect to $w_{p}(t)$ only if $w_{p}(t)=0$
or $w_{p}(t)=\pm w_{p_{m}}$. Similar arguments can be used to prove
the claim for the evader, where instead of minimizing $H$ we need
to maximize $H$ with respect to $w_{e}(t)$.
\end{IEEEproof}
\begin{lem}
\label{lem:same_line}The straight line paths that are followed by
both pursuer and evader are parallel to each other i.e. $\phi_{p}=\phi_{e}$.
\end{lem}
\begin{IEEEproof}
By Theorem \ref{thm:hamil}, $\lambda_{p_{x}}(T_{a})=\frac{\partial\Phi}{\partial x_{p}}=\mu_{x},\,\lambda_{p_{y}}(T_{a})=\frac{\partial\Phi}{\partial y_{p}}=\mu_{y},\,$
and $\phi_{p}(T_{a})=\tan^{-1}(\mu_{y}/\mu_{x})$. Similarly we can
show, $\lambda_{e_{x}}(T_{a})=\frac{\partial\Phi}{\partial x_{e}}=-\mu_{x},\,\lambda_{e_{y}}(T_{a})=\frac{\partial\Phi}{\partial y_{e}}=-\mu_{y},\,\phi_{e}(T_{a})=\tan^{-1}(\mu_{y}/\mu_{x})$.
Further, from (\ref{eq:adjoint_pursuer_evader}), $\lambda_{p_{x}},\lambda_{p_{y}},\lambda_{e_{x}},\lambda_{e_{y}}$
are constants. Thus we can conclude that $\phi_{p}(t)=\phi_{e}(t)=\tan^{-1}(\mu_{y}/\mu_{x})$
for all time $t$.
\end{IEEEproof}
Thus from the Hamiltonian formalism we have concluded that the trajectories
are either minimum turning radius circles or straight lines parallel
to some fixed line.

\noindent \textbf{Proof of Theorem \ref{thm:one_valid_tangent}:}
By Theorem \ref{thm:cs-type}, the strategies of both the pursuer
and the evader are of the type $CS$. Thus the capture takes place
along the straight line path. However, from Proposition \ref{lem:same_line}
it was shown that the straight lines are parallel. Thus if the capture
must occur along the straight lines, the lines must be coincident.
Since, both the pursuer and the evader travel along the circle and
straight line the paths must be the common tangents to circles in
one of the $PE$-pair.

\section{Conclusion}

In this work we have derived a characterization of time optimal pursuit
evasion game between two Dubins vehicles using continuous subsets
of reachable sets. Using this sets we show that the time optimal saddle
point strategies of both the pursuer and evader are circle and a straight
line. Using this characterization, we have derived feedback saddle
point strategies for the game of two cars. The solutions are computed
geometrically with insignificant computational effort. Available (computationally
intensive) algorithms for computing such feedback strategies cannot
be easily implemented in real time. On the other hand our proposed
algorithm can be run on onboard guidance computers typically available
on applications requiring time-optimal pursuit strategies, such as
missiles, fixed-wing aircraft, differential drive robots etc. Though
we have presented $\text{F}({\bf x}(t))$ as a continuous time state
feedback, any implementation of our algorithm will involve discretization
of time and evaluation of $\text{F}({\bf x}(t))$ at discrete intervals.
Due to the simplicity of the computation, the evaluation of $\text{F}({\bf x}(t))$
can be completed for small discretization intervals. A quantitative
study on the effect of discretization on the path optimality, however,
remains a topic of further research. Future work includes deriving
feedback laws for the case when the evader is very close to the pursuer
and for other configurations where capture is possible.

\section{Appendix}

In this section we give proof of Lemma \ref{lem:left_containment}.
In order to do this we define a kinematic point. A kinematic point
can turn instantaneously and move with velocity $v_{m}$. The equations
governing the motion of the kinematic point are:
\begin{eqnarray*}
\dot{x}(t) & = & v_{m}\cos(\theta(t))\\
\dot{y}(t) & = & v_{m}\sin(\theta(t))
\end{eqnarray*}
The movement of the kinematic point is also referred to as simple
motion in literature (see \cite{rufus} for details). The reachable
set of such a kinematic point, denoted by $R^{k}({\bf p_{0}},\bar{t})$
at time $\bar{t}$, is a circle of radius $v_{m}\bar{t}$ located
at ${\bf p_{0}}$. We consider two kinematic points: 
\begin{enumerate}
\item A kinematic point located at the same position as the pursuer (Dubins
vehicle) and having maximum velocity $v_{p_{m}}$. We refer to this
kinematic point as the kinematic pursuer. This kinematic pursuer stays
at point ${\bf p_{0}}$ for $t\in[0,\tilde{t}]$. It starts moving
at $t=\tilde{t}$. Thus at time $t>\tilde{t}$ its reachable set $R_{p}^{k}({\bf p_{0}},t)$
will be a circle of radius $v_{p_{m}}(t-\tilde{t})$ with center at
${\bf p_{0}}$. 
\item A second kinematic point located at the same position as the evader
(Dubins vehicle) and having maximum velocity $v_{e_{m}}$. We refer
to this kinematic point as the kinematic evader. The kinematic evader
starts moving at $t=0$. Thus its reachable set at time $t$, $R_{e}^{k}({\bf e_{0}},t)$,
will be a circle of radius $v_{e_{m}}t$.
\end{enumerate}
The following result is known \cite{cockayne1975plane}:
\begin{lem}
\label{lem:kinematic_evader_contains}The reachable set of any Dubins
vehicle moving with maximum velocity $v_{_{m}}$ is contained inside
the reachable set of kinematic point moving with velocity $v_{_{m}}$.
\end{lem}
\begin{lem}
\label{lem:kinematic_kinematic} For any initial pursuer position
${\bf p_{0}}\in\mathbb{R}^{3}$ and any initial evader position ${\bf e_{0}}\in\mathbb{R}^{3}$
and for all $\tilde{t}<\infty$ there exists $\bar{t}<\infty$ such
that the reachable set of the kinematic evader $R_{e}^{k}({\bf e_{0}},\bar{t})$
is contained inside the reachable set of the kinematic pursuer $R_{p}^{k}({\bf p_{0}},\bar{t})$
starting with a delay of $\tilde{t}$.
\end{lem}
\begin{IEEEproof}
Let the initial distance between the kinematic pursuer and the kinematic
evader be $d_{pe}^{0}$. At time $\tilde{t}$, point $A$ is the farthest
point from the pursuer in evader's reachable set and is located on
a straight line joining the evader and the pursuer as shown in Figure
\ref{fig:farthest_point}. Thus, for the lemma to hold there must
exist a time $t$ s.t. the reachable set of the pursuer has radius
greater than $d_{pe}^{0}+v_{e_{m}}t$ i.e. $v_{p_{m}}(t-\tilde{t})\geq d_{pe}^{0}+v_{e_{m}}t$.
This implies that $(v_{p_{m}}-v_{e_{m}})t\geq d_{pe}^{0}+v_{p_{m}}\tilde{t}$.
Now the right hand side of the equation is always a finite quantity
and since $v_{p_{m}}>v_{e_{m}}$ there exists a time $t=\bar{t}$
such that the inequality holds.
\end{IEEEproof}
\begin{figure*}
\centering{}%
\begin{minipage}[t]{0.45\textwidth}%
\begin{center}
\includegraphics[scale=0.2]{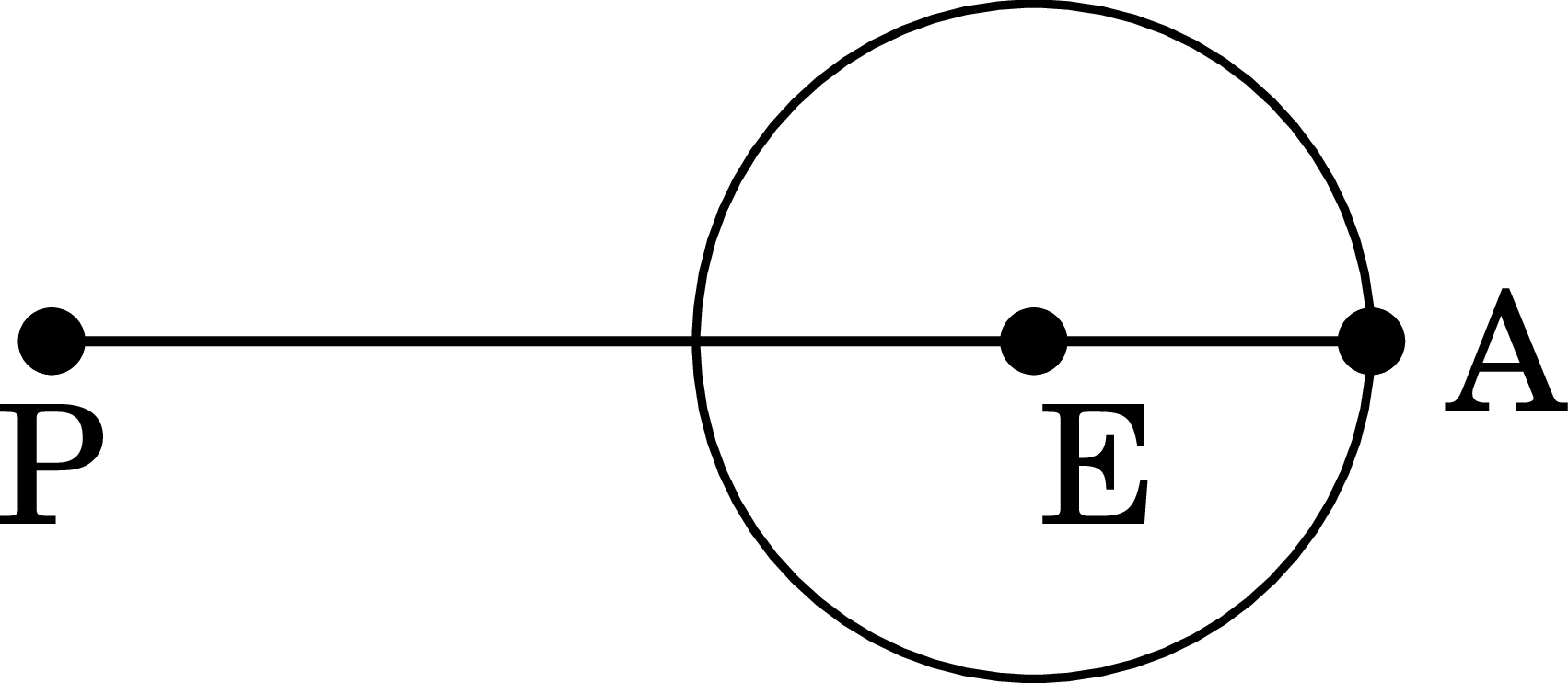}
\par\end{center}
\caption{\label{fig:farthest_point}Farthest point at time $\tilde{t}$}
\end{minipage}\hfill{}%
\begin{minipage}[t]{0.45\textwidth}%
\begin{center}
\includegraphics[scale=0.15]{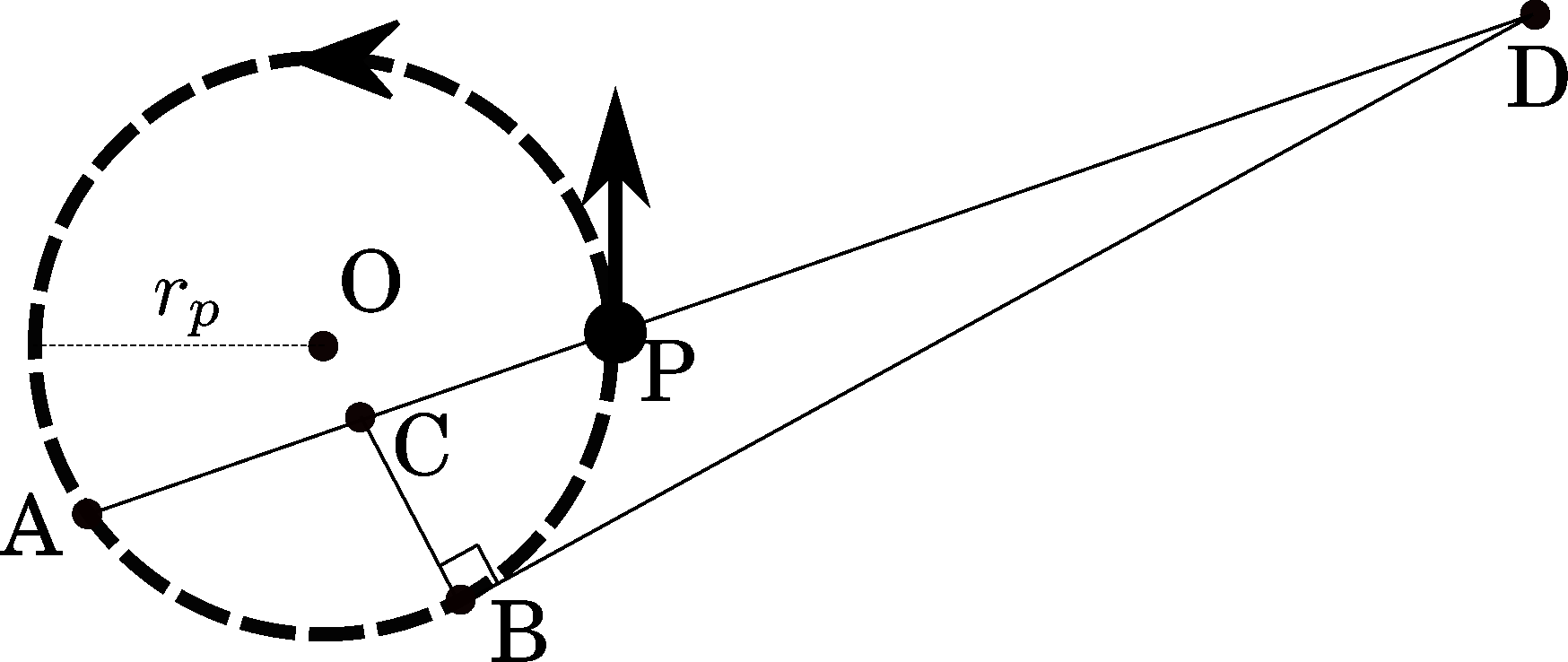}
\par\end{center}
\caption{\label{fig:left-coverage}Point reached by left reachable set after
some delay}
\end{minipage}
\end{figure*}

\begin{lem}
\label{lem:kinematic_pursuer_contained} Let $\bar{R}_{p}^{k}({\bf p_{0}},t):=R_{p}^{k}({\bf p_{0}},t)/A_{p}(0)$
i.e. the set of points contained in the reachable set of kinematic
pursuer except those which belong to the anti-clockwise pursuer circle.
For any initial pursuer position ${\bf p_{0}}\in\mathbb{R}^{3}$ and
any initial evader position ${\bf e_{0}}\in\mathbb{R}^{3}$the set
$\bar{R}_{p}^{k}({\bf p_{0}},t)$ of the kinematic pursuer, starting
to move at time $\tilde{t}=(2\pi+2)r_{p}/v_{p_{m}}$, is contained
inside the left reachable set of the pursuer having the kinematics
of a Dubins vehicle and starting at $t=0$ for all $t\geq\tilde{t}$.
\end{lem}
\begin{IEEEproof}
Let the pursuer (Dubins vehicle) be located at point P with orientation
as shown in Figure \ref{fig:left-coverage}. The dotted circle shown
in Figure \ref{fig:left-coverage} is the left minimum turning radius
circle of the pursuer. Also, let $D$ be a point located outside the
left pursuer circle as shown in Figure \ref{fig:left-coverage}. The
kinematic point is also located at $P$ and stays at point $P$ up
to time $\tilde{t}=(2\pi+2)r_{p}/v_{p_{m}}$. Thus, the minimum time
required for this kinematic pursuer to reach point $D$, say $t_{1}$,
is time it stays at point $P$ ( $t\in[0,\tilde{t}]$) plus time required
to cover distance $PD$. Thus, $t_{1}=\tilde{t}+\hat{t}_{1}$, where
$\hat{t}_{1}=\text{len}(PD)/v_{p_{m}}$ (since it is constrained to
stay at $P$ for time $t\in[0,\tilde{t}]$ ). 

Now the minimum time required by $CS$ type of curve given by (\ref{eq:left_reachable_set})
to reach point $D$ is given by 
\begin{eqnarray*}
t_{2} & = & \{\text{len(Arc \ensuremath{PAB})}+\text{len}(BD)\}/v_{p_{m}}
\end{eqnarray*}
Since $\text{len(Arc \ensuremath{PAB})}\leq2\pi r_{p}$, $\text{len}(BD)\leq\text{len}(AD)$,
and $\text{len}(AD)\leq2r_{p}+\text{len}(PD)$,
\begin{eqnarray*}
t_{2} & \leq & (2\pi r_{p}+\text{len}(AD))/v_{p_{m}}\\
 & \leq & (2\pi r_{p}+2r_{p}+\text{len}(PD))/v_{p_{m}}\\
 & = & \hat{t}_{1}+(2\pi+2)r_{p}/v_{p_{m}}=t_{1}
\end{eqnarray*}
Thus $t_{2}\leq t_{1}$ and the claim follows. 
\end{IEEEproof}
\begin{lem}
\label{lem:pursuer_circles_forbidden}If $d_{pe}^{0}\geq2r_{p}+2\pi r_{p}(v_{e_{m}}/v_{p_{m}})$
the pursuer can intercept the evader before it enters anti-clockwise
pursuer circle by the trajectories of the type $LS$.
\end{lem}
\begin{IEEEproof}
Let the pursuer be located at a distance $d_{pe}^{0}$ away from the
evader. By the arguments similar to those used in Lemma \ref{lem:evader_circles_forbidden}
we can conclude that the evader requires at least time $t_{e}:=[d_{pe}^{0}-2r_{p}]/v_{e_{m}}$
to reach any point on the pursuer's circle.

Any point on the left pursuer circle can be reached by the pursuer
in time $t\leq t_{p}:=2\pi r_{p}/v_{p_{m}}$. Hence if $t_{p}$ is
less than the minimum time in which the evader can reach any point
on the pursuer circles then the pursuer can intercept the evader before
it can enter the anti-clockwise pursuer circle. Thus, if 
\begin{eqnarray*}
t_{e} & \geq & t_{p}\\{}
[d_{pe}^{0}-2r_{p}]/v_{e_{m}} & \geq & 2\pi r_{p}/v_{p_{m}}\\
d_{pe}^{0} & \geq & 2r_{p}+2\pi r_{p}(v_{e_{m}}/v_{p_{m}})
\end{eqnarray*}
 the claim follows.
\end{IEEEproof}

\subsubsection*{Proof of Lemma \ref{lem:left_containment}}

By Lemma \ref{lem:pursuer_circles_forbidden}, even if the evader's
reachable set can extend into the anti-clockwise pursuer circle this
part of evader's reachable set cannot form a part of $R_{e}^{-}({\bf e_{0}},R_{p}^{l},T_{l})$
if $d_{pe}^{0}\geq2r_{p}+2\pi r_{p}(v_{e_{m}}/v_{p_{m}})$. Similarly,
from Lemma \ref{lem:kinematic_pursuer_contained} we have that the
set $\bar{R}_{p}^{k}({\bf p_{0}},t)$ of kinematic pursuer is contained
contained inside the reachable set of pursuer. Further, at time $\bar{t}$
the kinematic pursuer's reachable set contains the reachable set of
the kinematic evader (by Lemma \ref{lem:kinematic_kinematic}) and
hence the reachable set of the evader (by Lemma \ref{lem:kinematic_evader_contains}).
This implies that that the reachable set of the evader is contained
inside the left reachable set of the pursuer.

\bibliographystyle{IEEEtran}
\bibliography{twocars}

\vspace*{-1cm}

\begin{IEEEbiography}[{\includegraphics[width=1in,height=1.25in,clip,keepaspectratio]{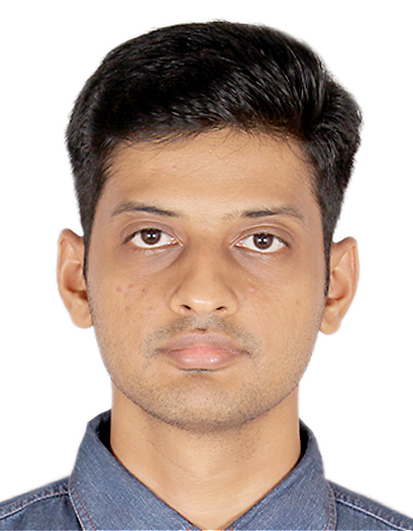}}]
{Aditya Chaudhrai} received the Bachelors degree in Electrical Engineering from University of Mumbai, India, in 2013 and M.Tech. degree from the Indian Institute of Technology Guwahati, India, in 2015. Currently, he is pursuing the Ph. D. degree in Control and Computing at the Indian Institute of Technology Bombay, Mumbai, India. His research interests include time optimal control, pursuit-evasion games, multi-agent systems and adaptive control. \end{IEEEbiography}


\begin{IEEEbiography}[{\includegraphics[width=1in,height=1.25in,clip,keepaspectratio]{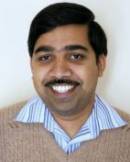}}]
{Debraj Chakraborty} received the B.E.E. degree from Jadavpur University, Kolkata, India, in 2001, the M.Tech. degree from the Indian Institute of Technology Kanpur, India, in 2003, and the Ph. D. degree from the University of Florida, Gainesville, in 2007. He joined the Indian Institute of Technology Bombay, Mumbai, India, in 2007, where he is currently an Associate Professor in the Department of Electrical Engineering. His research interests include optimal control, linear systems, and multi-agent systems. \end{IEEEbiography}
\end{document}